\documentclass[11pt]{article}

\usepackage[letterpaper,margin=1in]{geometry}
\usepackage{xcolor,graphicx,amssymb,amsthm,amsmath,hyperref,url,cite,mathrsfs,stmaryrd}
\counterwithin{figure}{section}
\usepackage{thm-restate}
\usepackage[utf8]{inputenc}
\usepackage[T1]{fontenc}

\newtheorem{theorem}{Theorem}[section]
\newtheorem{lemma}[theorem]{Lemma}
\newtheorem{corollary}[theorem]{Corollary}
\newtheorem{proposition}[theorem]{Proposition}

\newcommand{\set}[1]{\ensuremath{\{#1\}}}
\newcommand\emphdef[1]{{\textit{\textbf{#1}}}}
\newcommand{\cZ}{\mathbb{Z}}
\newcommand{\cR}{\mathbb{R}}
\newcommand{\cH}{\mathbb{H}}
\newcommand{\surface}{\mathcal{S}}
\newcommand{\patchsystem}{\Sigma}
\newcommand{\rtriangulation}{T}
\newcommand{\graph}{G}
\newcommand{\embeddedgraph}{H}
\newcommand{\loopgraph}{L}
\newcommand{\spanningtree}{Y}
\newcommand{\drawing}{\delta}
\newcommand{\poly}{\text{poly}}

\title{Untangling Graphs on Surfaces%
\thanks{All three authors are supported by the grant ANR-17-CE40-0033 (SoS), and the first and last authors are supported by the grant ANR-19-CE40-0014 (MIN-MAX), of the French National Research Agency ANR. The conference version of this paper was published in SODA 2024.}%
}
\author{%
Éric Colin de Verdière%
\thanks{LIGM, CNRS, Univ Gustave Eiffel, F-77454 Marne-la-Vallée, France}
\and
Vincent Despré%
\thanks{Université de Lorraine, CNRS, Inria, LORIA, F-54000 Nancy, France}
\and
Loïc Dubois
\thanks{LIGM, CNRS, Univ Gustave Eiffel, F-77454 Marne-la-Vallée, France}
}
\date{}

\begin{document}
\pagenumbering{arabic}
\maketitle

\begin{abstract}
  Consider a graph drawn on a surface (for example, the plane minus a finite set of obstacle points), possibly with crossings.  We provide an algorithm to decide whether such a drawing can be \emph{untangled}, namely, if one can slide the vertices and edges of the graph on the surface (avoiding the obstacles) to remove all crossings; in other words, whether the drawing is homotopic to an embedding.  While the problem boils down to planarity testing when the surface is the sphere or the disk (or equivalently the plane without any obstacle), the other cases have never been studied before, except when the input graph is a cycle, in an abundant literature in topology and more recently by Despré and Lazarus [SoCG 2017, J.\ ACM 2019], who gave a near-linear algorithm for this problem.

  Our algorithm runs in $O(m+\poly(g+b)n\log n)$ time, where $g\ge0$ and~$b\ge0$ are the genus and the number of boundary components of the input orientable surface~$\surface$, and $n$ is the size of the input graph drawing, lying on some fixed graph of size~$m$ cellularly embedded on~$\surface$.

  We use various techniques from two-dimensional computational topology and from the theory of hyperbolic surfaces.  Most notably, we introduce \emph{reducing triangulations}, a novel discrete analog of hyperbolic surfaces in the spirit of \emph{systems of quads} by Lazarus and Rivaud [FOCS~2012] and Erickson and Whittlesey [SODA 2013], which have the additional benefit that reduced paths are unique and stable upon reversal; they are likely of independent interest.  Tailored data structures are needed to achieve certain homotopy tests efficiently on these triangulations.  As a key subroutine, we rely on an algorithm to test the \emph{weak simplicity} of a graph drawn on a surface by Akitaya, Fulek, and T\'oth~[SODA 2018, TALG 2019].  
\end{abstract}

\section{Introduction}\label{sec:introduction}

In this paper, we study the following problem: Given a \emph{drawing} $\drawing$ of a graph $\graph$ on a (compact, connected, orientable) surface $\surface$, possibly with boundary, decide whether it is possible to \emph{untangle}~$\drawing$, that is, to make $\drawing$ crossing-free by a continuous motion; more formally, whether $\delta$ is homotopic to an embedding, where the homotopy may move vertices and edges. See Figure~\ref{fig:graph-plane}.

We first remark that, in the case where the surface is topologically trivial (the sphere or the disk), this problem boils down to planarity testing, and is thus solvable in linear time (see Hopcroft and Tarjan~\cite{ht-ept-74}).  Before stating our results in detail and presenting the main techniques, we survey related works, on curves and graphs on surfaces.

\begin{figure}[!h]
\centering
\includegraphics[scale=1.3]{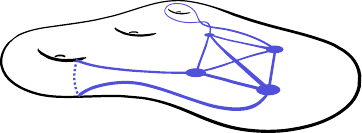}
\caption{A drawing of a graph that cannot be untangled, on an orientable surface of genus three without boundary.}\label{fig:graph-plane}
\end{figure}

\paragraph*{Related work on curves on surfaces.}

Topological questions on curves on surfaces have been an important field of study in the mathematical community for more than one hundred years. Dehn~\cite{d-tkzf-12} gave combinatorial characterizations of whether a closed curve on a surface is contractible (can be continuously moved to a point) or whether two such closed curves are freely homotopic (can be continuously deformed into each other).  Poincaré~\cite{p-ccal-04} provided similar characterizations to decide whether a closed curve on a surface is homotopic to a simple closed curve (in our language, can be untangled).  Although Dehn and Poincaré did not provide detailed analyses of their constructions, it is clear that they can be transformed into effective algorithms; see, e.g., Stillwell~\cite[Chapter~6]{s-ctcgt-93} for an overview.  

Generally, surfaces of genus at least two without boundary turn out to be the most difficult case.  For such a surface~$\surface$, a key insight, already present in these early works but reused in many other ones mentioned below, is the following: $\surface$ can be endowed with a hyperbolic metric (in the same way as the torus can be endowed with a Euclidean metric); under this metric, every closed curve is homotopic to a unique geodesic; moreover, any family of geodesic curves crosses minimally (has the least number of crossings among all curves in their respective homotopy classes).

The problem of determining whether a closed curve can be untangled has been extensively considered by the mathematical community in a long series of papers since the early 1960s; see Reinhart~\cite{r-ajccs-62}, Chillingworth~\cite{c-sccs-69,c-wns2-72}, and Birman and Series~\cite{bs-ascs-84}.  Later Cohen and Lustig~\cite{cl-pggin-87} and Lustig~\cite{l-pggin2-87} (see also Hass and Scott~\cite{hass1985intersections}) built upon these works to determine the geometric self-intersection number of a closed curve and the geometric intersection number of two closed curves (the number of crossings of homotopic curves crossing minimally).  De Graaf and Schrijver~\cite{gs-mcmcr-97} proved that it is possible to make curves cross minimally using homotopy (a.k.a.\ Reidemeister) moves that never increase the number of (self-)crossings.

These problems have been revisited under a more algorithmic lens by computational topologists since the 1990s.  Contractibility and homotopy can be tested in linear time, as proved by Lazarus and Rivaud~\cite{lr-hts-12} and Erickson and Whittlesey~\cite{ew-tcsr-13}, building upon earlier work by Dey and Guha~\cite{dg-tcs-99}.  One can decide whether a closed curve of length $n$, lying on a graph of size~$m$ itself embedded on~$\surface$, can be untangled in $O(m+n\log n)$ time, as proved by Despré and Lazarus~\cite{dl-cginc-19}, who also proved that the geometric self-intersection number can be computed in $O(m+n^2)$ time. Related works, by Chang and Erickson~\cite{ce-upc-17} and Chang and de Mesmay~\cite{cm-tcsma-22}, focus on computing the number of non-increasing homotopy moves needed to make a set of closed curves cross minimally.  Although all these algorithms are purely combinatorial, in many cases their proofs of correctness involve tools from hyperbolic geometry, suitably discretized, or at least rely on some intuition from hyperbolic geometry.

Related to deciding whether a curve can be untangled, the following problem has been studied recently: Decide whether a given closed curve drawn in a graph is \emph{weakly simple}, namely, whether it can be untangled in a neighborhood of the graph by an arbitrarily small perturbation.  Chang, Erickson, and Xu~\cite{cex-dwsp-15}, building upon earlier work by Cortese, Di Battista, Patrignani, and Pizzonia~\cite{cdpp-09}, provided a near-linear algorithm. 

\paragraph{Related work on graphs on surfaces.}

Many of the above questions can be extended from curves to graphs on surfaces.  Surprisingly, the literature studying them for graphs is rather scarce, in stark contrast with the central importance of graphs in theoretical computer science.

In the mathematical community, the only work that we are aware of, due to Ladegaillerie~\cite{l-cip1c-84}, is a reduction from the test of existence of an isotopy (a continuous family of embeddings) between two embeddings of a graph on a surface, to a test of isomorphism between combinatorial maps and an infinite number of homotopy tests between closed curves.  This characterization has been refined by Colin de Verdière and de Mesmay~\cite{cm-tgis-14}, leading to a polynomial time algorithm.

The well studied embeddability problem is to decide whether an input graph~$G$ can be embedded on an input surface~$\surface$.  Our problem adds some homotopy constraints to the embeddability problem.  The latter is NP-complete when $\surface$ is part of the input, as proved by Thomassen~\cite{thomassen1989graph}; however, Mohar~\cite{mohar1999linear} has given a linear time algorithm when $\surface$ is fixed, later simplified and improved by Kawarabayashi, Mohar, and Reed~\cite{kawarabayashi2008simpler}.  Our result indicates that adding these homotopy constraints makes the problem solvable in polynomial time, even if $\surface$ is part of the input.

There is a vast interest in crossing numbers for graphs drawn in the plane or in surfaces (for a survey, see Schaefer~\cite{s-gcnvs-13}), including constraints on the parity of the number of crossings between edges, related to the Hanani-Tutte theorem (see Schaefer~\cite{s-htrr-14}).  However, the constraint that we study, namely, fixing the homotopy of the drawing, appears to be entirely new.  A geometric variant has been studied by Goaoc, Kratochv\'\i{}l, Okamoto, Shin, Spillner, and Wolff~\cite{gkossw-upg-09}.  Given a straight-line drawing of a graph in the plane, they proved that the minimum number of vertices that have to be moved in order to untangle the drawing (the edges remaining drawn as line segments) is NP-hard to compute and to approximate; they also provided upper and lower bounds on the number of required moves.

More directly related to our work, in a recent paper, Atikaya, Fulek, and T\'oth~\cite{aft-rweg-19} extended the weak simplicity problem to graphs:  Given a drawing~$\drawing$ of a graph~$\graph$ in another graph~$\embeddedgraph$, itself embedded on a surface~$\surface$, is it possible to make $\delta$ simple with an arbitrarily small perturbation?  They solved this problem in near-linear time.  We will heavily rely on their algorithm, and postpone a detailed description to Section~\ref{sec:weak-embeddings}.  More recently, Fulek~\cite{f-egeg-20} gave a polynomial time algorithm to decide whether a graph drawn in the plane (without obstacles) can be perturbed by an arbitrarily small perturbation to turn it into an embedding in a pre-specified isotopy class.

\paragraph*{Our results.}

The input to our problem is a fixed cellular embedding$~\embeddedgraph$ of a graph on a surface$~\surface$, together with a drawing~$\drawing$ of a graph~$\graph$ on~$\embeddedgraph$, in the sense that $\drawing$ maps vertices of~$\graph$ to vertices of~$\embeddedgraph$ and edges of~$\graph$ to walks in~$\embeddedgraph$.  We obtain the following result:
\begin{theorem}\label{thm:surf}
  Let $\embeddedgraph$ be a graph of size~$m$ cellularly embedded on an orientable surface $\surface$ of genus $g \geq0$ with $b\ge0$ boundary components. Let $\graph$ be a graph and let $\drawing:\graph\to\embeddedgraph$ be a drawing of size~$n$. We can decide whether there is an embedding of~$\graph$ on~$\surface$ homotopic to~$\drawing$ in $O(m+(g+b)^2n\log((g+b)n))$ time. If so, we can construct in $O((g+b)^2(mn^2 + n \log((g+b)n)))$ time a weak embedding $\delta' : G \to H$, homotopic to $\delta$, such that $\delta'$ maps each edge of $G$ to a walk of length $O(g^2mn)$ in $H$ if $b=0$, and of length $O((g+b)mn)$ otherwise.
\end{theorem}
(Again, the cases of the sphere and the disk boil down to planarity testing so we omit these from now on.)
As in some previous papers in the area, e.g.,~\cite{ew-tcsr-13,lr-hts-12}, we use the RAM model, in which pointers and integers bounded by a polynomial in the input size can be manipulated in constant time (see Aho, Hopcroft, and Ullman.~\cite{ahu-daca-74} or Agarwal~\cite[Section~40.1]{a-rs-18}).  If the drawing~$\delta$ can be untangled, it is clear from our proofs that one can compute an embedding homotopic to~$\delta$ in polynomial time; details are omitted from this version.

We also consider an alternative framework in the case where $\surface$ is the plane minus a finite set of obstacle points, and $\drawing$ is a piecewise linear drawing of~$\graph$ avoiding the obstacles.  In this framework, we prove (here we use the real RAM model):

\begin{theorem}\label{thm:plane}
Let $P$ be a set of $p$ points of~$\cR^2$. Let $G$ be a graph, and let $\drawing:G\to\cR^2\setminus P$ be a piecewise linear drawing of size~$n$. In time $O(p^{5/2}n \log(pn))$, we can decide whether there is an embedding homotopic to $\drawing$ in $\mathbb R^2 \setminus P$. If so, we can construct in additional $O(p^5 n^2 \log(pn))$ time a piecewise-linear embedding homotopic to $\drawing$ in $\mathbb R^2 \setminus P$.
\end{theorem}

\paragraph*{A new tool: reducing triangulations.}

As is typical, the most difficult case is when the input surface~$\surface$ has genus at least two and no boundary, so we focus on this case.  For this purpose, we introduce the concept of \emph{reducing triangulation} of~$\surface$, a triangulation~$\rtriangulation$ of~$\surface$ with all vertex degrees at least eight and whose dual graph is bipartite.  Reducing triangulations form a new discrete analog of hyperbolic surfaces, and we define the notion of \emph{reduced walk} in~$\rtriangulation$, which behave similarly as geodesics in the continuous case.  A (possibly closed) walk~$W$ in~$\rtriangulation$ is reduced if no ``local'' reduction rule can be applied to~$W$.  We prove that $W$ is (possibly freely) homotopic to a unique (strongly) reduced walk, which we can compute in linear time, and that (strongly) reduced walks are stable upon reversal.  \emph{Systems of quads}, introduced by Lazarus and Rivaud~\cite{lr-hts-12}, refined by Erickson and Whittlesey~\cite{ew-tcsr-13}, and reused by Despré and Lazarus~\cite{dl-cginc-19}, have the same objective; in systems of quads, one can ensure either uniqueness in a given homotopy class or stability upon reversal, but (presumably) not both at the same time, which induces substantial technicalities in the previous papers~\cite{lr-hts-12,ew-tcsr-13,dl-cginc-19}.  Reducing triangulations are thus likely to be of independent interest.  Incidentally, the definition of reduced walks depends on the choice of an orientation of the surface, which is the main reason why we state our results for orientable surfaces only. 

Since the conference publication of this paper, reducing triangulations have been reused to give more efficient untangling algorithms for the specific case of multicurves~\cite{dubois2024making}, and to provide a discrete version of Tutte's barycentric theorem for graphs on surfaces~\cite{verdiere2025discrete}.

\paragraph*{Overview of the paper.}

After some preliminaries (Section~\ref{sec:preliminaries}), Sections \ref{sec:valid-triangulations} to~\ref{sec:contraction} are devoted to a proof of Theorem~\ref{thm:surf} in the case where $\embeddedgraph$ is a reducing triangulation on a surface of genus at least two without boundary; more precisely:

\begin{theorem}\label{thm:main-theorem}
  Let $\rtriangulation$ be a reducing triangulation of an orientable surface $\surface$ of genus~$g \geq 2$ without boundary. Let $\graph$ be a graph and let $\drawing:\graph\to\rtriangulation$ be a drawing of size~$n$. We can determine whether there is an embedding of~$\graph$ on~$\surface$ homotopic to~$\drawing$ in $O(gn\log(gn))$ time. If so, then we can construct in $O(n^2 + gn \log(gn))$ time a weak embedding $\delta' : G \to T$, homotopic to $\delta$, such that $\delta'$ maps each edge of $G$ to a walk of length $O(n)$ in $T$.
\end{theorem}

In detail, we introduce reducing triangulations, reduced walks and their properties, and reduction algorithms in Section~\ref{sec:valid-triangulations}. We then prove Theorem~\ref{thm:main-theorem} in a special case where $\graph$ is a (sparse) loop graph---each connected component of~$\graph$ contains a single vertex (Proposition~\ref{prop:untangling-a-loop-graph-general}, in Section~\ref{sec:untangle-loop-graph}).  The main idea, in the same spirit as Despré and Lazarus~\cite{dl-cginc-19}, is to use the fact that geodesics in hyperbolic surfaces cross minimally, and to prove that this remains true to some extent in our discrete analog.  In our case, roughly but not exactly, reducing all loops of a loop graph makes it weakly simple unless it cannot be untangled, so our algorithm eventually applies the weak simplicity test by Akitaya, Fulek, and T\'oth~\cite{aft-rweg-19}; but the proof of correctness involves putting a particular hyperbolic metric on the surface obtained by puncturing~$\surface$ and requires delicate arguments in the compactification of the universal cover of~$\surface$. Then, to solve the problem for an arbitrary graph~$\graph$, at a high level we contract a spanning forest of~$\graph$ and apply Proposition~\ref{prop:untangling-a-loop-graph-general} to the resulting loop graph, after removing the contractible loops and identifying the homotopic loops.  If the loop graph cannot be untangled, neither can the original graph.  Otherwise, it turns out that the graph can be untangled in~$\surface$ if and only if it is a weak embedding in (a neighborhood of) the loop graph, but proving this fact is rather subtle; we formalize this using the notion of \emph{factorization} (informally, the contraction of the graph, and the removals and identifications of loops in the resulting loop graph), and conclude the proof of Theorem~\ref{thm:main-theorem} by invoking once more the weak simplicity algorithm~\cite{aft-rweg-19}, all this in Section~\ref{sec:subdivision}.  A more efficient algorithm to compute a factorization, requiring tailored data structures to achieve certain homotopy tests efficiently, is deferred to Section~\ref{sec:contraction}.

After proving Theorem~\ref{thm:main-theorem}, we prove 
Theorem~\ref{thm:surf} for the case $g\ge2$, $b=0$ in Section~\ref{sec:untangling-closed-ggeq2}, essentially by converting our input embedded graph~$\embeddedgraph$ into a reducing triangulation.  In Section~\ref{sec:torus-and-boundary}, we prove our result for the remaining orientable surfaces, namely, the torus and the surfaces with non-empty boundary; while the overall strategy is the same, reducing triangulations are not needed anymore, which dramatically simplifies the algorithm.  Finally, in Section~\ref{sec:plane}, we prove our result on the plane with obstacles (Theorem~\ref{thm:plane}).

\section{Preliminaries}\label{sec:preliminaries}

\subsection{Graphs on surfaces, homotopies, and untangling}

In this paper, graphs are finite and undirected, but may have loops and parallel edges.  We use standard notions of graph theory, in particular the notion of \emphdef{walk} in a graph.  In a walk of length~$k$, $(v_0,e_0,v_1,e_1,\ldots,v_{k-1},e_{k-1},v_k)$, the first and last occurrences of vertices, $v_0$ and~$v_k$, are its \emphdef{endpoints}, and the other occurrences, $v_1,\ldots,v_{k-1}$, are the \emphdef{interior vertices}.  Vertices and edges may be repeated.  The two endpoints of a walk may coincide; a walk can even be reduced to a single vertex.  A \emphdef{closed walk} is similar to a walk, but the vertices are ordered cyclically instead of linearly; every vertex of a closed walk is an interior vertex.

We need some basic topology of surfaces~\cite{basic-topology,s-ctcgt-93}; we only provide the most basic definitions, sometimes only with an informal description.  All the surfaces we consider are connected and orientable, so we omit these adjectives in the sequel.  A compact \emphdef{surface}, possibly with boundary, is determined up to homeomorphism by its \emphdef{genus} (number of handles) and number of boundary components~\cite[Theorem~1.5]{basic-topology}.  We will occasionally consider \emphdef{non-compact surfaces}, which will be either obtained by removing finitely many points (\emphdef{punctures}) of a compact surface, or a universal cover, which are (usually) non-compact; see below.

A \emphdef{path} on a surface~$\surface$ is a continuous map $p:[0,1]\to\surface$; its \emphdef{endpoints} are $p(0)$ and~$p(1)$.  A \emphdef{loop} with \emphdef{basepoint} $b$ is a path whose both endpoints equal~$b$.  A path is \emphdef{simple} if it is one-to-one, except, of course, that $p(0)$ and~$p(1)$ coincide if $p$ is a loop.  A \emphdef{closed curve} is a continuous map $c:S^1\to\surface$, where $S^1=\cR/\cZ$ is the circle; it is \emphdef{simple} if it is one-to-one.  An \emphdef{arc} on a surface with boundary is a path that intersects the boundary precisely at its endpoints.  Paths and closed curves that differ only by their parameterizations are regarded as equal.

A \emphdef{drawing} of a graph~$\graph$ on a surface~$\surface$ maps every vertex of~$\graph$ to a point of~$\surface$, and every edge of~$\graph$ to a path with the appropriate endpoints.  An \emphdef{embedding} of~$\graph$ on~$\surface$ is a ``crossing-free'' drawing (vertices are mapped to pairwise distinct points, edges are mapped to simple, interior-disjoint paths, and the relative interior of each edge does not contain the image of a vertex). See Figure~\ref{fig:celldecomp}. The \emphdef{faces} of an embedding of~$\graph$ are the connected components of the complement of its image. A graph embedding is \emphdef{cellular} if every face is homeomorphic to an open disk, and a \emphdef{triangulation} if every face is homeomorphic to an open disk and is bounded by three sides of edges.  The \emphdef{rotation system} of an embedding of~$\graph$ is the cyclic ordering of the edges of~$\graph$ incident to each vertex in the embedding.  We use any of the numerous and standard data structures to represent combinatorial maps of cellular graph embeddings on orientable surfaces and move around quickly, e.g., the doubly-connected edge list, the halfedge data structure, or the gem representation~\cite{e-dgteg-03,k-ugpdd-99}.

\begin{figure}[!h]
\centering
\includegraphics[scale=1.3]{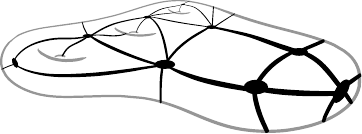}
\caption{A graph embedded on an orientable surface of genus three without boundary.}\label{fig:celldecomp}
\end{figure}

A \emphdef{homotopy} between two paths $p_0$ and~$p_1$ with the same endpoints is a continuous family of paths $(p_t)_{t\in[0,1]}$ with the same endpoints.  A \emphdef{(free) homotopy} between two closed curves $c_0$ and~$c_1$ is a continuous family of closed curves $(c_t)_{t\in[0,1]}$; this time, no point is required to be fixed.  A loop or closed curve is \emphdef{contractible} if it is homotopic to a constant loop or closed curve.  A \emphdef{homotopy} between two drawings of the same graph is a continuous family of  drawings between them; we emphasize that the vertices may move during the deformation.  A drawing of a graph on a surface can be \emphdef{untangled} if it is homotopic to an embedding.

Let $\embeddedgraph$ be a graph embedded in a surface $\surface$.  Consider also an abstract graph $\graph$.  A \emphdef{drawing}~$\varphi$ of $\graph$ in $\embeddedgraph$ is a drawing of~$\graph$ on~$\surface$ such that each edge of~$\graph$ is mapped to a walk in~$\embeddedgraph$ (possibly reduced to a single vertex).  The \emphdef{size} of $\varphi$ is the number of edges of $\graph$ plus the sum of the lengths of the walks $\varphi(e)$, for all edges $e$ of~$\graph$.

\subsection{Weak embeddings}\label{sec:weak-embeddings}

Let $\embeddedgraph$ be a graph embedded in a surface $\surface$, and let $\graph$ be an abstract graph.  A drawing $\varphi : \graph \to \embeddedgraph \subset \surface$ is a \emphdef{weak embedding} if there exist embeddings arbitrarily close to $\varphi$ (viewed as a drawing on~$\surface$), or equivalently but more formally if $\varphi$ is the limit of some sequence of embeddings $\psi : \graph \to \surface$ in the compact-open topology.  Akitaya, Fulek, and T\'oth~\cite{aft-rweg-19} provide an algorithm to decide whether such a $\varphi$ is a weak embedding.  We will heavily rely on this result, restated below, but we need some preparations.  As noted in~\cite[Section~1]{aft-rweg-19} the property for $\varphi$ to be a weak embedding does not depend on the precise embedding of~$\embeddedgraph$ on~$\surface$, but only on $\graph$, $\varphi$, and the abstract graph~$\embeddedgraph$ together with its rotation system.  Akitaya et al.\ formulate an alternative, more combinatorial definition of weak embeddings (in terms of \emph{strip systems}), and we present a trivially equivalent variation more suitable to our needs, the \emph{patch system}.

Intuitively, the patch system of a graph~$\embeddedgraph$ boils down to the dual graph of~$\embeddedgraph$, if $\embeddedgraph$ is cellularly embedded; but we cannot make this assumption in general.  Formally, the patch system of the graph~$\embeddedgraph$ (assumed here without loops or multiple edges, just for clarity of exposition, but this restriction can be dispensed of easily) equipped with its rotation system is defined as follows.  Consider an oriented closed disk~$D_v$ for each vertex $v$ of~$\embeddedgraph$; consider pairwise disjoint closed segments along the boundary of~$D_v$, one segment per edge incident to~$v$, ordered along the boundary of~$D_v$ as prescribed by the rotation system.  Now, for every edge~$uv$ of~$\embeddedgraph$, identify the corresponding segments of the disks $D_u$ and~$D_v$ in a way that respects their orientations.  These identifications result in a surface with boundary~$\patchsystem$, the \emphdef{patch system} of~$\embeddedgraph$.   The intersection of~$D_u$ and~$D_v$ is either a simple arc~$a(uv)$ in~$\patchsystem$, if $u$ and~$v$ are adjacent, or empty otherwise.  If $\embeddedgraph$ is embedded in a surface $\surface$ and inherits its rotation system from the embedding, then its patch system $\patchsystem$ can be thought of as a neighborhood of $\embeddedgraph$ in $\surface$, but its definition relies on combinatorial data only.

It follows from the considerations by Akitaya et al.~\cite[Section~1]{aft-rweg-19} that a map $\varphi : \graph \to \embeddedgraph \subset \surface$ is a weak embedding if and only if there exists an embedding $\psi: \graph \to \patchsystem$ that satisfies the following: (1) $\psi$ maps each vertex $v$ of $\graph$ inside~$D_v$; (2) $\psi$ maps each edge $e$ of~$\graph$, traversing edges $e_1,\ldots,e_k$ of~$\embeddedgraph$ in this order, to a path in~$\patchsystem$ crossing the arcs exactly in the order $a(e_1),\ldots,a(e_k)$.  We say that such an embedding $\psi : \graph \to \patchsystem$ \emphdef{approximates}~$\varphi$.  Finally, here is the result that we will use:
\begin{theorem}[Akitaya, Fulek, and Tóth~\cite{aft-rweg-19}]\label{thm:toth-et-al}
  Given an abstract graph~$\graph$, an embedded graph~$\embeddedgraph$, and a drawing $\varphi : \graph\to\embeddedgraph$ of size $n$, we can compute an embedding approximating~$\varphi$, or correctly report that no such embedding exists, in $O(n \log n)$ time.
\end{theorem}
(We note that the authors require that $\varphi$ be ``simplicial'', but we can assume this trivially by subdividing the graph~$\graph$.  The running time, in their theorem $O(m\log m)$ where $m$ is the number of edges, becomes $O(n \log n)$.)

\subsection{Covering spaces, lifts, and a property on homotopy}

Here we recall basic properties of covering spaces; see, e.g.,~\cite[Section~10.4]{basic-topology}.  A \emphdef{covering space} of a surface~$\surface$ is a (usually non-compact) surface~$\hat\surface$ equipped with a \emphdef{projection} $\pi:\hat\surface\to\surface$ that is a local homeomorphism.  The preimages of a point under~$\pi$ are its \emphdef{lifts}.  A \emphdef{lift} of a path~$p$ in~$\surface$ is a path $\hat p$ on~$\hat\surface$ such that $\pi\circ\hat p=p$.  Given any path $p$ and any lift~$\hat b$ of~$p(0)$, there is a unique \emphdef{lift} $\hat p$, a path in~$\tilde\surface$ starting at~$\hat b$.  A path~$p$ self-intersects if and only if either a lift of~$p$ self-intersects, or two lifts of~$p$ intersect.  One can similarly lift \emphdef{bi-infinite paths} $p:\cR\to\surface$, and even homotopies; see, e.g., ~\cite[10.11]{basic-topology}.

Closed curves have no basepoint, and thus typically lift to infinite paths in covering spaces.  Given a closed curve $c:S^1=\cR/\cZ\to\surface$, let $c'$ be the infinite path that ``wraps around'' $c$ infinitely many times, namely $c':\cR\to\surface$ is such that $c'(t)=c(t \text{ mod } 1)$.  A \emphdef{lift} of the closed curve~$c$ is, by definition, a lift of the (bi-infinite) path~$c'$.

We will mostly use universal covers:  The \emphdef{universal cover} $\tilde\surface$ is a covering space in which every loop is contractible.  A loop in~$\surface$ is contractible if and only if it lifts to a loop in~$\tilde\surface$.  If $\surface$ has positive genus and no boundary, then $\tilde\surface$ is homeomorphic to the open disk.

A closed curve is \emphdef{primitive} if it is not homotopic to the $k$th power of some other closed curve for some $k \geq 2$ (a closed curve iterated $k$ times).  For future reference, we quote (see Farb and Margalit~{\cite[Proposition~1.4]{fm-pmcg-12}, or Epstein~\cite[Theorem~4.2]{e-c2mi-66}):

\begin{lemma}\label{L:primitive}
  In an orientable surface possibly with boundary, every simple non-contractible closed curve is primitive.
\end{lemma}

\subsection{Hyperbolic surfaces and properties from Riemannian geometry}\label{sec:hyperbolic-surfaces}

We will use standard notions of hyperbolic geometry, described succinctly in Farb and Margalit~\cite[Chapter~1]{fm-pmcg-12}. See also Cannon, Floyd, Kenyon, and Parry~\cite{cfkp-hg-97} for a thorough introduction to the properties of the hyperbolic plane.

One model of the \emphdef{hyperbolic plane}~$\cH$, the Poincaré model, is the unit open disk endowed with a specific Riemannian metric of constant curvature $-1$ defined by $ds^2=4(dx^2+dy^2)/(1-x^2-y^2))^2$ (intuitively, it is easier to move around at the center of the disk than close to its boundary).  In this model, the maximal shortest paths are the subsets of circles (or, in the limit case, lines) touching the boundary of the disk orthogonally.

A \emphdef{hyperbolic surface} is a metric surface locally isometric to the hyperbolic plane.  A compact topological surface without boundary can be endowed with a hyperbolic metric if and only if its genus is at least two.  The hyperbolic metric on such a surface lifts to its universal cover, which becomes isometric to the hyperbolic plane~$\cH$.

We need a few definitions and properties from Riemannian geometry that we will use in the realm of hyperbolic surfaces:  A \emphdef{crossing} between two paths is a point where they intersect in their relative interiors, and actually cross in a transverse manner.  A \emphdef{geodesic} is a path that is locally shortest.  Given any tangent vector~$v$ at a point~$p$, there is a unique maximal geodesic with tangent vector~$v$ at~$p$.  As a consequence, if any two geodesics intersect at~$p$, they are either crossing at~$p$, or they are tangent, which implies that they \emphdef{overlap}: They are both part of the unique maximal geodesic passing through~$p$ with that common tangent vector.

Some non-compact hyperbolic surfaces can be constructed from \emphdef{ideal} hyperbolic polygons, the sides of which are geodesics of infinite length. Pair the sides of a collection of ideal hyperbolic polygons, and identify the two sides in each pair in a way that respects the orientations of the polygons. The result is a hyperbolic surface whose topological type is that of a surface obtained by \emphdef{puncturing} (removing points from) a compact surface without boundary, but punctures are relegated to infinity.

Both kinds of hyperbolic surfaces presented above enjoy the following specific properties: (1) There is a unique geodesic path homotopic to a given path; (2) there is a unique geodesic closed curve freely homotopic to a given closed curve, provided that curve is non-contractible and not homotopic to a neighborhood of a puncture.  In both cases, the geodesic is the unique shortest path or closed curve in its (free) homotopy class.

\section{Reducing triangulations and reduced walks}\label{sec:valid-triangulations}

In this section, we introduce \emph{reducing triangulations} and \emph{reduced walks} and their properties.   We start by defining these objects, then show the existence of a unique (strongly) reduced (closed) walk in any given homotopy class and any non-trivial free homotopy class, and finally show how to compute them in linear time.

\subsection{Reducing triangulations, turns, and reduced walks}\label{sec:trails}

A triangulation of a surface without boundary is \emphdef{6-reducing} if its dual is bipartite (each triangle is colored either red or blue, and adjacent triangles have different colors) and its vertices all have degree at least six. It is \emphdef{8-reducing} if its dual is bipartite and its vertices all have degree at least eight. \textbf{In this paper we use the term \emph{reducing triangulation} for 8-reducing triangulations.}  In this section and in Section~\ref{sec:reduced-walks-unicity}, however, for possible future use~\cite{verdiere2025discrete}, we state some results in the wider class of 6-reducing triangulations\footnote{In the conference version of this paper, we used the term \emph{reducing triangulation} for 8-reducing triangulations, like in this version, but without making the distinction between 6-reducing and 8-reducing, thus obtaining results on 8-reducing triangulations only.}. A straightforward application of Euler's formula shows that any reducing triangulation of any surface of genus $g \geq 2$ without boundary has size $O(g)$  (this is not the case for 6-reducing triangulations). It is easily seen that reducing triangulations with at most two vertices exist for every such surface (see Section~\ref{sec:untangling-closed-ggeq2}, and in particular Figure~\ref{fig:valid-tr-from-canonical-sc}).  We will use reducing triangulations for these surfaces, and also for their covering spaces; indeed, note that any reducing triangulation on a surface naturally lifts to a reducing triangulation in its covering spaces, by lifting the colors. 

Let $W$ be a walk (closed or not) in a 6-reducing triangulation~$\rtriangulation$.  Assume that a subwalk of~$W$ traverses the directed edge~$e$, arrives at vertex~$v$, and from there traverses the directed edge~$e'$.  Vertex~$v$ of~$W$ makes a \emphdef{$k$-turn}, $k\ge0$, if the walk of length two composed of edges $e$ and~$e'$, in this order, leaves exactly $k$~triangles of~$\rtriangulation$ to its left in the cyclic ordering around~$v$ between $e$ and~$e'$.  Similarly (and not exclusively), it makes a \emphdef{$-k$-turn}, $k\ge1$, if that walk leaves exactly $k$~triangles of~$\rtriangulation$ to its right.  When a turn is not a 0-turn, the notation is ambiguous because it can be represented either by a positive integer or by a negative integer (whose absolute values sum up to the degree of the vertex in~$\rtriangulation$). We always use an integer between $-3$ and~$3$ if possible.  For a degree-six vertex,  we prefer $3$ over~$-3$.  For turns that cannot be represented by any integer between $-3$ and~$3$, we choose the positive integer.

We sometimes need to be more precise and refine the notation: For $k\in\cZ$, a \emphdef{$k_b$-turn}, respectively a \emphdef{$k_r$-turn}, is a $k$-turn such that $e$ (with the previous notation) has a blue, respectively red, triangle to its left.  See Figures~\ref{fig:some-turns} and~\ref{fig:turns-in-context}.

\begin{figure}
\centering
\includegraphics[scale=1]{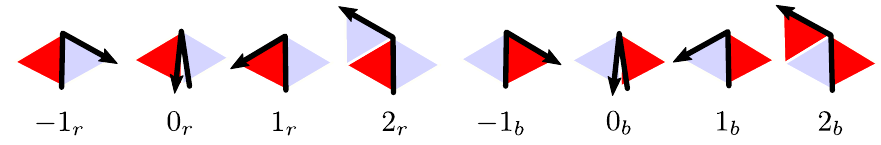}
\caption{Some of the turns that a walk can make in a 6-reducing triangulation.}\label{fig:some-turns}
\end{figure}

\begin{figure}
\centering
\includegraphics[scale=0.8]{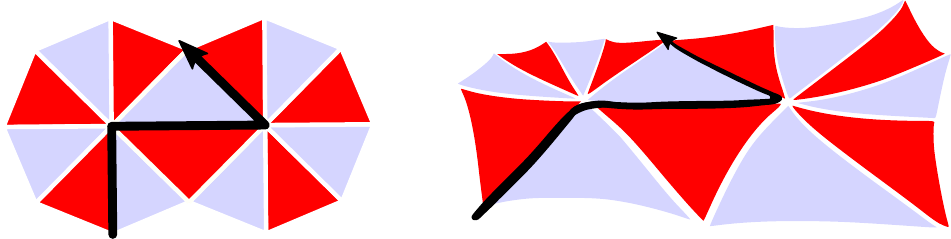}
\caption{In a 6-reducing triangulation, a portion of a walk that makes a $-2_r$-turn followed by a $1_b$-turn (both are bad turns).}\label{fig:turns-in-context}
\end{figure}

We call \emphdef{bad turn} any 0-turn, 1-turn, $-1$-turn, $2_r$-turn or $-2_r$-turn. A walk in $\rtriangulation$ is \emphdef{reduced} if none of its interior vertices makes a bad turn.  A closed walk is \emphdef{reduced} if none of its vertices makes a bad turn, and it is \emphdef{strongly} reduced if moreover not all of them make a $3_r$-turn, and not all of them make a $-3_b$-turn\footnote{Among the reduced closed walks, we will mostly consider those which are \emph{strongly} reduced. In fact, in the conference version of this paper, we did not make the distinction: we considered only strongly reduced closed walks, which we simply called reduced closed walks.}.  Intuitively, reduced walks are discrete geodesics in the triangulation~$\rtriangulation$ where all triangles are equilateral:  A reduced walk leaves an angle at least~$\pi$ on both sides at each interior vertex, except when the vertex makes a $2_b$-turn or a $-2_b$-turn, which corresponds to an angle of~$2\pi/3$ on one side.  The bipartiteness of the triangulation then ``breaks ties'' for determining the geodesic.  We emphasize that the notion of reduced walks requires an orientation of the surface.

Here are a few immediate but crucial properties that we will use repeatedly.  Any subwalk of a reduced walk (closed or not) is also reduced.  The reversal of any reduced walk (closed or not) is also reduced, because reversing a walk exchanges $2_r$-turns with $-2_r$-turns, and $3_r$-turns with $-3_b$-turns. The same properties hold for strongly reduced closed walks.

The \emphdef{turn sequence} of a (closed) walk is the list of turns made by the walk at its interior vertices; this is a cyclic sequence if the walk is closed.  We use some straightforward notations for turn sequences:  Exponents denote iterations, stars denote arbitrary nonnegative integers, and vertical bars denote ``or''.  For example, $23^*4$ denotes a 2 followed by a nonnegative number of~3s, followed by a~4.  $(23^*4)^*$ denotes a nonnegative number of concatenations of patterns of that form.  $(2|4)$ denotes either a 2 or a~4.

\subsection{Uniqueness of reduced walks}\label{sec:reduced-walks-unicity}

In this section, and in this section only, we consider the wider class of 6-reducing triangulations. We prove:
\begin{proposition}\label{prop:trail-homotopy}
In a 6-reducing triangulation~$\rtriangulation$, any two homotopic reduced walks are equal.
\end{proposition}

Intuitively, reduced walks are geodesics, and the proof of Proposition~\ref{prop:trail-homotopy} goes by showing that any two distinct homotopic reduced walks would form a monogon or a bigon (a disk bounded by one or two subwalks, in the universal cover), and that this is impossible.  It relies on the following lemmas.

\begin{lemma}\label{lem:use-bipartition}
  In a 6-reducing triangulation~$\rtriangulation$, any walk whose turn sequence is of the form $23^*2$ contains a $2_r$-turn.   Any walk in~$\rtriangulation$ whose turn sequence is of the form $23^*4$ or $43^*2$ contains a $2_r$-turn or a $4_r$-turn.
\end{lemma}
\begin{proof}
  By bipartiteness of the 6-reducing triangulation~$\rtriangulation$, the color of the triangle of $\rtriangulation$ located to the left of the walk changes when passing over an interior vertex that makes a 2-turn or 4-turn, and not when the interior vertex makes a 3-turn.
\end{proof}

\begin{lemma}\label{lem:gauss-bonnet}
In a 6-reducing triangulation~$\rtriangulation$, let $C$ be a simple closed walk that bounds a (non-empty) disk to its left. For every $k\ge1$, let $m_k$ be the number of vertices of~$C$ at which $C$ has exactly $k$~triangles to its left. Then $2 m_1 + m_2 \geq 6 + \sum_{k\ge4}m_k$.
\end{lemma}
The proof can be viewed as a consequence of a discrete version of the Gauss--Bonnet theorem~\cite[Section 2.3]{ew-tcsr-13}:
\begin{proof}
Let $D$ be the restriction of the 6-reducing triangulation to the closed disk bounded by $C$. Consider the following discharging argument. Give initial weight $6$ to each vertex and each triangle of $D$, and weight $-6$ to each edge of $D$. By Euler's formula, the weights initially attributed sum up to~$6$. Discharge as follows. For each incidence between a vertex $v$ and an edge $e$, transfer $3$ from $v$ to $e$. For each incidence between a vertex $v$ and a triangle~$t$, transfer $2$ from $t$ to $v$.  Now, edges and triangles all have weight $0$, while every vertex~$v$ has weight $\kappa(v) := 6 - 3 \deg(v) + 2 \deg'(v)$ where $\deg(v)$ and $\deg'(v)$ denote respectively the number of edge incidences and triangle incidences of $v$ in~$D$. We proved $6 = \sum_v\kappa(v)$.

Every vertex $v$ that lies in the interior of $D$ (not on $C$ itself) satisfies $\deg(v) = \deg'(v)$ and thus $\kappa(v) = 6 - \deg(v)\le0$, since each vertex of~$\rtriangulation$ has degree at least six.  Every vertex~$v$ that lies on $C$ satisfies $\deg(v) = \deg'(v) + 1$ and thus $\kappa(v) = 3 - \deg'(v)$.  Thus, we have $6=\sum_v\kappa(v)\le\sum_{k\ge1}(3-k)m_k$, implying the result.
\end{proof}

\begin{lemma}\label{L:bad left turns}
In a 6-reducing triangulation~$\rtriangulation$, let $C$ be a simple closed walk that bounds a (non-empty) disk to its left. There are at least three vertices at which $C$ makes a $1$-turn or a $2_r$-turn.
\end{lemma}

\begin{proof}
Let $S$ be the set of vertices of~$C$ that make a $1$-turn or a $2_r$-turn. By contradiction, assume $|S|\le2$.   Using the notations of Lemma~\ref{lem:gauss-bonnet}, we have $m_1\le|S|\le2$; indeed, any vertex that makes a 1-turn belongs to~$S$. We consider the subwalks with turn sequence of the form $2(1|3)^*2$. These subwalks may only share their first and last edges and are otherwise disjoint. Lemma~\ref{lem:gauss-bonnet} implies that $m_2\ge(6-2m_1)+\sum_{k\ge4}m_k$; thus, there are at least $6-2m_1$ such sequences.  At most $|S|-m_1$ vertices in~$S$ make a 2-turn, and thus at most $2(|S|-m_1)$ such sequences start or end in~$S$, because the sequences all start and end with a vertex that makes a 2-turn.  There remains at least $(6-2m_1)-2(|S|-m_1)=6-2|S|\ge1$ such sequences whose first and last elements are not in~$S$.  In such a sequence, exactly one of the two vertices that make a 2-turn actually makes a $2_r$-turn by Lemma~\ref{lem:use-bipartition}, which is impossible because the vertices of~$C$ that make a $2_r$-turn are all in~$S$.
\end{proof}

\begin{proof}[Proof of Proposition~\ref{prop:trail-homotopy}]
The 6-reducing triangulation $\rtriangulation$ lifts to an infinite 6-reducing triangulation $\tilde \rtriangulation$ in its universal cover.  Assume that there exist two distinct homotopic reduced walks in~$\rtriangulation$ (possibly one of them being a single vertex). Let $W_1$ and~$W_2$ be lifts of these walks in~$\tilde\rtriangulation$, with the same endpoints; they are also reduced.  We show below that this implies that there exists, in $\tilde\rtriangulation$, a simple closed walk~$C$ with at most two vertices that make a bad turn. This contradicts Lemma~\ref{L:bad left turns}.

Indeed, if one of $W_1$ and~$W_2$ is not simple, it contains a non-empty subwalk with the same starting and ending vertex~$v$, and otherwise not repeating any vertices; since subwalks of reduced walks are also reduced, the claim holds (only at $v$ the walk can make a bad turn).  Otherwise, $W_1$ and~$W_2$ are simple, distinct, and have the same first and last vertices, which implies that they admit non-empty subwalks that have the same endpoints and are otherwise disjoint.  (Indeed: up to exchanging $W_1$ and~$W_2$, some edge in $W_1$ does not belong to~$W_2$.  Consider the subwalk of~$W_1$ of minimum length containing that edge and intersecting~$W_2$ at its endpoints.  The subwalks of~$W_1$ and~$W_2$ with these endpoints have the desired property.)  Let $C$ be the simple closed walk that is the concatenation of these subwalks (or their reversals) $W'_1$ and~$W'_2$.  The subwalks $W'_1$ and~$W'_2$ are also reduced, since subwalks and reversals of reduced walks are reduced. Hence, $C$ has at most two vertices that make a bad turn, which concludes as desired.
\end{proof}

\subsection{Uniqueness of strongly reduced closed walks}

From now on, we consider only reducing triangulations, for which all vertices have degree at least eight.

\begin{proposition}\label{prop:uniqueness-reduced-closed-walks}
   In a reducing triangulation~$\rtriangulation$, any two freely homotopic, non-contractible, strongly reduced closed walks are equal.
\end{proposition}

The proof of Proposition~\ref{prop:uniqueness-reduced-closed-walks} relies on the following lemma.
\begin{lemma}\label{lemma:reduced-closed-walks}
 Consider two disjoint simple closed walks $C$ and $C'$ that bound an annulus in~$\rtriangulation$. Then $C$ and $C'$ are not both strongly reduced.
\end{lemma}
\begin{proof}
We prove the claim by contradiction so assume that $C$ and $C'$ are both strongly reduced. Let $A$ be the restriction of the reducing triangulation to the closed annulus bounded by $C$ and $C'$. Let $i$ be the number of vertices in the interior of $A$, and let $B$ be the vertex set of the boundary of~$A$. For every vertex $v$ of $A$ let $\deg(v)$ and $\deg'(v)$ be the number of incidences of $v$ with respectively the edges and the triangles of $A$.

We first claim that $2i \leq \sum_{v \in B} (3 - \deg'(v))$. We prove this claim with the same discharging rules as in the proof of Lemma~\ref{lem:gauss-bonnet}. Give initial weight $6$ to each vertex and each triangle of $A$, and weight $-6$ to each edge of $A$. By Euler's formula, the weights initially attributed sum up to $0$. Discharge as follows. For each incidence between a vertex $v$ and an edge $e$, transfer $3$ from $v$ to~$e$. For each incidence between a vertex $v$ and a triangle~$t$, transfer $2$ from $t$ to~$v$. Now, edges and triangles all have weight $0$, while each vertex $v$ has weight $\kappa(v) := 6 - 3 \deg(v) + 2 \deg'(v)$. We proved $0 = \sum_v\kappa(v)$. Every vertex $v$ of $A$ satisfies $\kappa(v) = 3 - \deg'(v)$ if $v \in B$ and $\kappa(v)=6-\deg(v)\leq -2$ otherwise. Therefore, $0 \leq \sum_{v\in B}(3-\deg'(v))-2i$, proving the claim.

We orient $C$ and~$C'$ so that $A$ lies to their left.  Let $C_1\in\set{C,C'}$, and let $B_1\subset B$ be the set of vertices of~$C_1$.  We remark that, because $C_1$ has no bad turn, every $v \in B_1$ satisfies $\deg'(v) \geq 2$.  Moreover, because $C_1$ is strongly reduced, by Lemma~\ref{lem:use-bipartition}, its turn sequence does not contain $23^*2$ as a subword, and it is not of the form $23^*$.  Thus, $\sum_{v \in B_1} (3 - \deg'(v)) \leq 0$, with equality if and only if the turn sequence of $C_1$ has the form $(23^*43^*)^*$ or $3^*$.

From the above claim and the conclusion of the previous paragraph, we deduce that $A$ has no interior vertices and that each of the turn sequences of $C$ and~$C'$ is of the form $(23^*43^*)^*$ or $3^*$, which must actually be $(2_b3^*4_r3^*)^*$ (by Lemma~\ref{lem:use-bipartition} and because there is no $2_r$) or $3_b^*$ (by definition of a strongly reduced closed walk).

\begin{figure}
\centering
\includegraphics[scale=1.3]{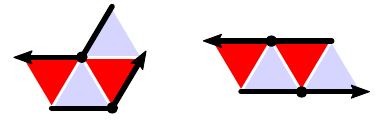}
\caption{The cases leading to a contradiction at the end of the proof of Lemma~\ref{lemma:reduced-closed-walks}.}\label{fig:unique-closed}
\end{figure}

Assume that the turn sequence of one boundary component of~$A$ has the form $(2_b3^*4_r3^*)^*$.  Let $v$ be a vertex of~$B$ that makes a $2_b$-turn. Then, the vertex across~$v$ in the other boundary component of~$A$ makes a $4_b$-turn (Figure~\ref{fig:unique-closed}, left), which is a contradiction since no $4_b$ appears in the allowed forms above.  So the turn sequences of $C$ and~$C'$ are both of the form $3_b^*$; however, for a similar reason as above, if $v$ is a vertex of~$B$ that makes a $3_b$-turn, the vertices across~$v$ make a $3_r$-turn (Figure~\ref{fig:unique-closed}, right), a contradiction.
\end{proof}

\begin{proof}[Proof of Proposition~\ref{prop:uniqueness-reduced-closed-walks}]
  Consider a reducing triangulation $\rtriangulation$ of a surface $\surface$ and two freely homotopic strongly reduced closed walks $C$ and $C'$ in $\rtriangulation$.

  We consider the \emph{annular cover}~$\hat\surface$ of~$\surface$ defined by~$C$~\cite[Section~1.1]{ce-tnpcs-10}.  This is a covering space homeomorphic to an open annulus, in which $C$ lifts to a closed curve~$\hat C$, and every simple closed curve is either contractible or homotopic to~$\hat C$ or its reverse.  Lifting the homotopy from~$C$ to~$C'$ yields a lift~$\hat C'$ that is homotopic to~$\hat C$.  Both $\hat C$ and~$\hat C'$ are strongly reduced closed walks in~$\hat\rtriangulation$, the lift of the reducing triangulation~$\rtriangulation$.

  We claim that $\hat C$ (and, for the same reasons, also~$\hat C'$) is simple.  Indeed, lifting~$\hat C$ to the universal cover~$\tilde\surface$ of~$\surface$ yields a lift, which is a bi-infinite path~$\tilde P$, and is actually reduced in the reducing triangulation~$\tilde\rtriangulation$ obtained from lifting~$\rtriangulation$.  By Proposition~\ref{prop:trail-homotopy}, $\tilde P$ is simple.  Since it is the \emph{only} lift of~$\hat C$, this implies that $\hat C$ is simple.
  
  Now there are two cases. If $\hat C$ and $\hat C'$ are disjoint, Lemma~\ref{lemma:reduced-closed-walks} implies a contradiction.  If they intersect, at vertex~$v$ say, then we transform $\hat C$ and~$\hat C'$ into loops with starting and ending vertex~$v$; these two loops have the same turning number, and are thus homotopic reduced walks, and equal by Proposition~\ref{prop:trail-homotopy}. Thus $\hat C$ and $\hat C'$ are equal, and so are $C$ and~$C'$ by projection.
\end{proof}

\subsection{Reducing a walk}\label{sec:reduced-walks-reduction}

\begin{proposition}\label{prop:trail-algo}
 Given a walk~$W$ of length~$n$ in a reducing triangulation~$\rtriangulation$, we can compute a reduced walk homotopic to~$W$ in $O(n)$ time.
\end{proposition}
The techniques of the proof of Proposition~\ref{prop:trail-algo} will be reused in Section~\ref{sec:contract-datastruct}.  In particular, we need a few data structures here, which will be extended later.

Our reducing triangulation~$\rtriangulation$ is stored as an embedded graph, but we need additional information.  Specifically, for each vertex~$v$ of~$\rtriangulation$, we number the directed edges with source~$v$ in clockwise order (starting at an arbitrary directed edge).  This number is stored on each directed edge.  Moreover, vertex~$v$ contains an array $A_v$ of pointers such that $A_v[i]$ is the directed edge number~$i$ with source~$v$.  This array has size the degree of~$v$, which is also stored in~$v$.  We note that this additional data can be computed in time linear in the size of the reducing triangulation, so we assume that the reducing triangulation incorporates this information.  With this data structure, we can perform the following operations in constant time: (1) given two directed edges $uv$ and~$vw$, compute the integer~$i$ such that $uvw$ forms an $i$-turn; (2) Given a directed edge~$uv$ and an integer~$i$, compute the directed edge~$vw$ such that the walk $uvw$ forms an $i$-turn.  Again, this can be done as a preprocessing step in time linear in the size of the reducing triangulation.

We also need a data structure, called \emphdef{compressed homotopy sequence}, that represents walks in a compact form, in the same spirit as Erickson and Whittlesey~\cite[Section~4.1]{ew-tcsr-13} use run-length encoding to encode turn sequences.  An \emphdef{elementary subwalk} of~$W$ is an inclusionwise maximal subwalk of~$W$ whose turn sequence has the form $3^k$ or $(-3)^k$ for some $k\ge1$, or is a single symbol $a\not\in\set{-3,3}$.  The elementary subwalks of~$W$ are naturally ordered along~$W$ and cover~$W$, with overlaps because the last edge of an elementary subwalk appears as the first edge of the following one.  We represent~$W$ by storing the elementary subwalks of~$W$ in a doubly linked list in order along~$W$; each elementary subwalk is described by the images in~$\rtriangulation$ of its first and last directed edges, as well as its turn sequence in compressed form, namely, whether it is of the form $3^k$, $(-3)^k$, or $a$, and the integer $k$ or the symbol~$a$.  (Walks~$W$ of length zero or one cannot be encoded in this data structure, but can be handled separately in a trivial manner.)

\begin{figure}
\centering
\includegraphics[scale=1.5]{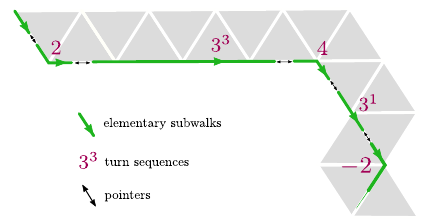}
\caption{The compressed homotopy sequence data structure.}\label{fig:compressed-homotopy-sequence}
\end{figure}

\begin{proof}[Proof of Proposition~\ref{prop:trail-algo}]
  The \emph{extended turn sequence} of a walk~$W$ is its usual turn sequence, with a new symbol~$E$ appended before its first symbol and after its last letter.  We consider the \emph{reduction moves} depicted in Figure~\ref{fig:walk-reduction-moves} that modify a walk $W$ by homotopy.  Each of them can be applied when a specific subsequence appears in the extended turn sequence of~$W$ \emph{or of its reversal}:
\begin{itemize}
    \item\emph{spur move}: $0$,
    \item\emph{spike move}: $1$,
    \item\emph{bracket move}: $23^k2$ for some $k\ge0$,
    \item\emph{flip move}: $a3^k2_r3^\ell b$ for some $k,\ell\ge0$ and some symbols $a,b\not\in\set{-1,0,1,2,3}$.  (In particular, $a$ and/or~$b$ can be equal to~$E$.)
\end{itemize}

\begin{figure}
\centering
\def\svgwidth{40em}
\begingroup%
  \makeatletter%
  \providecommand\color[2][]{%
    \errmessage{(Inkscape) Color is used for the text in Inkscape, but the package 'color.sty' is not loaded}%
    \renewcommand\color[2][]{}%
  }%
  \providecommand\transparent[1]{%
    \errmessage{(Inkscape) Transparency is used (non-zero) for the text in Inkscape, but the package 'transparent.sty' is not loaded}%
    \renewcommand\transparent[1]{}%
  }%
  \providecommand\rotatebox[2]{#2}%
  \newcommand*\fsize{\dimexpr\f@size pt\relax}%
  \newcommand*\lineheight[1]{\fontsize{\fsize}{#1\fsize}\selectfont}%
  \ifx\svgwidth\undefined%
    \setlength{\unitlength}{374.91644912bp}%
    \ifx\svgscale\undefined%
      \relax%
    \else%
      \setlength{\unitlength}{\unitlength * \real{\svgscale}}%
    \fi%
  \else%
    \setlength{\unitlength}{\svgwidth}%
  \fi%
  \global\let\svgwidth\undefined%
  \global\let\svgscale\undefined%
  \makeatother%
  \begin{picture}(1,0.34343041)%
    \lineheight{1}%
    \setlength\tabcolsep{0pt}%
    \put(0,0){\includegraphics[width=\unitlength,page=1]{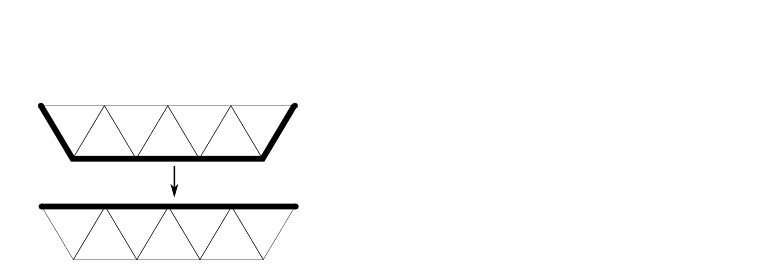}}%
    \put(0.054465,0.10754375){\color[rgb]{0,0,0}\makebox(0,0)[lt]{\lineheight{1.5}\smash{\begin{tabular}[t]{l}bracket move\end{tabular}}}}%
    \put(0,0){\includegraphics[width=\unitlength,page=2]{reduce-3.pdf}}%
    \put(0.13380617,0.3206461){\color[rgb]{0,0,0}\makebox(0,0)[lt]{\lineheight{1.5}\smash{\begin{tabular}[t]{l}spur move\end{tabular}}}}%
    \put(0,0){\includegraphics[width=\unitlength,page=3]{reduce-3.pdf}}%
    \put(0.53289139,0.31974239){\color[rgb]{0,0,0}\makebox(0,0)[lt]{\lineheight{1.5}\smash{\begin{tabular}[t]{l}spike move\end{tabular}}}}%
    \put(0,0){\includegraphics[width=\unitlength,page=4]{reduce-3.pdf}}%
    \put(0.4766444,0.07112417){\color[rgb]{0,0,0}\makebox(0,0)[lt]{\lineheight{1.5}\smash{\begin{tabular}[t]{l}flip move\end{tabular}}}}%
    \put(0,0){\includegraphics[width=\unitlength,page=5]{reduce-3.pdf}}%
  \end{picture}%
\endgroup%

\caption{The reduction moves for the reduction of walks.}\label{fig:walk-reduction-moves}
\end{figure}

We observe that if no reduction move can be applied to~$W$, then $W$ is reduced.  Moreover, for each walk~$W$, let $\varphi(W)$ be equal to three times the length of~$W$, plus the number of bad turns of~$W$.  We observe that $\varphi(W)$ strictly decreases at each move; indeed, a spur, spike, or bracket move decreases the length of $W$ (by at least one) and creates at most two bad turns, while a flip move does not affect the length and removes at least one bad turn (no bad turn is created because the degree of each vertex is at least eight, and by the conditions on $a$ and~$b$).  Thus, any sequence of moves has length~$O(n)$.

The algorithm is as follows.  In a first step, we walk along~$W$ to initialize its compressed homotopy sequence in~$O(n)$ time (using Operation~(1) above). Then, in a second step, we traverse this sequence in order, and reduce~$W$ as soon as a possible move is encountered, maintaining the compressed homotopy sequence as it evolves.  Each possible move involves at most six consecutive elementary subwalks, which are replaced by at most seven elementary subwalks.  Then, before resuming the algorithm and looking for a possible move, we need to backtrack by at most five elementary subwalks in the compressed homotopy sequence, in order not to skip any possible move created by performing the previous one.  Finally, in a third step, when we reach the end of the list, no move is possible any more; we convert our compressed homotopy sequence back to a walk on~$\rtriangulation$ (using Operation~(2) above).  All of this takes $O(n)$ time because the sequence of moves has length~$O(n)$, and because the length of a walk does not increase when reducing it, so that the compressed homotopy sequence always has length~$O(n)$.
\end{proof}

\subsection{Reducing a closed walk}

\begin{proposition}\label{prop:closed-walk-reduction}
 Given a closed walk~$C$ of length~$n$ in a reducing triangulation~$\rtriangulation$, we can compute a strongly reduced closed walk freely homotopic to~$C$ in $O(n)$ time.
\end{proposition}
\begin{proof}
  We first remove the bad turns from~$C$.  For this purpose, we use the same reduction moves as in the proof of Proposition~\ref{prop:trail-algo}, but the turn sequence is now cyclic.  We observe that, in each move, the part of~$C$ that is removed corresponds to an actual subwalk of~$C$ (and does not ``wrap around'' it).  In particular, in a bracket move, the entire cyclic turn sequence of~$C$ cannot be $23^k$, because the color of the triangle to the left of the walk changes after passing a~$2$. As a limit case, in a flip move, the symbols $a$ and~$b$ in the sequence $a3^k2_r3^\ell b$ may correspond to the same vertex of the cyclic walk; see Figure~\ref{fig:walk-reduction-moves-closed}, left for an illustration in the case where the turn at $a=b$ is a 4-turn.
  \begin{figure}
  \centering
  \def\svgwidth{40em}
  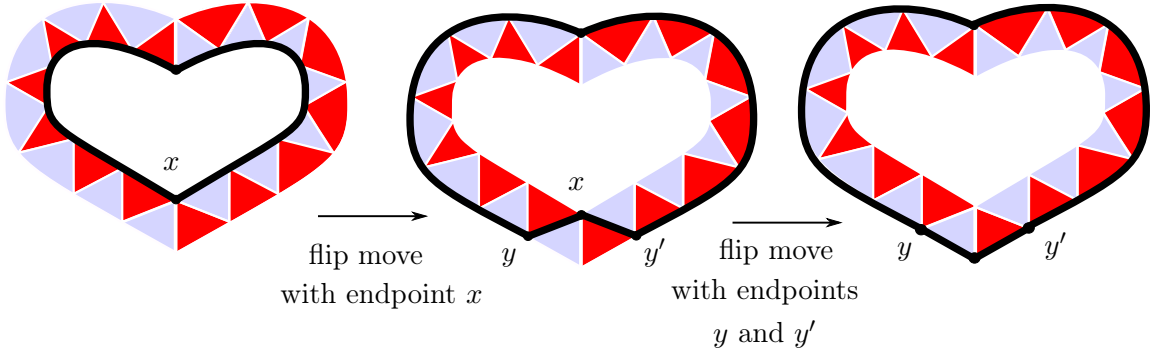
  \caption{The ``heart-like'' situation in the proof of Proposition~\ref{prop:closed-walk-reduction}.}\label{fig:walk-reduction-moves-closed}
  \end{figure}

  As in the proof of Proposition~\ref{prop:trail-algo}, if no move is possible, then $C$ has no bad turn. If $\varphi(C)$ is equal to three times the length of~$C$, plus the number of bad turns of~$C$, then $\varphi(C)$ decreases at each move, except in the very special ``heart-like'' case of Figure~\ref{fig:walk-reduction-moves-closed}, corresponding to the case where the cyclic turn sequence of~$C$ (or its reversal) is $43^k2_r3^\ell$ for some $k,\ell\ge0$.  The flip move does not affect the length and removes the bad 2-turn, but both edges of~$C$ incident to the 4-turn are rotated by one triangle, which creates another bad 2-turn, see Figure~\ref{fig:walk-reduction-moves-closed}.  However, after this move, the only possible move flips the new bad 2-turn into a good one, at which point the closed walk is strongly reduced.  Thus, as in the proof of Proposition~\ref{prop:trail-algo}, any sequence of moves has length $O(n)$.
  
  Finally, if the turn sequence of~$C$ is of the form $(3_r)^*$ or~$(-3_b)^*$, then $C$ can be strongly reduced in $O(1)$ time using the new move depicted in Figure~\ref{fig:walk-reduction-moves-closed-final}.
  \begin{figure}
  \centering
  \includegraphics[scale=1]{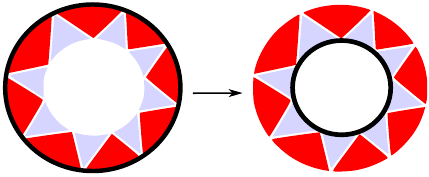}
  \caption{The last reduction move in the proof of Proposition~\ref{prop:closed-walk-reduction}.}\label{fig:walk-reduction-moves-closed-final}
  \end{figure}
\end{proof}

\section{Untangling a loop graph}\label{sec:untangle-loop-graph}

A \emphdef{loop graph} is a graph $\loopgraph$ whose connected components have a single vertex. A drawing $\lambda$ of a loop graph $\loopgraph$ in a surface $\surface$ is \emphdef{sparse} if, under~$\lambda$, the edges of each connected component of~$\loopgraph$ are non-contractible and pairwise non-homotopic.  In this section, we prove:

\begin{proposition}\label{prop:untangling-a-loop-graph-general}
Let $T$ be a reducing triangulation of a surface of genus at least two without boundary. Let $L$ be a loop graph, and let $\lambda : L \to T$ be a sparse drawing of size $n$. In $O(n \log n)$ time, we can determine whether $\lambda$ can be untangled, and if so we can compute a weak embedding $\lambda' : \loopgraph \to T$ homotopic to $\lambda$, in which for every loop $\ell$ of $\loopgraph$ the respective lengths $k'$ and $k$ of the walks $\lambda'(\ell)$ and $\lambda(\ell)$ satisfy $k' = O(k)$.
\end{proposition}

Reducing triangulations exist for every surface of genus at least two; see Figure~\ref{fig:valid-tr-from-canonical-sc}.  Throughout this section, $\surface$ has genus at least two and is without boundary.

\subsection{Relation to weak embeddability of straightened loop graphs}

For the needs of this section, we say that a drawing $\lambda$ of a loop graph $\loopgraph$ in a reducing triangulation $\rtriangulation$ is \emphdef{straightened} if it is sparse and satisfies each of the following on every connected component $\loopgraph_0$ of $\loopgraph$. Firstly, there is an edge $e$ of $\loopgraph_0$ that is mapped by $\lambda$ to a \emph{strongly} reduced \emph{closed} walk in $\rtriangulation$, where $e$ is seen as a (non-directed) closed curve.  (This is consistent because (strongly) reduced (closed) walks are stable upon reversal.)  Secondly, every other edge~$e'$ of $\loopgraph_0$ is mapped by $\lambda$ to a reduced walk in $\rtriangulation$, where $e'$ is seen as a (non-directed) path this time.  We say that $e$ is a \emphdef{major edge} and that $e'$ is a \emphdef{minor edge} of $\loopgraph$ (with respect to $\lambda$).  (A similar idea has been used several times in the continuous setting~\cite[Section~5]{dkpt-23}.)

\begin{lemma}\label{lem:straightening}
Consider a sparse drawing $\lambda$ of size $n$ of a loop graph $\loopgraph$ in a reducing triangulation $\rtriangulation$ of~$\surface$. In time $O(n)$, we can compute a straightened drawing $\lambda'$ homotopic to $\lambda$, in which for every loop $\ell$ of $\loopgraph$ the respective lengths $k'$ and $k$ of the walks $\lambda'(\ell)$ and $\lambda(\ell)$ satisfy $k' = O(k)$.
\end{lemma}
\begin{proof}
The algorithm works independently on the connected components of $\loopgraph$, so assume that $\loopgraph$ is connected and let $v$ be the vertex of $\loopgraph$. We choose the major edge $e$ of $\loopgraph$ as an edge whose image $\lambda(e)$ has minimal length $m \geq 1$. We reduce the closed walk $\lambda(e)$ with Proposition~\ref{prop:closed-walk-reduction} to construct $\lambda'(e)$ in time $O(m)$. While reducing the closed walk $\lambda(e)$, we update homotopically the vertex $\lambda(v)$ so that $\lambda(v)$ remains on $\lambda(e)$ at all times, and we record the walk $W$ performed by $\lambda(v)$ in $\rtriangulation$. In the end, the walk $W$ has length $O(m)$. For every other (minor) edge~$e'$ of~$\loopgraph$, we concatenate the reverse of~$W$ with $\lambda(e')$ and then with~$W$, and we reduce this walk with Proposition~\ref{prop:trail-algo} to construct $\lambda'(e')$ in time linear in the length of $\lambda(e')$.
\end{proof}

Here is the main technical ingredient of this section:
\begin{proposition}\label{prop:untangle-loop-graph}
Consider a straightened drawing $\lambda$ of a loop graph $\loopgraph$ in a reducing triangulation~$\rtriangulation$ of a surface without boundary of genus at least two. If there is an embedding homotopic to $\lambda$, then $\lambda$ is a weak embedding.
\end{proposition}

This proposition implies immediately Proposition~\ref{prop:untangling-a-loop-graph-general}:

\begin{proof}[Proof of Proposition~\ref{prop:untangling-a-loop-graph-general}, assuming Proposition~\ref{prop:untangle-loop-graph}]
Apply Lemma~\ref{lem:straightening} to compute in $O(n)$ time a straightened drawing~$\lambda'$ of~$\loopgraph$ homotopic to $\lambda$. Apply Theorem~\ref{thm:toth-et-al} to determine in $O(n \log n)$ time whether $\lambda'$ is a weak embedding. If not, then conclude by Proposition~\ref{prop:untangle-loop-graph} that there is no embedding homotopic to $\lambda$. Otherwise return $\lambda'$.
\end{proof}

The rest of this section is devoted to the proof of Proposition~\ref{prop:untangle-loop-graph}.  We need to start with some preliminaries.

\subsection{The limit points of lifts in the universal cover}\label{sec:limit-points}

Because $\surface$ has genus at least two and has no boundary, the universal cover~$\tilde\surface$ of~$\surface$ can be compactified into a topological space~$\tilde\surface\cup\partial\tilde\surface$, by adding a set $\partial \tilde \surface$ of \emphdef{limit points}, such that the compactified space is homeomorphic to the closed disk, and under this homeomorphism $\tilde\surface$ is represented by the open disk and $\partial\tilde\surface$ is represented by the circle.  Moreover, every lift $\tilde c:\cR\to \tilde \surface$ of every non-contractible closed curve on $\surface$ admits well defined limit points at $+\infty$ and~$-\infty$ in $\partial \tilde \surface$. We will rely on the following topological properties.

\begin{restatable}{lemma}{LemExistDistinctLimitPoints}\label{lem:exist-distinct-limit-points}
  If $\tilde{c} : \cR \to \tilde \surface$ is a lift of a non-contractible closed curve on $\surface$, then $\lim_{+\infty}\tilde c$ and $\lim_{-\infty}\tilde c$ exist (in $\tilde \surface \cup \partial \tilde \surface$) and are distinct points of $\partial \tilde \surface$.
\end{restatable}

\begin{restatable}{lemma}{LemLiftLimitPoints}\label{lem:lift-limit-points}
Lift a homotopy $c \simeq d$ between non-contractible closed curves on $\surface$ to a homotopy $\tilde{c} \simeq \tilde{d}$ between lifts of $c$ and $d$. Then $\tilde c$ and $\tilde d$ have the same limit points.
\end{restatable}

\begin{restatable}{lemma}{LemSeparateLimitPoints}\label{lem:separate-limit-points}
Consider lifts $\tilde{c}$ and $\tilde{d}$ of non-contractible closed curves on $\surface$. Assume either that $\tilde{c}$ and $\tilde{d}$ intersect exactly once, or that they are disjoint lifts of the same curve on $\surface$. Then the four limit points of $\tilde{c}$ and $\tilde{d}$ are pairwise distinct.
\end{restatable}

\begin{restatable}{lemma}{LemSameLimit}\label{L:samelimit}
  Let $c$ and~$d$ be two non-contractible closed curves on $\surface$.  If $c$ and $d$ admit lifts with the same pairs of limit points, there is a closed curve~$e$ such that $c$ and~$d$ are homotopic to powers of~$e$.
\end{restatable}

For the sake of concision, we defer the (classical) construction of the limit points and the proofs of the above lemmas to Appendix~\ref{A:limit}. Interestingly, although the properties of the limit points are expressed purely topologically above, they all follows from folklore arguments in hyperbolic geometry:  The limit points are defined by endowing~$\surface$ with a hyperbolic metric and lifting this metric to its universal cover $\tilde \surface$; then $\tilde \surface$ is isometric to the hyperbolic plane, which can be compactified as described in Appendix~\ref{sec:compactification} (by adding the boundary of the open disk in the Poincaré model), the limit points are naturally defined, and the properties follow.  However, we stress that this hyperbolic structure is only used for the proofs in the appendix.  We will need a different hyperbolic metric in the coming subsections.

\subsection{Patch system and hyperbolic metric on the punctured surface}

To prove Proposition~\ref{prop:untangle-loop-graph}, we assume that there is an embedding homotopic to $\lambda$ in $\surface$, and we will prove that $\lambda$ is a weak embedding.  Recall that since the reducing triangulation~$\rtriangulation$ is a graph cellularly embedded on~$\surface$, its patch system corresponds to its dual.  More precisely, to build (the interior of) the patch system~$\patchsystem$ of~$\rtriangulation$, we remove (puncture) a point $p_t$ from the interior of each triangle~$t$ of $\rtriangulation$; then, for each edge $e$ of $\rtriangulation$, we consider the two distinct triangles $t,t'$ incident to $e$ in $\rtriangulation$, and we draw a simple arc between $p_t$ and $p_{t'}$ that crosses $e$ exactly once and does not intersect any other edge or any vertex of $\rtriangulation$. We assume without loss of generality that the resulting arcs are pairwise disjoint. From now on we see the patch system $\patchsystem$ as included in $\surface$ in this way. (Compared to Section~\ref{sec:weak-embeddings}, the patch system is open, not closed, but this does not change anything.)  Now recall from Section~\ref{sec:weak-embeddings} the notion of map $L \to \patchsystem$ that \emph{approximates} $\lambda$. We will show that $\lambda$ is a weak embedding by exhibiting an embedding $L \to \patchsystem$ that approximates $\lambda$. We construct our candidate $\mu : L \to \patchsystem$ for an embedding that approximates $\lambda$ as follows.

Firstly, we endow $\patchsystem$ with a complete hyperbolic metric for which the arcs of $\patchsystem$ are geodesics, as follows. The arcs of $\patchsystem$ separate $\patchsystem$ into open disks. Each resulting connected component $P$ has $k \geq 8$ incidences with arcs since the vertices of $\rtriangulation$ have degree at least eight. We replace $P$ by a hyperbolic polygon with $k$ geodesic sides, with the particularity that these $k$ geodesic sides all have \emph{infinite} length (such polygons are called ideal polygons, see Section~\ref{sec:hyperbolic-surfaces}). Then we identify the sides of the polygons by pairs corresponding to the arcs of $\patchsystem$.

\begin{figure}
\centering
\includegraphics[scale=0.3]{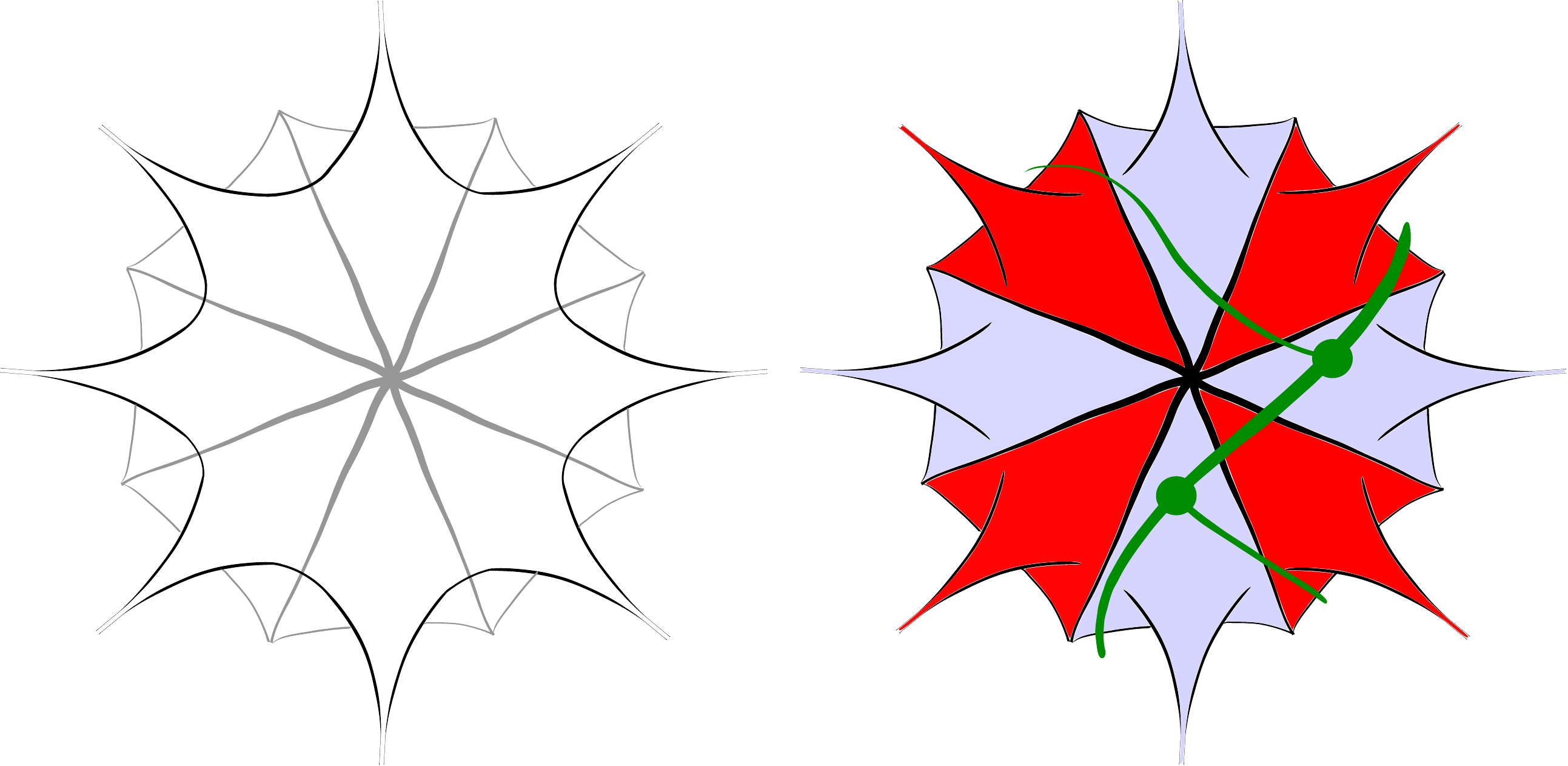}
\caption{The map $\mu$ (in green) in the hyperbolic patch system $\patchsystem$ of the reducing triangulation $\rtriangulation$, in the proof of Proposition~\ref{prop:untangle-loop-graph}.}\label{fig:mu}
\end{figure}

Secondly, we see $\lambda$ as a map from $\loopgraph$ to $\patchsystem$ and we construct $\mu$ by modifying $\lambda$ homotopically on every connected component $\loopgraph_0$ of $\loopgraph$, as follows. Let $v$ be the vertex of $\loopgraph_0$ and let $e$ be the major edge of $\loopgraph_0$. We start by replacing homotopically the image of $e$ by the geodesic closed curve $c_e$ in its free homotopy class, thus possibly changing the image of $v$. Here we make use of the fact that $\lambda(e)$ is neither contractible nor homotopic to the neighborhood of a puncture in $\patchsystem$, since it is non-contractible in $\surface$. We require that during this homotopy $\lambda \simeq \mu$ the image of $v$ remains in the same disk of $\patchsystem$ (never belongs to an arc of $\patchsystem$). We can do so since the strongly reduced closed walk $\lambda(e)$ makes no 0-turn and thus has the same sequence of crossings with the arcs of $\patchsystem$ as the geodesic closed curve $c_e$. Then, we replace for each minor edge $e'$ of $\loopgraph_0$ the image of $e'$ by the geodesic segment in its homotopy class, this time holding the image of $v$ fixed. The map $\mu$ approximates $\lambda$ since the images of the vertices of $\loopgraph$ remained in the same disks of $\patchsystem$ and since for every edge~$e''$ of $\loopgraph$ the walk $\lambda(e'')$ makes no 0-turn.

Finally, we may ensure without loss of generality that the vertices of $\loopgraph$ are mapped by $\mu$ to distinct points of $\patchsystem$, and belong to no arc of~$\patchsystem$. That can be achieved by moving infinitesimally the image of each vertex of $\loopgraph$ along the image of its major edge.

At this point $\mu$ is not necessarily an embedding; indeed, several major edges can overlap. The strategy is to prove that this is the only reason why $\mu$ may fail to be an embedding, and that such overlaps can be eliminated by a small perturbation of~$\mu$.

\subsection{Proof of Proposition~\ref{prop:untangle-loop-graph}}

In this section, we conclude the proof of Proposition~\ref{prop:untangle-loop-graph}.  To ease the reading, we abuse the notation and write $\loopgraph$ to refer to the map $\mu$ on $\loopgraph$ (defined in the previous section).

Recall from Section~\ref{sec:limit-points} that the universal cover $\tilde \surface$ of~$\surface$ is homeomorphic to an open disk, which can be compactified into a closed disk~$\tilde \surface \cup \partial \tilde \surface$, so that the properties of Section~\ref{sec:limit-points} can be applied. Note that, here, we view $\tilde \surface \cup \partial \tilde \surface$ as a purely topological space, and completely ignore any metric on it. In this section, we consider the subset~$\tilde\patchsystem$ of~$\tilde\surface$ made of the points that are lifts of~$\patchsystem$; in other words, $\tilde\patchsystem$ is obtained from~$\tilde\surface$ by removing the lifts of the punctures.
(It turns out that $\tilde\patchsystem$ is a covering space of~$\patchsystem$, but we will never use this property.)  $\tilde\patchsystem$ naturally inherits a hyperbolic metric, obtained by lifting the hyperbolic metric of~$\patchsystem$.

By construction, lifts of edges of~$\loopgraph$ are geodesics under the metric of $\tilde \patchsystem$.  A \emphdef{major lift} is a lift of a major edge; it is a bi-infinite geodesic path in $\tilde \patchsystem$ with two distinct limit points on~$\partial \tilde \surface$, by Lemma~\ref{lem:exist-distinct-limit-points} and sparsity. A \emphdef{minor lift} is a lift of a minor edge; it is a geodesic path in $\tilde \patchsystem$.

If $P$ and $Q$ are minor lifts of distinct minor edges, then they may share one of their endpoints, but not both, since $\loopgraph$ is sparse and since the vertices of $\loopgraph$ are distinct points of $\patchsystem$.  In particular, the minor lifts of distinct minor edges are distinct.  However, the situation is slightly more complicated for major lifts.  Two major lifts obtained from different major edges are considered different.  Recall that major edges are primitive (by Lemma~\ref{L:primitive} and sparsity). So any two major lifts $\tilde c,\tilde d:\cR\to\tilde\patchsystem$ of the same major edge that differ only by a homeomorphism of~$\cR$ actually differ by an integer translation; in that case, we see them as equal.

The proof of Proposition~\ref{prop:untangle-loop-graph} combines a long series of relatively easy lemmas on the lifts of the edges of~$\loopgraph$ in~$\tilde\patchsystem$. Recall that by assumption $L$ can be untangled.

\begin{lemma}\label{L:nogon}
  Every lift of an edge is simple.  Any two lifts of edges (possibly the same edge) that do not overlap intersect at most once, and if they intersect in their relative interiors, then they cross.
\end{lemma}
\begin{proof}
If two portions of geodesic paths intersect in their relative interiors without overlapping, then they cross. Now assume, for a contradiction, that a lift is not simple; it contains a loop, which projects to a geodesic loop $\ell$ in $\patchsystem$ that is contractible in~$\surface$. The sequence of crossings of $\ell$ with the arcs of $\patchsystem$ is that of a reduced walk.  Thus, by Proposition~\ref{prop:trail-algo}, the loop~$\ell$ does not cross any arc of $\patchsystem$, and so $\ell$ is contractible in $\patchsystem$, which is impossible since $\ell$ is geodesic. 

Similarly, if two lifts intersect twice without overlapping, they form two paths with the same endpoints and otherwise disjoint, which project to geodesic paths $p$ and~$q$ in~$\patchsystem$ that are homotopic in~$\surface$.  The sequence of crossings of $p$ and~$q$ with the arcs of~$\patchsystem$ must be the same by Proposition~\ref{prop:trail-algo}, so $p$ and~$q$ are homotopic in~$\patchsystem$, which is impossible since $p$ and~$q$ are geodesics.
\end{proof}

\begin{lemma}\label{L:maj}
    Any two major lifts either have the same image or are disjoint.
\end{lemma}
\begin{proof}
    Assume that two major lifts $P$ and~$Q$ intersect and do not have the same image. They cross exactly once by Lemma~\ref{L:nogon}. But then, their limit points are pairwise distinct by Lemma~\ref{lem:separate-limit-points}, and interleaved on~$\partial \tilde \surface$ (see Figure~\ref{fig:major-major-minor} (left)), which implies that $\loopgraph$ cannot be untangled by Lemma~\ref{lem:lift-limit-points}.
\end{proof}

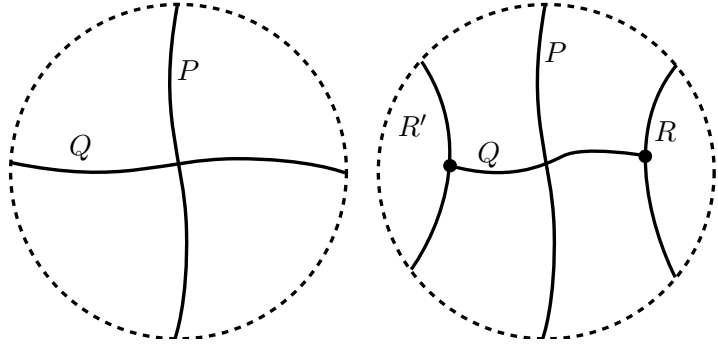
\begin{figure}
\centering
\def\svgwidth{25em}
\begingroup%
  \makeatletter%
  \providecommand\color[2][]{%
    \errmessage{(Inkscape) Color is used for the text in Inkscape, but the package 'color.sty' is not loaded}%
    \renewcommand\color[2][]{}%
  }%
  \providecommand\transparent[1]{%
    \errmessage{(Inkscape) Transparency is used (non-zero) for the text in Inkscape, but the package 'transparent.sty' is not loaded}%
    \renewcommand\transparent[1]{}%
  }%
  \providecommand\rotatebox[2]{#2}%
  \newcommand*\fsize{\dimexpr\f@size pt\relax}%
  \newcommand*\lineheight[1]{\fontsize{\fsize}{#1\fsize}\selectfont}%
  \ifx\svgwidth\undefined%
    \setlength{\unitlength}{156.10342026bp}%
    \ifx\svgscale\undefined%
      \relax%
    \else%
      \setlength{\unitlength}{\unitlength * \real{\svgscale}}%
    \fi%
  \else%
    \setlength{\unitlength}{\svgwidth}%
  \fi%
  \global\let\svgwidth\undefined%
  \global\let\svgscale\undefined%
  \makeatother%
  \begin{picture}(1,0.47736671)%
    \lineheight{1}%
    \setlength\tabcolsep{0pt}%
    \put(0,0){\includegraphics[width=\unitlength,page=1]{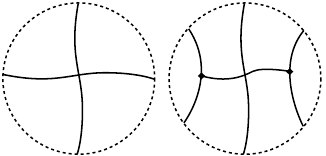}}%
    \put(0.23972012,0.362613){\color[rgb]{0,0,0}\makebox(0,0)[lt]{\lineheight{1.25}\smash{\begin{tabular}[t]{l}$P$\end{tabular}}}}%
    \put(0.75052051,0.38804342){\color[rgb]{0,0,0}\makebox(0,0)[lt]{\lineheight{1.25}\smash{\begin{tabular}[t]{l}$P$\end{tabular}}}}%
    \put(0.09043249,0.25865749){\color[rgb]{0,0,0}\makebox(0,0)[lt]{\lineheight{1.25}\smash{\begin{tabular}[t]{l}$Q$\end{tabular}}}}%
    \put(0.65784072,0.24933811){\color[rgb]{0,0,0}\makebox(0,0)[lt]{\lineheight{1.25}\smash{\begin{tabular}[t]{l}$Q$\end{tabular}}}}%
    \put(0.90302536,0.27883505){\makebox(0,0)[lt]{\lineheight{1.25}\smash{\begin{tabular}[t]{l}$R$\end{tabular}}}}%
    \put(0.54720872,0.28494846){\makebox(0,0)[lt]{\lineheight{1.25}\smash{\begin{tabular}[t]{l}$R'$\end{tabular}}}}%
  \end{picture}%
\endgroup%

\caption{(Left) Proof of Lemma~\ref{L:maj}. (Right) Proof of Lemma~\ref{L:majmin}.}\label{fig:major-major-minor}
\end{figure}

\begin{lemma}\label{L:majminsame}
  Let $P$ be a major lift and $Q$ a minor lift such that $P$ and~$Q$ are lifts of edges in the same connected component of~$\loopgraph$, and such that an endpoint of~$Q$ belongs to the image of~$P$.  Then no other point of~$Q$ lies on the image of~$P$.
\end{lemma}
\begin{proof}
  Otherwise, $P$ and~$Q$ overlap by Lemma~\ref{L:nogon}. Let $P'$ be the part of~$P$ that starts and ends at the endpoints of~$Q$ (which are both lifts of the vertex of the corresponding connected component~$\loopgraph_0$ of~$\loopgraph$).  By Lemma~\ref{L:primitive}, the projection of~$P'$ is the major edge of~$\loopgraph_0$.  Moreover, the projection of~$Q$ is a minor edge of~$\loopgraph_0$.  This contradicts the sparsity of~$\loopgraph$.
\end{proof}

Before treating the other cases, we need a definition.  Let $P$ be a minor lift and let $v$ and~$v'$ be the endpoints of~$P$ (they are distinct). Let $p$ be the minor edge of~$\loopgraph$ that is the projection of~$P$, and let $q$ be the major edge of~$\loopgraph$ that lies in the same connected component of~$\loopgraph$ as~$p$.  Let $Q$ and~$Q'$ be the major lifts, starting at~$v$ and~$v'$ respectively, that are lifts of~$q$.  We say that $Q$ and~$Q'$ form the \emphdef{H-block} of~$P$.  Indeed, $Q$, $P$, and~$Q'$ together form the shape of the embedded letter~``H'', because they touch only at $v$ and~$v'$ by the preceding lemmas, and moreover the four limit points of $Q$ and~$Q'$ are pairwise distinct by Lemma~\ref{lem:separate-limit-points}.

\begin{lemma}\label{L:majmin}
  There is no intersection between a major lift~$P$ and the relative interior of a minor lift~$Q$.
\end{lemma}
\begin{proof}
  Let $R$ and~$R'$ be the major lifts forming the H-block of~$Q$. Assume first that there is an intersection between $P$ and the relative interior of~$Q$ that is not a crossing. Then $P$ and~$Q$ overlap by Lemma~\ref{L:nogon}. In particular, $P$ intersects $R$ and $R'$. However $R$ and~$R'$ are disjoint, being part of the same H-block. This contradicts Lemma~\ref{L:maj}.

  Now assume that there is an intersection between $P$ and the relative interior of~$Q$ that is a crossing.  By Lemma~\ref{L:nogon}, there is exactly one. By Lemma~\ref{L:maj}, and since $P$ does not have the same image as~$R$ or~$R'$, the major lifts $P$, $R$, and~$R'$ are disjoint.  See Figure~\ref{fig:major-major-minor} (right).  Also, $P$ does not have the same pair of limit points as~$R$ (or~$R'$): for otherwise, by Lemma~\ref{L:samelimit} and Lemma~\ref{L:primitive}, $P$ and $R$ would project to closed curves freely homotopic in $S$ (up to reversal), and thus freely homotopic in $\patchsystem$ by Proposition~\ref{prop:uniqueness-reduced-closed-walks}, but then $P$ and $R$ would project to the same geodesic closed curve, which is impossible since they are disjoint.  Any homotopy of~$\loopgraph$ induces a homotopy of the lifts $P$, $R$, and~$R'$, preserving the limit points on~$\partial \tilde \surface$.  Thus, even after a homotopy, one crossing between $P$ and $Q$, $R$, or~$R'$ must remain (even if $P$ shares one endpoint with~$R$ and/or one endpoint with~$R'$). This contradicts the assumption that $L$ can be untangled.
\end{proof}

Let $v$ be a vertex of~$\loopgraph$, and let $e$ be a directed major edge of~$\loopgraph$ based at~$v$.  We say that $e$ is \emphdef{pulled to the right} if there exists a minor edge based at~$v$ that leaves~$v$ and/or arrives at~$v$ to the right of~$e$; in other words, in counterclockwise order around~$v$, we do not see consecutively the target of~$e$ and the source of~$e$.  We say that a directed major lift is \emph{pulled to the right} if its projection is.

\begin{lemma}\label{L:majoverlap}
  Let $P$ and~$Q$ be two major lifts, projecting to different edges of~$\loopgraph$. Assume that they overlap; assume (up to reversing one of them) that they have the same direction.  Then (1) at most one of $P$ and~$Q$ is pulled to the right, and (2) none of them is pulled both to the left and to the right.
\end{lemma}
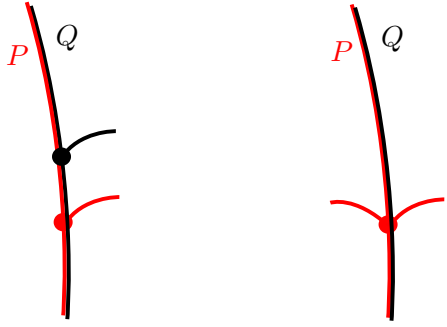
\begin{figure}
\centering
\def\svgwidth{20em}
\begingroup%
  \makeatletter%
  \providecommand\color[2][]{%
    \errmessage{(Inkscape) Color is used for the text in Inkscape, but the package 'color.sty' is not loaded}%
    \renewcommand\color[2][]{}%
  }%
  \providecommand\transparent[1]{%
    \errmessage{(Inkscape) Transparency is used (non-zero) for the text in Inkscape, but the package 'transparent.sty' is not loaded}%
    \renewcommand\transparent[1]{}%
  }%
  \providecommand\rotatebox[2]{#2}%
  \newcommand*\fsize{\dimexpr\f@size pt\relax}%
  \newcommand*\lineheight[1]{\fontsize{\fsize}{#1\fsize}\selectfont}%
  \ifx\svgwidth\undefined%
    \setlength{\unitlength}{104.78780365bp}%
    \ifx\svgscale\undefined%
      \relax%
    \else%
      \setlength{\unitlength}{\unitlength * \real{\svgscale}}%
    \fi%
  \else%
    \setlength{\unitlength}{\svgwidth}%
  \fi%
  \global\let\svgwidth\undefined%
  \global\let\svgscale\undefined%
  \makeatother%
  \begin{picture}(1,0.59243698)%
    \lineheight{1}%
    \setlength\tabcolsep{0pt}%
    \put(0,0){\includegraphics[width=\unitlength,page=1]{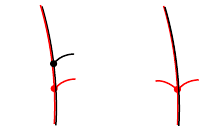}}%
    \put(0.14911882,0.4601921){\color[rgb]{1,0,0}\makebox(0,0)[lt]{\lineheight{1.25}\smash{\begin{tabular}[t]{l}$P$\end{tabular}}}}%
    \put(0.23592508,0.49533528){\makebox(0,0)[lt]{\lineheight{1.25}\smash{\begin{tabular}[t]{l}$Q$\end{tabular}}}}%
    \put(0.7126099,0.46743274){\color[rgb]{1,0,0}\makebox(0,0)[lt]{\lineheight{1.25}\smash{\begin{tabular}[t]{l}$P$\end{tabular}}}}%
    \put(0.80099611,0.49302001){\makebox(0,0)[lt]{\lineheight{1.25}\smash{\begin{tabular}[t]{l}$Q$\end{tabular}}}}%
  \end{picture}%
\endgroup%

\caption{The two cases forbidden by Lemma~\ref{L:majoverlap}.}\label{fig:left-right-property}
\end{figure}
\begin{proof}
  We first prove~(1); see Figure~\ref{fig:left-right-property}, left.  Assume, for the sake of a contradiction, that both $P$ and~$Q$ are pulled to the right.  Let $\bar P$ and~$\bar Q$ be these lifts after a homotopy that untangles~$\loopgraph$.  (Generally, we use bars to denote lifts after this homotopy.) They must be disjoint (except at their limit points), so without loss of generality, up to exchanging $P$ and~$Q$, assume that $\bar P$ lies to the left of~$\bar Q$.  Let $R$ be the minor lift pulling $P$ to the right; thus, $P$ is part of the H-block of $R$; let $P'$ be the lift that forms the H-block of~$R$ together with~$P$.  Then, $\bar Q$ intersects either $\bar R$, or $\bar P'$ (twice), which is impossible.

  We now prove~(2); see Figure~\ref{fig:left-right-property}, right.  Assume that $P$ is pulled both to the left and to the right.  After the homotopy, the lifts $\bar P$ and~$\bar Q$ are disjoint, so without loss of generality, up to reversing the orientations of $P$ and~$Q$, assume that $\bar P$ lies to the left of~$\bar Q$.  We are then in the same situation as the previous paragraph, and conclude similarly.
\end{proof}

\begin{lemma}\label{L:minmin}
  The relative interiors of any two distinct minor lifts are disjoint.
\end{lemma}
\begin{proof}
  If the relative interiors of two distinct minor lifts $P$ and~$Q$ intersect without crossing, they overlap. Since $P$ and $Q$ are distinct they do not have the same pair of endpoints. Thus the relative interior of one of $P$ and~$Q$ must intersect a major lift, a contradiction with Lemma~\ref{L:majmin}.

\begin{figure}
\centering
\def\svgwidth{40em}
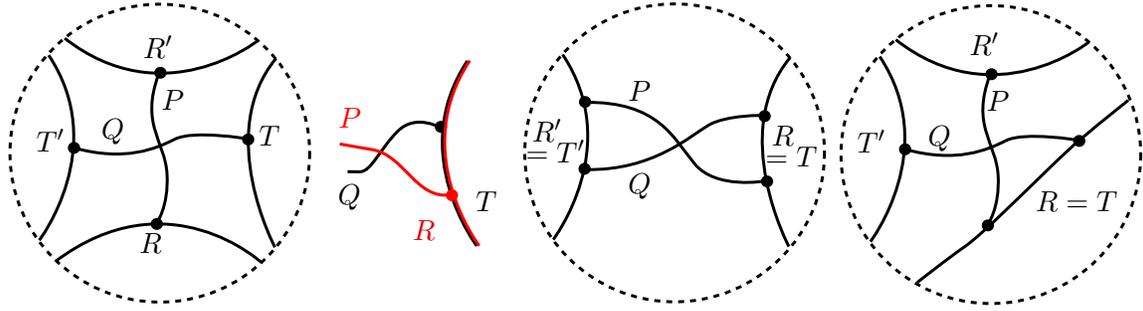
\caption{The four cases in the proof of Lemma~\ref{L:minmin}.}\label{fig:minor-minor}
\end{figure}

  If $P$ and~$Q$ cross, they do so exactly once by Lemma~\ref{L:nogon}. Let $R$ and~$R'$ be the major lifts forming the H-block of~$P$, and similarly let $T$ and~$T'$ be the major lifts forming the H-block of~$Q$. By Lemma~\ref{L:maj}, and up to exchanging notations, we distinguish between four cases, depicted in Figure~\ref{fig:minor-minor}:
  \begin{itemize}
  \item If $R$, $R'$, $T$, and~$T'$ are pairwise disjoint, the cyclic ordering of the limit points on~$\partial \tilde \surface$ is necessarily $r_1r_2t_1t_2r'_1r'_2t'_1t'_2$, with obvious notations ($R$ has limit points $r_1$ and~$r_2$, and so on), with possible identifications of two consecutive limit points in this cyclic order if they come from different lifts.  Any homotopy of~$\loopgraph$ induces a homotopy of the lifts $P$, $Q$, $R$, $R'$, $T$, and~$T'$, preserving the limit points on~$\partial \tilde \surface$.  Thus, even after a homotopy, one crossing must remain. This contradicts the assumption that $L$ can be untangled.
  \item If $R$ is distinct from $T$ but has the same image as $T$, these two lifts with the same image, when directed in the same way, are both pulled to the right, or both pulled to the left, which is impossible (Lemma~\ref{L:majoverlap}(1)).
  \item If $R$ is equal to $T$ and $R'$ is equal to $T'$, then any homotopy of $\loopgraph$ induces a homotopy of $P,Q,R,R'$, preserving the four limit points on $\partial \tilde \surface$ and the relative orders of the endpoints of $P$ and $Q$ on $R$ and $R'$. Thus, even after a homotopy, one crossing must remain.
  \item If $R$ is equal to $T$ and $R'$ is disjoint from $T'$, then, again, a crossing must remain after lifting a homotopy of $\loopgraph$.\qedhere
  \end{itemize}
\end{proof}

\begin{proof}[Proof of Proposition~\ref{prop:untangle-loop-graph}]
  By Lemmas \ref{L:nogon}, \ref{L:majmin}, and~\ref{L:minmin}, the only reason why $\loopgraph$ may fail to be an embedding is because two distinct major lifts intersect, which implies that they have the same image by Lemma~\ref{L:maj}, and thus come from distinct major edges.

  So now, consider an inclusionwise maximal set~$A$ of at least two overlapping, simple, major edges, directed in the same way. By Lemma~\ref{L:majoverlap} $A$ contains at most one edge pulled to the left (and not to the right), at most one edge pulled to the right (and not to the left), and possibly several edges which are pulled neither to the left nor to the right.  We can slightly perturb these edges to make them disjoint: from left to right, the edge pulled to the left (if it exists), then, the edges not pulled at all, and finally, the edge pulled to the right (if it exists). After this operation we have an embedding of~$\loopgraph$ that approximates $\lambda$.
\end{proof}

\section{Untangling on reducing triangulations: Proof of Theorem~\ref{thm:main-theorem}}\label{sec:subdivision}

In this section, we prove Theorem~\ref{thm:main-theorem}, by providing an algorithm to untangle a graph on a reducing triangulation. We already know how to do this for sparse loop graphs, so at a very high level, our strategy is to transform the input graph~$\graph$ into a sparse loop graph; but this hides many technicalities, both at the conceptual and at the computational level. We first compute a \emph{graph factorization}, described in the next subsection.   Then we need some auxiliary lemmas, mostly relating the possibility of untangling a graph to that of the corresponding sparse loop graph.  Finally, we conclude the proof of Theorem~\ref{thm:main-theorem}, with a weaker complexity than announced; this is fixed in Section~\ref{sec:contraction}, which contains a better algorithm for computing a graph factorization.

\subsection{Graph factorization}\label{S:factorization}

To transform an arbitrary graph $\graph$ into a sparse loop graph, the natural idea is to contract a spanning forest of~$\graph$, to ignore the resulting contractible loops, and to identify the resulting homotopic loops.  In this section, we describe this operation formally and provide a first, rather naïve, algorithm to compute it.  Unless noted otherwise, $\surface$ is an arbitrary surface, with or without boundary.

\begin{figure}[!h]
\centering
\includegraphics[scale=0.8]{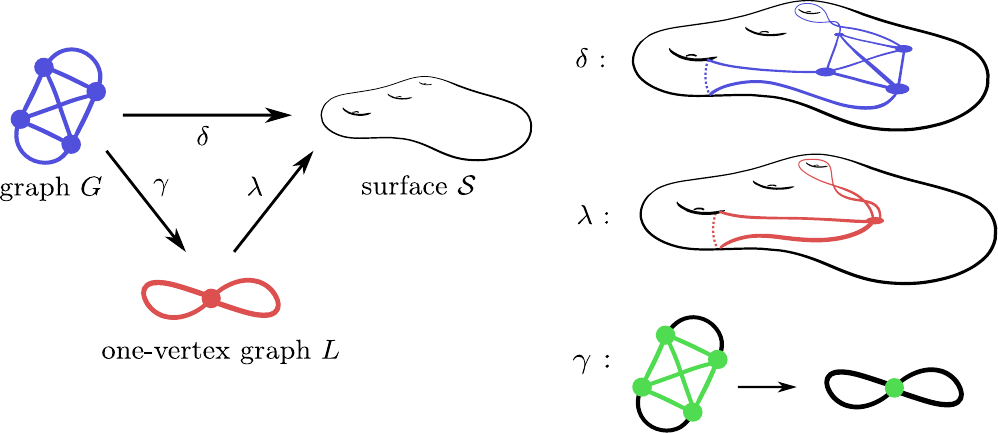}
\caption{A factorization $(L, \lambda, \gamma)$ of the drawing $(G, \delta)$ of Figure~\ref{fig:graph-plane}.}\label{fig:example-1}
\end{figure}

We need to carefully distinguish an abstract graph from its drawings (or embeddings).  Let $\graph_1$ be a connected component of~$\graph$ and let $\drawing$ be the input drawing of~$\graph_1$ on~$\surface$; moreover, let $\spanningtree$ be a spanning tree of~$\graph_1$. A \emphdef{factorization} of $(\graph_1,\drawing)$ (with respect to~$\spanningtree$) is obtained by the following process (see Figure~\ref{fig:example-1}).  First, change the drawing~$\drawing$ of~$\graph_1$ by contracting homotopically the spanning tree~$\spanningtree$ to a single point~$p$.  Let $\drawing'$ be the new drawing of~$\graph_1$ on~$\surface$.  The non-tree edges of~$\graph_1$ are now drawn as loops under~$\drawing'$.  Move these loops that are contractible to the constant loop based at~$p$.  Now, whenever there are several loops of~$\graph_1$ in the same homotopy class (with basepoint~$p$), select one of these loops, $\ell$, arbitrarily, and redraw the other ones in the same way as~$\ell$.  Let $\drawing''$ be the new drawing of~$\graph_1$ on~$\surface$.  We remark that, under~$\drawing''$, the graph~$\graph_1$ is first drawn onto a one-vertex graph~$\loopgraph$, which is itself sparsely drawn in~$\surface$.  The factorization of~$(\graph_1,\drawing)$ is given by the one-vertex graph~$\loopgraph$, its sparse drawing~$\lambda$ on~$\surface$, and the drawing~$\gamma$ of~$\graph_1$ onto~$\loopgraph$.  By construction, $\drawing$ and~$\lambda\circ\gamma$ are homotopic; see Figure~\ref{fig:example-1}, left.

Equivalently, but we will not need this equivalence, a factorization of $(\graph_1,\drawing)$ (with respect to~$\spanningtree$) is given by:
\begin{itemize}
\item a one-vertex graph~$\loopgraph$, and a \emph{sparse} drawing~$\lambda$ of~$\loopgraph$ on~$\surface$,
\item a drawing~$\gamma$ of~$\graph_1$ \emph{onto}~$\loopgraph$, mapping each edge of~$\spanningtree$ to the vertex of~$\loopgraph$, and each edge of~$\graph_1\setminus \spanningtree$ to either the vertex of~$\loopgraph$ or a loop of~$\loopgraph$,
\end{itemize}
such that the drawings $\drawing$ and~$\lambda\circ\gamma$ of~$\graph_1$ on~$\surface$ are homotopic.

A \emphdef{factorization} of a drawing $\drawing$ of a disconnected graph $\graph$ is given by a factorization for each of the connected components of $\graph$. The \emphdef{depth} of a factorization $(L,\lambda,\gamma)$ of~$(\graph,\drawing)$ is the maximal length of the walks $\lambda(\ell)$ over the loops $\ell$ of $L$. The \emphdef{width} of $(L,\lambda,\gamma)$ is the number of loops of $L$.

A crucial tool is the following proposition.
\begin{proposition}\label{prop:contraction-algorithm}
  Under the hypotheses of Theorem~\ref{thm:main-theorem}, we can, in $O(gn\log(gn))$ time, compute a factorization of~$(\graph,\drawing)$ of width $O(g)$ and depth $O(n)$, or correctly report that $\drawing$ cannot be untangled on~$\surface$.
\end{proposition}

We first give a brute force algorithm to prove this proposition with a worse running time, namely $O(n^3)$, and output width, namely $O(n)$, in order to illustrate the general strategy.  While this can be improved using simple arguments, to achieve the claimed running time of $O(gn\log(gn))$ requires  several algorithmic ingredients.  Most notably, we need a new data structure, possibly of independent interest, to store the contracted loops of~$\graph$ in a compressed form.  We defer the full proof of Proposition~\ref{prop:contraction-algorithm}  to Section~\ref{sec:contraction}.

\begin{proof}[Proof of Proposition~\ref{prop:contraction-algorithm}, with a worse running time and output width]
  Without loss of generality, we may assume that $\graph$ is connected.  We take an arbitrary vertex~$v$ of an arbitrary spanning tree~$\spanningtree$ of~$\graph$.  We then contract~$\spanningtree$ to~$v$, thus obtaining a drawing~$\drawing'$ of~$\graph$ in which all edges of~$\spanningtree$ are mapped to~$v$.  Each of the $O(n)$ non-tree edges is drawn in~$\drawing$ as a walk of length $O(n)$.  We reduce all these walks in $O(n^2)$ total time using Proposition~\ref{prop:trail-algo}.    
  and contract, in the drawing~$\drawing$, the edges of~$\spanningtree$.  Now, by Proposition~\ref{prop:trail-homotopy}, this ensures that contractible loops are shrunk to the basepoint and that homotopic loops overlap, which implies that the image of~$\graph$ is that of a sparse drawing of a one-vertex graph.

  We then compare pairwise these $O(n)$ walks, each of length $O(n)$, in $O(n^3)$ total time.  From there, we can immediately compute the loop graph~$\loopgraph$, whose image is the union of these walks, its drawing~$\lambda$, and the map $\gamma$. The width and the depth of the output are both $O(n)$.
\end{proof}

\subsection{Auxiliary lemmas}\label{sec:auxiliary-lemmas}

We present three auxiliary lemmas needed for the proof of Theorem~\ref{thm:main-theorem}.  These lemmas are valid for arbitrary (compact, orientable) surfaces, with arbitrary genus, with or without boundary, which will be useful later.

The following lemma might not be new, but we could not find a reference, so we provide a proof.  It may be of independent interest.  We will only use its corollary (Corollary~\ref{cor:loops-and-order}).  An \emphdef{isotopy} between two embeddings of a graph~$\graph$ is a continuous family of embeddings of~$\graph$ between them.

\begin{lemma}\label{lem:loops-and-order}
  On a surface~$\surface$ with or without boundary, let each of $L=(\ell_1,\ldots,\ell_k)$  and~$L'=(\ell'_1,\ldots,\ell'_k)$ be a set of simple, pairwise disjoint, pairwise non-homotopic, non-contractible loops with basepoint~$v$.  Assume that $\ell_i$ and~$\ell'_i$ are homotopic (with basepoint fixed) for each~$i$. Then $\loopgraph$ and~$L'$ are isotopic (with basepoint fixed).
\end{lemma}
\begin{proof}
  Let $\hat\surface$ be obtained from~$\surface$ by removing a small disk~$D$ around~$v$.  We can assume that all loops are piecewise linear with respect to a fixed triangulation of~$\surface$~\cite[Appendix]{e-c2mi-66}, and thus that no crossing occurs in~$D$, and that each loop crosses the boundary of~$D$ exactly twice.  We let $\hat\ell_i$ and $\hat\ell'_i$ be the arcs in~$\hat\surface$ that are the pieces of $\ell_i$ and~$\ell'_i$ obtained after removing~$D$, and $\hat L$ and $\hat L'$ be the corresponding sets of arcs.

  For each $i$, since $\ell_i$ and~$\ell'_i$ are homotopic on~$\surface$, they are isotopic on~$\surface$, by a result of Epstein~\cite[Theorem~4.1]{e-c2mi-66}.  So $\hat\ell_i$ and~$\hat\ell'_i$ are homotopic on $\hat\surface$, where the homotopy allows to slide the endpoints on $\partial\hat\surface$.

   Because $\hat L$ and~$\hat L'$ are simple, whenever there is an embedded bigon or an embedded ``half-bigon''~\cite[Section~1.2.7]{fm-pmcg-12} between $\hat L$ and~$\hat L'$, there is an innermost embedded bigon or innermost embedded half-bigon, which we can remove using an isotopy of~$L$ (sliding endpoints on the boundary is allowed), decreasing the number of crossings.  So without loss of generality we can assume that there is no embedded bigon or half-bigon between $\hat L$ and $\hat L'$.

  In particular, $\hat\ell_i$ and~$\hat\ell'_i$ are in minimal position in their homotopy classes~\cite[Section~1.2.7]{fm-pmcg-12}, and since they are homotopic (allowing sliding on the boundary), they are disjoint.  The corresponding loops $\ell_i$ and~$\ell'_i$ bound a disk on~$\surface$.  Moreover, because the loops in~$\loopgraph$ are pairwise non-homotopic, and because there is no embedded (half-)bigon between $\hat L$ and~$\hat\ell'_i$, the disk does not meet~$\loopgraph$.  We can thus, for each~$i$, push $\ell_i$ to~$\ell'_i$ by an isotopy of~$L$.
\end{proof}
\begin{corollary}\label{cor:loops-and-order}
  Let $\lambda$ and~$\lambda'$ be homotopic sparse embeddings of the same one-vertex graph~$\loopgraph$ on a surface with or without boundary.  Then the rotation systems of $\lambda$ and~$\lambda'$ are the same.
\end{corollary}
\begin{proof}
  Assume first that the surface has no boundary.  Note that the basepoint~$b$ may move during the homotopy between $\lambda$ and~$\lambda'$, following a (possibly non-simple) path~$p$.  By an ambient isotopy of the surface (a continuous family of self-homeomorphisms), starting from the embedding~$\lambda$, we push the basepoint~$b$ along path~$p$, obtaining an embedding~$\lambda''$ isotopic to~$\lambda$, and thus with the same rotation system as~$\lambda$, that is also homotopic to~$\lambda'$ \emph{with basepoint fixed}.  We can then apply Lemma~\ref{lem:loops-and-order} to $\lambda'$ and~$\lambda''$, obtaining the result.
\end{proof}

We now present our two main technical lemmas concerned with factorizations.  Here, our notations distinguish graphs from their drawings.

\begin{lemma}\label{lem:factoriz-end1}
On a surface with or without boundary, consider a factorization $(\loopgraph,\lambda,\gamma)$ of $(\graph,\drawing)$.  If $(\graph,\drawing)$ can be untangled, then $(\loopgraph,\lambda)$ can be untangled.
\end{lemma}

\begin{proof}
  Let $\drawing'$ be an embedding of~$\graph$ homotopic to~$\drawing$.  Starting from~$\drawing'$, we contract the edges of~$\graph$ that are part of the spanning trees used to define the factorization, obtaining a loop graph embedded on~$\surface$.  We then remove contractible loops and keep only one loop whenever several loops are in the same homotopy class.  In this way, we obtain an embedding~$\lambda'$ of~$\loopgraph$ in~$\surface$ that is homotopic to~$\lambda$.  Thus, $(\loopgraph,\lambda)$ can be untangled.
\end{proof}

\begin{lemma}\label{lem:factoriz-end2}
On a surface with or without boundary, consider a factorization $(\loopgraph,\lambda,\gamma)$ of $(\graph,\drawing)$.  If $\drawing$ and $\lambda$ are embeddings, then $(\graph,\lambda\circ\gamma)$ is a weak embedding.
\end{lemma}

See Figure~\ref{fig:example-2} for an illustration of this lemma.
\begin{figure}
\centering
\includegraphics[scale=0.8]{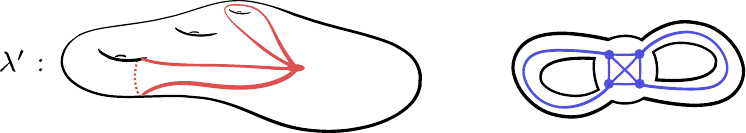}
\caption{(Continuation of Figures~\ref{fig:graph-plane} and~\ref{fig:example-1}) Although the map $\lambda$ is homotopic to an embedding $\lambda'$ (left), the map $\lambda' \circ \gamma$ is not a weak embedding (right), which implies by Lemma~\ref{lem:factoriz-end2} that $\delta$ is not homotopic to an embedding.}\label{fig:example-2}
\end{figure}

\begin{proof}
  We first remark that, without loss of generality, we can assume that $\graph$ is connected; this is because $\gamma$ maps the connected components of~$\graph$ to those of~$\loopgraph$ bijectively.

  Then, we prove that we can isotope~$(\graph,\drawing)$ so that its new image in~$\surface$ lies in the neighborhood of some sparsely embedded one-vertex graph.  For this purpose, we almost-contract the spanning tree~$\spanningtree$ used to define the factorization, keeping the fact that we have an embedding of~$\graph$.  The resulting loops fall into homotopy classes.  The loops in the trivial homotopy class can be isotoped close to the basepoint; indeed, in an embedded one-vertex graph, any contractible simple loop bounds a disk with only (possibly) contractible loops inside it~\cite[Theorem~1.7]{e-c2mi-66}.  The loops in any other given homotopy class can be bundled together so that they are parallel; indeed, in an embedded one-vertex graph, any pair of homotopic loops bounds a disk with only (possibly) contractible or homotopic loops inside it.  To summarize, we have an embedding~$\drawing'$ of~$\graph$ homotopic to~$\drawing$ whose image lies in a neighborhood of the image of some embedding~$\lambda'$ of a sparse loop graph~$\loopgraph'$.
  
  Actually, there is a canonical isomorphism between $\loopgraph$ and~$L'$, because each edge of $\loopgraph$ and~$L'$ corresponds to a particular set of non-tree edges of~$\graph$ (forming a single homotopy class after contracting~$\spanningtree$); so we can identify $\loopgraph$ and~$\loopgraph'$ via this isomorphism.  Also, for use at the end of the proof, we remark that $\lambda$ and~$\lambda'$ are (freely) homotopic.
  
  By construction, the embedding $\drawing'$ can be chosen to lie in any neighborhood of the map $\lambda'\circ\gamma$ (for the compact-open topology of maps from $\graph$ to $\surface$). Hence $\lambda'\circ\gamma$ is a weak embedding. Alternatively, it would now be easy to build the patch system of~$\lambda'(L')$ (see Section~\ref{sec:weak-embeddings}) associated to~$\lambda'$, so that $\drawing$ and this patch system indeed witness that $(\graph,\lambda'\circ\gamma)$ is a weak embedding.

  Now, we recall that $\lambda$ and~$\lambda'$ are homotopic embeddings of~$\loopgraph$, so the rotation systems are the same in $\lambda$ and~$\lambda'$ by Corollary~\ref{cor:loops-and-order}.  Also, as recalled in Section~\ref{sec:weak-embeddings}, the fact that $\lambda'\circ\gamma$ is a weak embedding depends only on the rotation system of~$(\loopgraph,\lambda')$, not on~$\lambda'$ itself.  This implies that $(\graph,\lambda\circ\gamma)$ is a weak embedding as well.
\end{proof}

\subsection{End of proof of Theorem~\ref{thm:main-theorem}}\label{sec:end-of-proof}

We are now in a position to give the algorithm for Theorem~\ref{thm:main-theorem}.  In a nutshell, we compute a factorization of~$\graph$, untangle the corresponding loop graph~$\loopgraph$ (or return that $\graph$ cannot be untangled), determine whether~$\graph$ is a weak embedding in (a tubular neighborhood of)~$\loopgraph$, and declare that $\graph$ can be untangled if and only if it is the case.

\begin{proof}[Proof of Theorem~\ref{thm:main-theorem}]
  Let $\rtriangulation$ be a reducing triangulation of a surface~$\surface$ of genus~$g$, let $\graph$ be a graph, and let $\drawing$ be a drawing of~$\graph$ on~$\rtriangulation$. Let us first explain how to determine whether there is an embedding homotopic to $\delta$ in $O(gn \log(gn))$ time:

  1.  We apply Proposition~\ref{prop:contraction-algorithm} (proved with the announced running time in the following section), that is: In $O(gn\log(gn))$ time, we either determine that $(\graph,\drawing)$ cannot be untangled on~$\surface$ (and thus abort), or compute a factorization of~$(\graph,\drawing)$ of width $O(g)$ and depth $O(n)$ given by the loop graph~$\loopgraph$, its sparse drawing~$\lambda$ on~$\surface$, and the drawing~$\gamma$ of~$\graph$ onto~$\loopgraph$, such that $\drawing$ and $\lambda\circ\gamma$ are homotopic.
  
  2.  We decide whether $(\loopgraph,\lambda)$ can be untangled on~$\surface$, using Proposition~\ref{prop:untangling-a-loop-graph-general}, in $O(gn\log(gn))$ time.  If it is not the case, we abort, because $\graph$ cannot be untangled, by Lemma~\ref{lem:factoriz-end1}.  Otherwise, $(\loopgraph,\lambda)$ can be untangled, and Proposition~\ref{prop:untangling-a-loop-graph-general} provides a weak embedding $\lambda' : \loopgraph \to T$ homotopic to $\lambda$, in which for every loop $\ell$ of $\loopgraph$ the walk $\lambda'(\ell)$ has length $O(n)$. Using the result of Akitaya, Fulek, and T\'oth~\cite{aft-rweg-19}, restated in Theorem~\ref{thm:toth-et-al} above, we can compute an embedding $\lambda''$ that approximates $\lambda'$, and read the rotation system of $\lambda''$, still in $O(gn\log(gn))$ time.  By Corollary~\ref{cor:loops-and-order}, and since $(\loopgraph,\lambda)$, and thus also $(\loopgraph,\lambda'')$, is sparse, we know that this rotation system is actually the same for all untanglings of~$(\loopgraph,\lambda)$.
  
  3.  Using Theorem~\ref{thm:toth-et-al} again, we determine whether $(\graph,\lambda''\circ\gamma)$ is a weak embedding on~$\surface$.  The rotation system is actually the information needed on~$\lambda''$ for the input of the algorithm of Akitaya et al. Thus, their algorithm determines whether $(\graph,\lambda''\circ\gamma)$ is a weak embedding in $O(gn\log(gn))$ time, because the input has size~$O(gn)$.

  4.  If $(\graph,\lambda''\circ\gamma)$ is a weak embedding, then the algorithm returns that $(\graph,\drawing)$ can be untangled; this is correct because $\drawing$ and $\lambda''\circ\gamma$ are homotopic on~$\surface$.  Otherwise, the algorithm returns that $(\graph,\drawing)$ cannot be untangled; this is correct by Lemma~\ref{lem:factoriz-end2}.  We have proved the correctness at every step, and by the preceding considerations, the running time is $O(gn\log(gn))$.
  
Assuming that there is an embedding homotopic to $\drawing$, we proved that the composed drawing $\delta' := \lambda' \circ \gamma$ is a weak embedding. And $\delta'$ can be computed in $O(n^2)$ time, because there are $O(n)$ edges to draw, each of length $O(n)$.
\end{proof}

\section{Efficient graph factorization: Proof of Proposition~\ref{prop:contraction-algorithm}}\label{sec:contraction}

This section is devoted to the proof of Proposition~\ref{prop:contraction-algorithm} with the claimed running time. In particular, we fix a reducing triangulation $T$ of an orientable surface $\surface$ of genus $g \geq 2$ without boundary. While a simple proof of Proposition~\ref{prop:contraction-algorithm} with a weaker complexity was presented in Section~\ref{S:factorization}, obtaining a near-linear complexity requires more data structures.

\subsection{Compressed homotopy trees}\label{sec:contract-datastruct}

Our main tool for the proof of Proposition~\ref{prop:contraction-algorithm} is a data structure, called \emph{compressed homotopy tree} extending the \emph{compressed homotopy sequence} from Section~\ref{sec:reduced-walks-reduction}.  While the compressed homotopy sequence is a linear list, the compressed homotopy tree is a tree-like structure.  We will also extend the proof techniques from that section.

Let $r$ be an arbitrary vertex of our reducing triangulation~$\rtriangulation$.  We need to support fast homotopy queries on an evolving set of walks starting at~$r$.  Intuitively, there is a map~$\kappa$ that associates a pointer to each walk starting at~$r$, such that two walks $W$ and~$W'$ are homotopic if and only if $\kappa(W)=\kappa(W')$.  We call these $\kappa(W)$ \emphdef{keys}.  Our data structure, a \emphdef{compressed homotopy tree}, provides efficient access to such a map~$\kappa$ in the following sense: If $W'$ is the concatenation of~$W$ with a single edge, then we can compute~$\kappa(W')$ from~$\kappa(W)$ quickly (in time logarithmic in the number of operations already performed).

\begin{figure}
\centering
\includegraphics[scale=1.7]{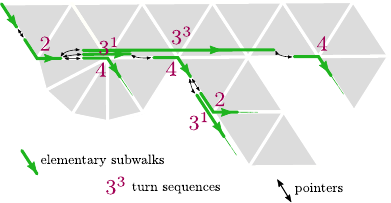}
\caption{Homotopy classes of walks are represented by storing the elementary subwalks of their reduced walks in a tree-like fashion.}\label{fig:data-structure}
\end{figure}

For this purpose, we remark that the homotopy class of a walk~$W$ is determined by the unique reduced walk~$\rho(W)$ homotopic to~$W$.  We store the reduced walks $\rho(W)$ in compressed form, by enhancing the data structure from Section~\ref{sec:reduced-walks-reduction} in a tree-like fashion, instead of a linear list (in a way similar to shortest path trees); for each walk~$W$ starting at~$r$, we let $\kappa(W)$ be a pointer to the last elementary subwalk of~$W$.  Each node of the tree corresponds to an elementary subwalk, and stores (see Figure~\ref{fig:data-structure}):
\begin{itemize}
  \item the first and last directed edges of~$\rtriangulation$ of the elementary subwalk;
  \item the turn sequence of the elementary subwalk, in compressed form (as in Section~\ref{sec:reduced-walks-reduction}: the elementary subwalk has a turn sequence of the form $3^k$, $(-3)^k$, or another symbol~$a$; we store its form together with the integer $k$ or the symbol~$a$);
  \item a pointer to its parent;
  \item the elementary subwalks of its children, stored in appropriate search trees (see below).
\end{itemize}
We maintain the following invariants:
\begin{itemize}
    \item the children of any given node in the tree correspond to pairwise distinct turn sequences;
    \item the first edge of an elementary subwalk equals the last edge of its parent.
\end{itemize}
By these invariants, and because reduced paths are unique, two walks starting at~$r$ are homotopic if and only if they have the same keys.  Thus, once the keys of two walks have been computed, we can, in $O(1)$ time, decide whether they are homotopic.

We now list a few operations that can be handled by our data structure.  Two of them are easy:
\begin{lemma}
  Compressed homotopy trees support the following operations:
  \begin{itemize}
  \item\textsc{TrivialKey}: return the key corresponding to the empty walk at~$r$, in $O(1)$ time;
  \item\textsc{ReducedWalk}: given the key of a walk~$W$, return the compressed turn sequence of the reduced walk homotopic to~$W$, in time linear in its length.
  \end{itemize}
\end{lemma}
\begin{proof}
  The \textsc{TrivialKey} operation is indeed trivial.

  The \textsc{ReducedWalk} operation can be done by traversing the tree from~$\kappa(W)$ up to the root, noting the labels of the nodes in order in which they are encountered, and then reading them in reverse order.
\end{proof}

More interestingly, we can also extend the set of walks considered:
\begin{lemma}
    Compressed homotopy trees support (with additional data structures described in the proof) the operation \textsc{Extend}: Given the key of a walk~$W$, finishing at vertex~$v$, and given a directed edge $e$ of~$\rtriangulation$ starting at~$v$, return the key of the concatenation of $W$ and~$e$, in $O(\log p)$ time, where $p$ is the number of \textsc{Extend} operations already performed.
\end{lemma}
\begin{proof}
  Because we must maintain the invariant that the children of any given node in the tree correspond to pairwise distinct turn sequences, we require one more ingredient to our data structure.  Given an elementary subwalk~$W$, we need to be able to find the child of~$W$ that has a given turn sequence in $O(\log p)$ time, or certify that such a child does not exist.  For this purpose, recall that each elementary subwalk has the turn sequence $3^k$ or $(-3)^k$ for some $k\ge1$, or $a$ for some $a\not\in\set{-3,3}$.  We split the children of~$W$ according to each of these three categories.  In a given category, a child of~$W$ is encoded by the nonnegative integer $k$ or $a$.  We use three red-black trees~\cite{gs-dfbt-78} (one for each category, indexed by either $k$ or $a$) to decide whether the child of~$W$ with a specified compressed turn sequence exists, or to insert it if it does not exist.  Each of these operations takes $O(\log p)$ time, because each red-black tree contains $O(p)$ elements.
  
  The rest of the proof reuses tools and arguments from Section~\ref{sec:reduced-walks-reduction}.  Let $W$ be a \emph{reduced} walk finishing at vertex~$v$; let $e$ be a directed edge starting at~$v$; let $W.e$ be the concatenation of~$W$ and~$e$; and let $W'$ be the reduced walk homotopic to~$W.e$.  We claim that the compressed homotopy sequences of $W$ and~$W'$ each have at most~49 elementary subwalks after their longest common prefix. To prove the claim it is enough to prove that the compressed homotopy sequences of $W.e$ and~$W'$ each have at most~48 symbols after their longest common prefix; indeed the compressed homotopy sequences of $W$ and $W.e$ do not differ by more than one final elementary subwalk.  Now recall from Section~\ref{sec:reduced-walks-reduction} that, for any walk~$W''$, $\varphi(W'')$ equals three times the length of~$W''$ plus the number of bad turns of~$W''$, and that $\varphi$ strictly decreases when performing a reduction move.  In the same spirit as the proof of Proposition~\ref{prop:trail-algo}, observe that $\varphi(W.e)\le\varphi(W)+4$.  Observe also that $\varphi(W)\le\varphi(W')+4$, since $W$ can be obtained from the concatenation of~$W'$ and the reversal of~$e$ by performing reduction moves, which decrease the value of~$\varphi$.  So $\varphi(W.e)\le\varphi(W')+8$.  Since $W'$ is the unique walk obtained by reducing~$W.e$, we have that any sequence of reductions on~$W.e$ has length at most eight.  Initially, only the last elementary subwalks of~$W.e$ can yield a reduction.  Each reduction replaces at most six consecutive elementary subwalks by at most seven ones, and then backtracks by at most six elementary subwalks (see the proof of Proposition~\ref{prop:trail-algo}).  Thus, only the last 48 elementary subwalks of~$W.e$ can be changed using successive reductions, which implies that the remaining first elementary subwalks of~$W.e$ are not affected when reducing to~$W'$. This proves the claim.

  Let $W$ and $e$ be as in the statement of the lemma; we can assume that $W$ is reduced.  The claim implies that we can compute $\kappa(W.e)$ from $\kappa(W)$ in $O(\log p)$ time, for example as follows.  Starting at the elementary subwalk representing~$W$, in $O(1)$ time we go up in the tree by 49 levels, reaching an elementary subwalk that represents a walk~$W_1$.  We can then compute, again in $O(1)$ time, the compressed homotopy sequence of~$W_2$, the walk obtained from~$W$ by removing its prefix~$W_1$ and by concatenating~$e$ to the result.  We reduce this compressed homotopy sequence in $O(1)$ time (see again the proof of Section~\ref{prop:trail-algo}); finally, using the first paragraph of this proof, we extend the elementary subwalk representing~$W$ with each elementary subwalk of this compressed homotopy sequence in the tree in turn, in $O(\log p)$ time, reaching eventually the elementary walk for~$W'$.
\end{proof}

Finally:
\begin{lemma}
  Compressed homotopy trees support (with additional data structures described in the proof) the operation \textsc{Partition}: Given a set~$X$ of instances of some data structure containing, in particular, a key, we can, in time linear in the size of~$X$, compute the partition of~$X$ such that two elements in~$X$ belong to the same part if and only if they contain the same key.
\end{lemma}
\begin{proof}
  This operation relies on easy bookkeeping techniques; in detail: Slightly extend the data structure of elementary subwalks with a pointer, initially NULL; for each element $x$ of~$X$ in turn, if its key refers to an elementary subwalk whose pointer is NULL, create a new part of the partition, initially containing only~$x$, and make the elementary subwalk refer to that part; otherwise, the key of~$x$ refers to an elementary subwalk~$W$ that has already been visited, and $x$ can be added to the part thanks to the pointer in~$W$.  After computing this partition, restore all pointers to NULL.
\end{proof}

One final detail: The root of the tree of our data structure must be handled in a special way.  By convention, it is a virtual elementary subwalk  whose only directed edge~$e$ does not belong to the triangulation~$\rtriangulation$ but points towards the root~$r$, so that by convention the turns from $e$ to each edge starting at~$r$ are $\infty$, with the obvious modifications on the algorithm to reduce walks.  We omit the tedious and trivial details.

\subsection{Algorithm}\label{sec:algo-contraction}

We can now give the algorithm for Proposition~\ref{prop:contraction-algorithm}.  The intuition is, again, to contract the edges of a spanning forest of~$\graph$ and to record the homotopy classes of the resulting loops.  This computes the sparse drawing~$\lambda$ of the loop graph~$\loopgraph$.  In a second step, we compute the drawing~$\gamma$ of~$\graph$ onto~$\loopgraph$.

\begin{proof}[Proof of Proposition~\ref{prop:contraction-algorithm}]
  By definition of a factorization, we can assume in the whole proof that $\graph$ is connected.  Let $r$ be an arbitrary root vertex of~$\graph$ and let $\spanningtree$ be a spanning tree of~$\graph$ obtained by a \emph{depth-first search} from~$r$.  Let us orient each edge $e$ of~$\graph\setminus \spanningtree$ so that its target $t(e)$ is an ancestor of its source $s(e)$.  Let $\ell_e$ be the loop that is the concatenation of the (unique) path in~$\spanningtree$ from $r$ to~$s(e)$, edge~$e$, and the (unique) path in~$\spanningtree$ from~$t(e)$ to~$r$.  Given a vertex~$v$ on the path from~$t(e)$ to~$r$, we let $\ell_e^v$ be the loop~$\ell_e$ where the last segment from $v$ to~$r$ is truncated (thus, $v$ always appears in $\ell_e^v$).

  In a first step, we compute the key of (the image in $T$ of) the path in~$\spanningtree$ from $r$ to each vertex of~$\graph$.  This can be done in $O(n\log n)$ time by traversing~$\spanningtree$ in top-down order and using the \textsc{Extend} operation $O(n)$ times.  Similarly, in the same amount of time, we can compute, for each edge $e$ of~$\graph\setminus \spanningtree$, the key $\kappa(\ell_e^{t(e)})$ of the image of~$\ell_e^{t(e)}$ in the triangulation~$\rtriangulation$. 

  In a second step, using the keys already computed, we compute the keys of the entire loops $\ell_e$, but to get an efficient algorithm we eliminate duplicate keys as we encounter them, and halt whenever there are too many distinct keys (i.e., too many homotopy classes, so that $\graph$ cannot be untangled).  We do this in a bottom-up way in the tree~$\spanningtree$.  For each vertex $v$ of~$\graph$, let $K(v)$ be the set of keys of the walks $\ell_e^v$, for each edge $e$ of~$\graph\setminus \spanningtree$ such that $v$ appears after~$e$ in $\ell_e$.

  Assume that we have computed the sets $K(v')$ for each child~$v'$ of a vertex~$v$; using the \textsc{Partition} operation, we can assume each $K(v')$ to be without duplicates.  If one of them has size larger than $12g$, then the \emph{directed} loops $\ell_e$ fall into at least $12g+1$ distinct homotopy classes, which, by a folklore result (see, e.g., \cite[Lemma~2.1]{ccelw-scsh-08}), implies that $\graph$ cannot be untangled, so we abort.  Otherwise, for each child $v'$ of~$v$ in~$T$ and for each key $k\in K(v')$, we apply the \textsc{Extend} function to $k$ and edge $v'v$.  Let $K(v',v)$ be the resulting set of keys.  We now observe that $K(v)$ is obtained by removing the duplicates in the union of the following sets of keys:
  \begin{itemize}
      \item the sets $K(v',v)$, for each child $v'$ of~$v$, and
      \item the set of keys of the form $\kappa(\ell_e^v)$, for each edge $e$ such that $t(e)=v$.
  \end{itemize}
  At the end of the recursion, unless the algorithm aborted, we have computed $K(r)$.  At each vertex~$v$, the number of times we apply \textsc{Extend} is $O(g\deg(v))$, where $\deg(v)$ is the degree of~$v$ in~$\spanningtree$.  So in total, \textsc{Extend} has been applied $O(gn)$ times, and the entire algorithm takes $O(gn\log (gn))$ time.

  When $K(r)$ has been computed, we can also compute the compressed turn sequence of the reduced walks of~$K(r)$ in time linear in their lengths, using the \textsc{ReducedWalk} operation. Actually, since $K(r)$ corresponds to $O(g)$ walks of length~$O(n)$ in~$\graph$, we may uncompress these turn sequences and return the explicit description of the $O(g)$ reduced walks in $O(gn)$ time, thus asymptotically without overhead in the running time.  We have thus computed the sparse drawing~$\lambda$ of the one-vertex graph~$\loopgraph$ on~$\surface$.

  Finally, remember that we also need to compute the drawing~$\gamma$ of~$\graph$ onto~$\loopgraph$.  In other words, we need to compute the mapping from each edge $e\in\graph\setminus \spanningtree$ to the loops of~$\loopgraph$ (or to its basepoint), or, equivalently, the key of~$\ell_e$ in~$K(r)$.  We now explain how to refine the second step above to achieve this.  In this second step, every key considered is of the form $k=\kappa(\ell_e^v)$; whenever we encounter such a key~$k$, we also store together with it a corresponding edge~$e$ such that $\kappa(\ell_e^v)=k$.  Whenever we detect redundancies in the current set of keys (see the end of Section~\ref{sec:contract-datastruct}), we actually compute the set of edges $e_1,\ldots,e_p\in\graph\setminus \spanningtree$, $p\ge1$, whose keys share a common value~$k$, which in particular implies $\kappa(\ell_{e_1})=\ldots=\kappa(\ell_{e_p})$.  In that case, we declare that the edge corresponding to the common key~$k$ is~$e_1$, declare that the \emph{leader} of $e_2,\ldots,e_p$ is~$e_1$, and then completely forget about the edges $e_2,\ldots,e_p$.  There is no overhead in the running time of the algorithm.  At the very end, once the keys in $K(r)$ have been computed (each of them coming with an associated edge in~$\graph\setminus \spanningtree$), we need to recover, for each edge $e\in\graph\setminus \spanningtree$, the corresponding key $\kappa(\ell_e)\in K_r$. For this purpose, we remark that the set of edges in $\graph\setminus \spanningtree$ is implicitly organized as a forest, in which the ``leader'' relation is actually a ``parent'' relation.  Thus, our problem boils down to this:  Given a forest, we need to compute, for each node, the root of its corresponding tree.  This can easily be done in time linear in the size of the forest, and thus, in our case, in $O(n)$ time.
\end{proof}

\section{Proof of Theorem~\ref{thm:surf} when $g\geq2$ and $b=0$}\label{sec:untangling-closed-ggeq2}

In this section we prove Theorem~\ref{thm:surf} in the case where our input surface~$\surface$ has genus at least two and is without boundary, using Theorem~\ref{thm:main-theorem}. To do so, we push our drawing into a reducing triangulation $T$, and we determine if the drawing in $T$ can be untangled. If required, we even compute a weak embedding in $T$, that we push back into the initial cellularly embedded graph $H$. The conversions between $T$ and $H$ can be described independently of the drawings, and are computed in the following lemma:

\begin{lemma}\label{L:conversion-closed-surface}
Let $H$ be a graph of size $m$ cellularly embedded on a surface $\surface$ of genus $g \geq 2$ without boundary. In $O(gm)$ time, we can construct a reducing triangulation $T$ of $S$, and a drawing $\varphi : H \to T$ homotopic to the inclusion map $H \subset \surface$, such that $\varphi$ maps each edge of $H$ to a walk of length $O(g)$ in $T$. Moreover, in $O(g^2m)$ time, we can push the 1-skeleton of $T$ arbitrarily close to the 1-skeleton of $H$ by isotopy, and compute the resulting weak embedding $\psi : T \to H$, such that $\psi$ maps each edge of $T$ to a walk of length $O(gm)$ in $H$.
\end{lemma}

\begin{figure}
    \centering
    \def\svgwidth{40em}
    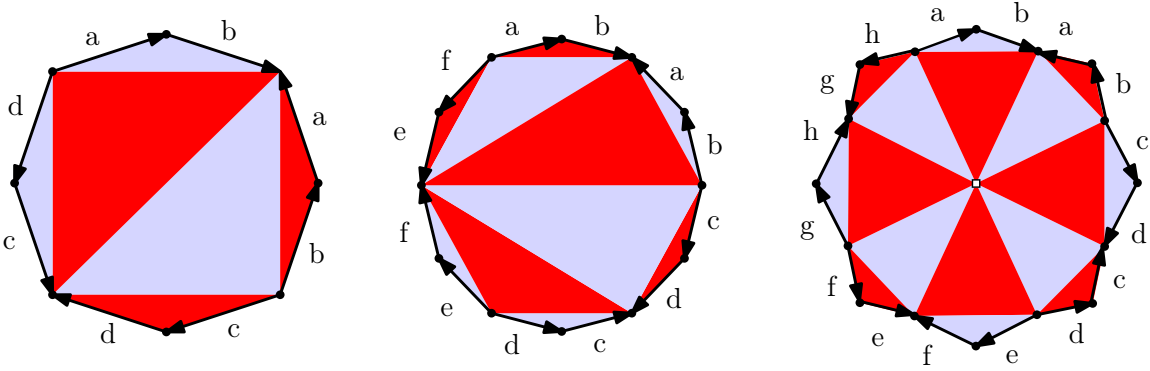
    \caption{How to build a reducing triangulation from a canonical polygonal schema of~$\surface$.  The cases $g=2$ (left) and $g=3$ (middle) use special constructions; the case $g\ge4$ (right) trivially generalizes to higher genus.  In all cases, there is a walk of length $O(1)$ between any two corners of the polygonal schema.}
    \label{fig:valid-tr-from-canonical-sc}
\end{figure}

A \emph{canonical system of loops}~$K$ of~$\surface$ is a set of pairwise disjoint simple loops with a common basepoint~$b$, such that, when cutting $\surface$ along~$K$, we obtain a \emph{canonical polygonal schema}, a $4g$-gon whose boundary reads $a_1b_1\bar a_1\bar b_1\ldots a_gb_g\bar a_g\bar b_g$ in this order, where $a_1,\ldots, a_g, b_1,\ldots, b_g$ denote the loops in~$K$, and bar denotes reversal.

\begin{figure}[ht]
    \centering
    \includegraphics[scale=0.8]{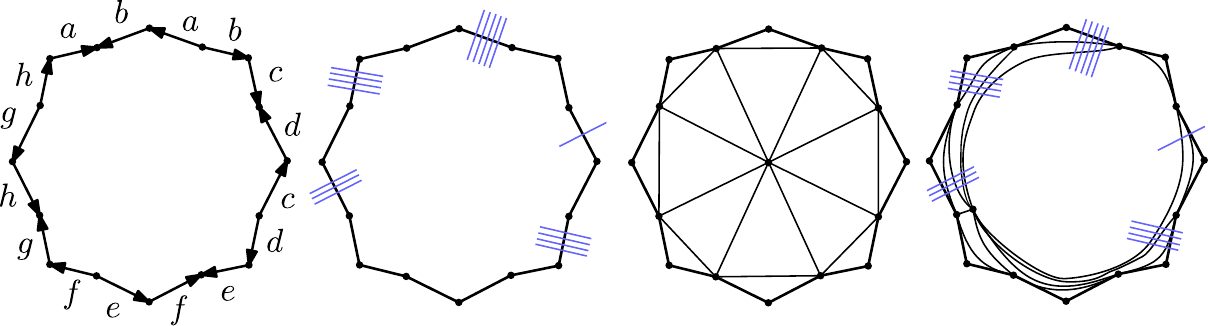}
    \caption{From left to right: (a) The canonical polygonal schema associated to the canonical system of loops $K$ for the surface $\surface$ of genus $g=4$. (b) A blue graph $Q$ on $\surface$, in general position with respect to $K$, that intersects every edge of $K$ at most $m$ times. (c) A reducing triangulation $T$ whose 1-skeleton contains $K$. (d) An embedding of $T$ on $\surface$ such that each edge of $T$ crosses $Q$ at most $O(gm)$ times.}
    \label{fig:embed-redux}
\end{figure}

\begin{figure}[ht]
    \centering
    \includegraphics[scale=0.8]{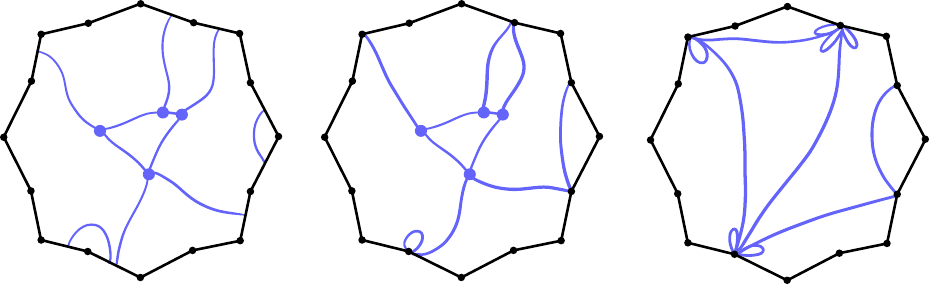}
    \caption{(Left) In the proof of Lemma~\ref{L:conversion-closed-surface}, the graph $Q$ is represented in blue in the face of $K$. (From Left to Middle) The intersection points between $Q$ and $K$ are slided along the edges of $K$. (From Middle to Right) Some subedges of $Q$ are contracted to push the vertices of $Q$ to the vertex of $K$.}
    \label{fig:push}
\end{figure}
\begin{figure}[ht]
    \centering
    \includegraphics[scale=0.15]{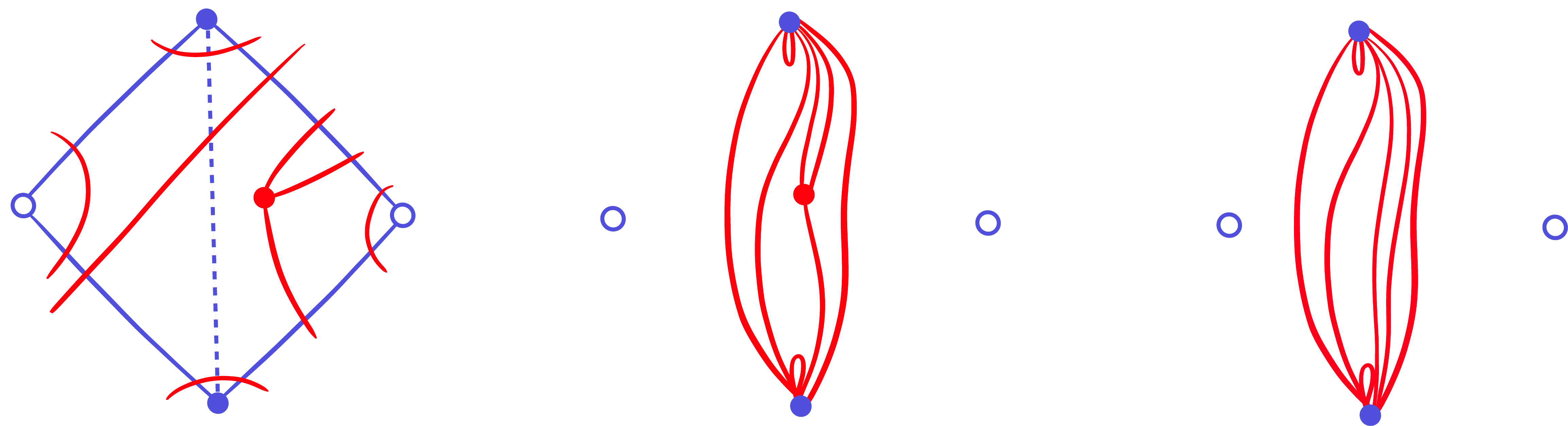}
    \caption{(Left) In the proof of Lemma~\ref{L:conversion-closed-surface}, a portion of $T$ is here represented in red in some face of $Q$. The two blue disk vertices belong to $H$. The two blue circle vertices were inserted in faces of $H$ to build $Q$, they are dual vertices. The dashed edge was deleted from $H$ to build $Q$. The plain edges belong to $Q$. (From Left to Middle) In the overlay $A$ between $T$ and $Q$, every edge of $Q$ is detached from its incident dual vertex, then contracted. (From Middle to Right) Some subedge incident to the red disk vertex of $T$ is contracted.}
    \label{fig:contractions}
\end{figure}

\begin{proof}
Let $Q$ be obtained from $H$ by inserting a new vertex in every face, by adding an edge between this vertex and every corner of the face, and by removing from $H$ all its initial edges. Then $Q$ is a quadrangulation of size $O(m)$.  Let $K$ be the canonical system of loops for the surface of genus $g$ without boundary. Using a result by Pocchiola, Lazarus, Vegter, and Verroust~\cite{lpvv-ccpso-01}, we compute in $O(gm)$ time an embedding of $K$ in general position with respect to $Q$, such that each loop of $K$ crosses each edge of $Q$ at most four times. Then we consider the reducing triangulation $T$ for the surface of genus $g$ without boundary depicted in Figure~\ref{fig:valid-tr-from-canonical-sc}. This triangulation $T$ extends $K$, so the embedding of $K$ extends to an embedding of $T$ in general position with respect to $Q$, in which each edge of $T$ crosses $Q$ at most $O(gm)$ times, see Figure~\ref{fig:embed-redux}. While we already computed the overlay of $K$ and $Q$, we do not compute the overlay of $T$ and $Q$ yet, as this overlay may have size $\Omega(g^2m)$. 

Let us first explain how to push the 1-skeleton of $H$ by homotopy into the 1-skeleton of $T$, and retrieve the resulting drawing $\varphi : H \to T$, all in $O(gm)$ time. As a preliminary, we push every edge $e$ of $H$ to a walk of length $2$ in $Q$. We shall now explain how to push $Q$ into the 1-skeleton of $T$ by homotopy, in a way that will send every edge of $Q$ to a walk of length $O(g)$ in $T$.  Consider the overlay between $K$ and $Q$.  For each loop of~$K$ crossed $k$ times by~$Q$, we contract $k-1$ such subedges, and then we contract some subedges of~$Q$ in order to bring every vertex of~$Q$ to the basepoint of~$K$ (Figure~\ref{fig:push}).  Let $Q'$ be the resulting contracted version of~$Q$; now, we have the overlay of~$Q'$ and of~$K$, two one-vertex graphs sharing the same vertex. The effect is that each edge of $Q'$ is transformed into an ordered set of $O(g)$ pairwise non-crossing arcs in the polygonal schema, each connecting two corners. Finally, we push each of those arcs into a path of length $O(1)$ in $T$.

Let us now explain how to push the 1-skeleton of $T$ arbitrarily close to $H$, and retrieve the resulting weak embedding $\psi : T \to H$, all in $O(g^2m)$ time. To do so, we forget about the modifications of the previous paragraph, and we compute the overlay $A$ of $T$ and~$Q$ (Figure~\ref{fig:embed-redux}, Right) in $O(g^2m)$ time. Then, we modify $A$ in two steps. See Figure~\ref{fig:contractions}. First, recall that by construction, every edge $e$ of $Q$ is between a vertex $v$ of $H$ and a dual vertex $w$ that was inserted in a face of $H$ when building $Q$.  In the overlay $A$, we bring the $k$ intersection points between $e$ and $T$ close to $v$, by contracting the $k-1$ subedges of $e$ that are not incident to $w$. Second, we contract some subedges of~$T$ in $A$ to put every vertex of $T$ close to a vertex of $H$. In the end, every edge $e$ of $T$ is close to a walk of length $O(gm)$ in $H$, so we can retrieve the desired weak embedding of the 1-skeleton of $T$ into $H$. 
\end{proof}

Given a graph~$\embeddedgraph$ embedded on~$\surface$, a \emphdef{monogon} is face of~$\embeddedgraph$ that is a disk bounded by a single edge of~$\embeddedgraph$ (which is thus a loop).  A \emphdef{bigon} is a face of~$\embeddedgraph$ that is a disk bounded by two edges of~$\embeddedgraph$ (which are thus parallel edges).

\begin{proof}[Proof of Theorem~\ref{thm:surf} in the case $g\ge2$, $b=0$]
First assume that we are only interested in determining whether there is an embedding of~$\graph$ on~$\surface$ homotopic to~$\drawing$. In that case, we claim that we can assume without loss of generality that $\embeddedgraph$ is a set of $O(g)$ loops with a common basepoint.  Indeed, we can, in $O(m)$ time, contract an arbitrary spanning tree of~$\embeddedgraph$ and transform the input drawing~$\drawing$ of~$G$ homotopically into a drawing of~$G$ in that new graph.  Then we can iteratively remove edges forming monogons, and merge together edges forming bigons, all in $O(m)$ time.  Euler's formula and double counting of the edge-face incidences then implies that $\embeddedgraph$ has $O(g)$ loops. This proves the claim. We use the claim immediately, and apply Lemma~\ref{L:conversion-closed-surface} to compute in $O(g^2)$ time a reducing triangulation $T$ of $\surface$, and a drawing $\varphi : H \to T$, homotopic to the inclusion map $H \subset \surface$, such that $\varphi$ maps each edge of $H$ to a walk of length $O(g)$ in $T$. The drawing $\varphi \circ \drawing$ is homotopic to $\drawing$, and has size $O(gn)$. We apply Theorem~\ref{thm:main-theorem} to determine in $O(g^2n \log(gn))$ time whether $\varphi \circ \drawing$ can be untangled.

Now, assuming that there is an embedding homotopic to $\delta$, we describe how to compute the desired weak embedding $\delta' : G \to H$. (This time we do not assume that $H$ is a set of loops.) First, we apply Lemma~\ref{L:conversion-closed-surface} to compute in $O(g^2m)$ time the reducing triangulation $T$ of $\surface$ and the drawing $\varphi : H \to T$, but also, and this is new, a weak embedding $\psi : T \to H$, obtained pushing the 1-skeleton of $T$ arbitrarily close to the 1-skeleton of $H$ by isotopy, that maps each edge of $T$ to a walk of length $O(gm)$ in $H$. The drawing $\varphi \circ \drawing$ is homotopic to $\drawing$, and has size $O(gn)$. We apply Theorem~\ref{thm:main-theorem} to construct in $O(g^2(n^2 + n \log(gn)))$ time a weak embedding $\delta'' : G \to T$, homotopic to $\varphi \circ \delta$, such that $\delta''$ maps each edge of $G$ to a walk of length $O(gn)$ in $T$. Then we return $\delta' := \psi \circ \delta''$, which is as desired.
\end{proof}
}

\section{Proof of Theorem~\ref{thm:surf} for the remaining cases}\label{sec:torus-and-boundary}

In this section we solve our untangling problem for the (orientable) surfaces not handled in the previous section, namely when $\surface$ is the torus or has non-empty boundary.

Recall that the entire algorithm for surfaces of genus at least two works as follows.  In Step~(0), we transform the input graph~$\embeddedgraph$ into a reducing triangulation (Section~\ref{sec:untangling-closed-ggeq2}).  Then we apply the four steps in the proof of Theorem~\ref{thm:main-theorem} (Section~\ref{sec:end-of-proof}): (1) We compute a factorization~$(\loopgraph,\lambda,\gamma)$ of our graph drawing~$(\graph,\drawing)$ (or immediately return that $(\graph,\drawing)$ cannot be untangled); (2) we untangle $(\loopgraph,\lambda)$, obtaining an embedding~$(\loopgraph,\lambda')$ (or return that $(\graph,\drawing)$ cannot be untangled); (3) we decide whether $(\graph,\lambda'\circ\gamma)$ is a weak embedding; (4) we return the answer accordingly.

We follow the same general strategy.  The arguments are still valid because they rely on auxiliary lemmas (Section~\ref{sec:auxiliary-lemmas}) valid for arbitrary surfaces.  Three ingredients must be adjusted: In Step~(0), we must transform the input graph; in Step~(1), we must compute a factorization; in Step~(2), we must untangle a loop graph.

\subsection{Untangling on the torus}

We now prove Theorem~\ref{thm:surf} when the input surface~$\surface$ is the torus.

The proof relies on three lemmas, for Steps (0), (1), and~(2), respectively.  We remark that we cannot apply the exact same strategy as in the case of genus at least two, because the torus admits neither a hyperbolic metric (by the Gauss--Bonnet theorem) nor a reducing triangulation (by Euler's formula).  On the other hand, we can replace it with a simpler structure because of the Abelian structure on the homotopy group of the torus, see below. A \emph{canonical system of loops}~$K$ of the torus is a set of two disjoint simple loops $(k_1,k_2)$with common basepoint that cross at the basepoint.  Similar to Lemma~\ref{L:conversion-closed-surface}, we have the following lemma, the proof of which is strictly easier than Lemma~\ref{L:conversion-closed-surface}, and thus omitted.

\begin{lemma}\label{lem:canon-torus}
  Let $\embeddedgraph$ be a graph of size~$m$ cellularly embedded on a torus~$\surface$. In $O(m)$ time, we can construct a canonical system of loops~$K$ of $\surface$, and a drawing $\varphi : H \to K$ homotopic to the inclusion map $H \subset \surface$, such that $\varphi$ maps each edge of $H$ to a walk of length $O(1)$ in $K$. Moreover, in $O(m)$ time, we can push the 1-skeleton of $K$ arbitrarily close to the 1-skeleton of $H$ by isotopy, and compute the resulting weak embedding $\psi : K \to H$, such that $\psi$ maps each edge of $K$ to a walk of length $O(m)$ in $H$.
\end{lemma}

The homotopy group on the torus is isomorphic to~$\cZ^2$; this means that every closed walk in~$K$ is homotopic to $k_1^{u_1}\cdot k_2^{u_2}$ (where ``$\cdot$'' denotes concatenation) for a unique $(u_1,u_2)\in\cZ^2$, and we represent homotopy classes by such vectors.  Moreover, walk concatenation translates, in homotopy, to vector addition, and free homotopy classes can also be represented by such vectors.  We say that a walk in~$K$ is \emphdef{reduced} if it is of the form $k_1^{u_1}\cdot k_2^{u_2}$, for some integers $u_1$ and~$u_2$.  Every walk in~$K$ is homotopic to a unique reduced walk; we can reduce walks in linear time. 

\begin{lemma}\label{lem:contract-torus}
Let $K$ be a canonical system of loops of the torus. Let $G$ be graph, and let $\delta : G \to K$ be a drawing of size $n$. We can, in time $O(n \log n)$, compute a factorization of $(\graph, \drawing)$, of size $O(n)$, or correctly report that $\drawing$ cannot be untangled.
\end{lemma}

\begin{proof}
  The algorithm is an easy variation on the one described in Section~\ref{sec:algo-contraction}. The only change concerns the data structure that stores homotopy classes of walks in~$K$, described in Section~\ref{sec:contract-datastruct} for reducing triangulations.  Here we replace the reduced walks in a reducing triangulation by reduced walks in~$K$.  The compressed homotopy tree structure is replaced by a two-level tree-like structure, the first level for the integer~$u_1$ and the second level for the integer~$u_2$ in the notations above.  The four operations introduced in Section~\ref{sec:contract-datastruct} extend to this context, with the same complexity.  Section~\ref{sec:algo-contraction} extends immediately, because it only uses these four operations.
\end{proof}

The following adapts classical arguments, following~\cite[Lemma~4.5]{dubois2024making}:

\begin{lemma}\label{lem:untangling-loops-torus}
Let $K$ be a canonical system of loops of the torus. Let $\loopgraph$ be a loop graph and $\lambda$ be a sparse drawing of $\loopgraph$ in $K$, of size $n$. We can decide whether $\lambda$ can be untangled in $O(n \log n)$ time. If so, we can compute a weak embedding $\lambda' : \loopgraph \to K$ of size $O(n)$, homotopic to $\lambda$.
\end{lemma}

\begin{proof}
We claim that we can assume, without loss of generality, that $\loopgraph$ is connected and contains at most three loops.  Indeed, if $(\loopgraph,\lambda)$ can be untangled, then each connected component of~$\loopgraph$ has at most three loops (by Euler's formula) and if it has at least two loops, then by sparsity every embedding of that connected component cuts the torus into a disk.  Thus, either the conclusion of the claim holds, or $\loopgraph$ has several connected components, each made of a single loop.  These loops must be pairwise freely homotopic, because otherwise they cannot be untangled; but then $(\loopgraph,\lambda)$ can be untangled if and only if its restriction to a single connected component of~$\loopgraph$ can be untangled, and a weak embedding of this single loop would immediately provide a weak embedding of $\loopgraph$ (and in that case the rotation systems are trivial).  All this takes linear time.  This proves the claim.


Let $K^*$ be the canonical system of loops dual to $K$ on the torus $\surface$. Let $k_1, k_2$ be the two loops of $K$, and let $k_1^*, k_2^*$ be their respective dual loops. We identify $\surface$ with the quotient $\cR^2 / \cZ^2$, such that $K^*$ lifts to the following grid: the vertex of $K^*$ lifts to $\cZ^2$, the loop $k_1^*$ lifts to the vertical segments between $(i,j)$ and $(i,j+1)$ for $i,j \in \cZ$, and the loop $k_2^*$ lifts to the horizontal segments between $(i,j)$ and $(i+1,j)$ for $i,j \in \cZ$; the loops $k_1$ and $k_2$ are oriented so that, in $\cR^2$, the lifts of $k_1$ cross the lifts of $k_1^*$ from left to right, and the lifts of $k_2$ cross the lifts of $k_2^*$ from bottom to top.  In this way, $\surface$ is endowed with a flat metric.

Now, we define a map $f:L\to\surface$, homotopic to~$\lambda$, by representing the loops of~$L$ as geodesics in this metric.  For this purpose, we first select arbitrarily the image of the basepoint of~$L$ under~$f$; let $p$ be one of its lifts in~$\mathbb R^2$.  The images of any loop~$\ell$ of~$L$ under~$f$ are then uniquely determined: Indeed, if the reduced walk associated to~$\lambda \circ \ell$ is $k_1^{u_1} \cdot k_2^{u_2}$, then $f \circ \ell$ must lift to a line segment from $p$ to $p+(u_1,u_2)$.  Finally, we can ensure that $f$ does not intersect the vertex of~$K^*$, by slightly perturbing~$p$ if necessary.
Note that $f$ actually maps each edge of $L$ to a geodesic closed curve (not only a geodesic path).  Now for every edge $\ell$ of $L$, the path $f \circ \ell$ crosses $O(n)$ edges of $K^*$.  We can compute the sequence of crossings of~$f \circ \ell$ with the edges of~$K^*$ in $O(n)$ time: Indeed, $f \circ \ell$ lifts to a line segment in~$\mathbb R^2$ that does not intersect the integer points, and given its endpoints, we can compute the sequence of horizontal and vertical edges of the grid crossed by the line segment, in order, with the standard Bresenham's line algorithm~\cite{5388473}, which uses integer arithmetic only.  This sequence of crossings of $f \circ \ell$ with~$K^*$ gives, by duality, a walk in $W_\ell$ in $K$. Let $\lambda' : L \to K$ be the drawing that maps each edge $\ell$ of $L$ to its walk $W_\ell$.  Finally, we determine whether $\lambda'$ is a weak embedding, in $O(n \log n)$ time, with Theorem~\ref{thm:toth-et-al}; if it is the case, we return $\lambda'$; otherwise, we return that $\lambda$ cannot be untangled.

We now prove correctness.  If $\lambda'$ is a weak embedding, then the answer is obviously correct.  Conversely, if $\lambda$ (and thus $f$) can be untangled, we claim that $f$ is an embedding.  It is then immediate that $\lambda'$ is a weak embedding, since $f$ approximates $\lambda'$ in the patch system of~$K$, which is the surface obtained from $\surface$ by removing the vertex of $K^*$.

There remains to prove the claim.  For this purpose, let us assume that $f$ can be untangled. Recall (see, e.g., Farb and Margalit~\cite[Section~1.2.3]{fm-pmcg-12}) that on the torus $\surface$, if a non-contractible geodesic closed curve $\gamma$ is homotopic to a simple closed curve, then $\gamma$ itself is simple. Also any two non-contractible non-homotopic simple geodesic closed curves $\gamma_0$ and $\gamma_1$ cross minimally, up to (free) homotopy. Thus every loop $\ell$ of $\loopgraph$ is mapped to a simple closed curve by $f$, by sparsity and since the image curve of $\ell$ is simple in an embedding homotopic to $f$. Also for any two distinct loops $\ell_1$ and $\ell_2$ of $\loopgraph$, the closed curves $f \circ \ell_1$ and $f \circ \ell_2$ cross only once (at the vertex of $L$), by sparsity and since the image curves of $\ell_1$ and $\ell_2$ cross only once in an embedding homotopic to $f$. That proves that $f$ is an embedding.
\end{proof}

The proof of Theorem~\ref{thm:surf} in the case of the torus is now straightforward.

\begin{proof}[Proof of Theorem~\ref{thm:surf} for the torus]
We apply Lemma~\ref{lem:canon-torus} to construct in $O(m)$ time a canonical system of loops~$K$ of $\surface$, and a drawing $\varphi : H \to K$ homotopic to the inclusion map $H \subset \surface$, such that $\varphi$ maps each edge of $H$ to a walk of length $O(1)$ in $K$. 

We apply the algorithm of the proof of Theorem~\ref{thm:main-theorem} in Section~\ref{sec:end-of-proof} to the drawing $\varphi \circ \delta$, but where Lemma~\ref{lem:contract-torus} replaces Proposition~\ref{prop:contraction-algorithm}, and Lemma~\ref{lem:untangling-loops-torus} replaces Proposition~\ref{prop:untangling-a-loop-graph-general}. That is, in time $O(n \log n)$, we determine whether $\varphi \circ \delta$ (and thus $\delta$) can be untangled, and if it is the case, in $O(n^2)$ time, we compute a weak embedding $\delta'' : G \to K$, homotopic to $\varphi \circ \delta$ (and thus to $\delta$), that maps each edge of $G$ to a walk of length $O(n)$ in $K$.

Now Lemma~\ref{lem:canon-torus} pushes the 1-skeleton of $K$ arbitrarily close to the 1-skeleton of $H$ by isotopy, and provides the resulting weak embedding $\psi : K \to H$, such that $\psi$ maps each edge of $K$ to a walk of length $O(m)$ in $H$. We compute $\delta' := \psi \circ \delta''$ in $O(mn^2)$ time, which is as desired.
\end{proof}

\subsection{Untangling on surfaces with non-empty boundary}\label{sec:boundary}

In this section, we prove Theorem~\ref{thm:surf} in the case where our input surface~$\surface$ has non-empty boundary.  An obvious strategy is to attach a handle to each boundary component and to apply Theorem~\ref{thm:surf} to the resulting surface without boundary.  While this is certainly a valid approach, we show here that we can circumvent a large part of the technical machinery, even though this does not improve the asymptotic running time.

A \emphdef{loop system} of~$\surface$ is a set~$Y$ of pairwise disjoint simple loops with a common basepoint on~$\surface$ such that each face of~$Y$ has genus zero and contains a single component of the boundary of~$\surface$.

\begin{lemma}\label{lem:canon-bd}
  Let $\embeddedgraph$ be a graph of size~$m$ cellularly embedded on a surface~$\surface$ of genus $g \geq 0$ with $b \geq 1$ boundary components. In $O((g+b)m)$ time, we can construct a loop system~$Y$ of $\surface$, and a drawing $\varphi : H \to Y$, homotopic to the inclusion map $Y \subset \surface$, such that $\varphi$ maps each edge of $H$ to a walk of length $O(g+b)$ in $Y$. Moreover, in the same amount of time, we can push the 1-skeleton of $Y$ arbitrarily close to the 1-skeleton of $H$ by isotopy, and compute the resulting weak embedding $\psi : Y \to H$, such that $\psi$ maps each edge of $Y$ to a walk of length $O(m)$ in $H$.
\end{lemma}

\begin{proof}
 Let us first explain how to construct the loop system $Y$ from $H$.  First, we contract a spanning tree in~$H$, so that now $H$ has only one vertex.  Second, as long as an edge of~$H$ is incident to a face containing a boundary component of~$\surface$ \emph{and} to a face containing no boundary component of~$\surface$, we remove it.  It is easy to achieve this in linear time, e.g., by using a spanning forest of the faces of~$H$ in which the ``seeds'' are the faces containing a boundary component. We constructed $Y$ from $H$ by contracting a spanning tree, and then by deleting some edges, so we have access to the desired weak embedding $\psi : Y \to H$ at no extra cost.

  Let us now explain how to push the 1-skeleton of $H$ into the 1-skeleton of $Y$ by homotopy, and retrieve the resulting drawing $\varphi : H \to Y$.  First, as above, we contract the edges of a spanning tree of~$H$; the contracted edges are mapped to the basepoint of~$Y$.  Let us now consider the other edges of~$H$.  Any such edge that is not in~$Y$ belongs to a unique face of~$Y$, and can be rerouted homotopically in $\varphi$ to a subpath of that face, on the side of that face that does not contain the puncture.  By Euler's formula, $Y$ contains $O(g+b)$ loops, so the faces have size $O(g+b)$ in total.  Thus, $\varphi$ maps each edge of $H$ to a walk of length $O(g+b)$ in $Y$.
\end{proof}

We remark that $\surface$ is homeomorphic to the patch system~$\Sigma$ of~$Y$, and the arcs of~$\Sigma$ correspond to the dual of~$Y$.  Since $\surface$ is a surface with non-empty boundary, the homotopy group of~$\surface$ is a free group; see, e.g., Stillwell~\cite[Chapter~2]{s-ctcgt-93}.  This implies two facts that we shall use later: (1) The homotopy classes of paths or loops in~$\surface$ correspond to sequences of crossings with the arcs of~$\Sigma$ that are \emph{reduced} (do not contain two consecutive crossings with the same arc in opposite directions); (2) the free homotopy classes of closed curves on~$\surface$ correspond to \emph{cyclically reduced} sequences of crossings.

(Such properties are folklore, and follow from their analogs for the homotopy of paths and closed walks in graphs.  It is clear that two paths with the same reduced sequence of crossings are homotopic; conversely, any two homotopic paths lift to paths with the same endpoints in the universal cover, and their sequences of crossings, once reduced by homotopies, must be the same because $\Sigma$ lifts to pairwise disjoint simple arcs, each of which separates the universal cover into two pieces.  A similar argument holds for closed walks, where the homotopy between them must be lifted.)

\begin{lemma}\label{lem:contract-boundary}
Let $Y$ be a loop system of size $m$ in a surface $\surface$ with non-empty boundary. Let $\graph$ be a graph and $\drawing$ be a drawing of $\graph$ in $Y$, of size $n$. In $O(mn \log(mn))$ time, we can compute a factorization of $(\graph, \drawing)$ of width $O(m)$, and depth $O(n)$,  or correctly report that $\drawing$ cannot be untangled.
\end{lemma}
\begin{proof}
The algorithm is the one described in Section~\ref{sec:algo-contraction} for the proof of Proposition~\ref{prop:contraction-algorithm}. In the same way as for the torus the only change concerns the data structure that stores homotopy classes of walks. Here we replace the reduced walks in a reducing triangulation by reduced walks in $Y$. That is, we store these walks as lists of directed edges of $Y$ in a tree-like fashion to implement the data-structure of Section~\ref{sec:contract-datastruct}.   The four operations introduced in Section~\ref{sec:contract-datastruct} extend to this context, with the same complexity.  Section~\ref{sec:algo-contraction} extends immediately, because it only uses these four operations.
\end{proof}

From the above factorization, we obtain a drawing~$\lambda$ of a loop graph~$\loopgraph$ on~$\patchsystem$.  As in Section~\ref{sec:untangle-loop-graph}, we choose an arbitrary edge of each connected component of~$\loopgraph$ and declare it to be a \emphdef{major edge}; the other edges are \emphdef{minor edges}.  We say that $\lambda$ is \emphdef{straightened} if, under~$\lambda$, (1) the sequence of crossings of each major edge with the arcs of~$\patchsystem$ is cyclically reduced, and (2) the sequence of crossings of each minor edge with the arcs of~$\patchsystem$ is reduced (possibly not cyclically).  The analog of Proposition~\ref{prop:untangle-loop-graph} becomes: 

\begin{lemma}\label{lem:untangling-loops-boundary}
Let $\loopgraph$ be a loop graph and let $\lambda : \loopgraph \to \patchsystem$ be a drawing. Assume that $\lambda$ is sparse and straightened. If there is an embedding homotopic to $\lambda$, then $\lambda$ is a weak embedding.
\end{lemma}

\begin{proof}
Since $\lambda$ is straightened, $\lambda(e)$ is cyclically reduced for every major edge $e$ of $L$, and $\lambda(e')$ is reduced for every minor edge $e'$ of $\loopgraph$. There is by assumption an embedding $\lambda' : \loopgraph \to \Sigma$ freely homotopic to~$\lambda$.

In a first step, we modify the embedding~$\lambda'$ by an ambient isotopy that fixes the boundary of~$\patchsystem$ so that the major edges of $L$ are cyclically reduced (with respect to the arcs of~$\patchsystem$).  For this purpose, we remark that the major edges are, in~$\lambda'$, pairwise disjoint simple closed curves.  Whenever there is a bigon between the image of a major edge and an arc of~$\patchsystem$, there exists an innermost bigon, which we can remove by an ambient isotopy; the number of crossings with the arcs decreases.  We repeat this operation until there is no more bigon, at which point the major edges are all cyclically reduced in~$\lambda'$.

Now consider a connected component $\loopgraph_0$ of $\loopgraph$. We can make the major edge $e$ of $\loopgraph_0$ cross the arcs of $\Sigma$ with the same sequences in $\lambda'$ and in $\lambda$ (not up to cyclic permutation, \emph{exactly} the same sequence) simply by sliding the image of the basepoint $v$ of $\loopgraph_0$ along the image of $e$ in $\lambda'$. We slide $v$ by an ambient isotopy in the tubular neighborhood of $e$. We can do so since the sequence of crossings of $\lambda'(e)$ is a cyclic permutation of the one of $\lambda(e)$; here we make use of the fact (see above) that if two freely homotopic closed curves are cyclically reduced, then their sequences of crossings with the arcs of $\Sigma$ are equal up to cyclic permutation. In $\lambda'$ we slide the image of $v$ along the image of $e$ so that the two sequences become equal; $\lambda'$ is still an embedding.  We do this for every connected component~$\loopgraph_0$ of~$\loopgraph$ in turn.

Consider again some connected component~$\loopgraph_0$ of~$\loopgraph$, with vertex~$v$.  We claim that we can modify $\lambda'$ by sliding the image of $v$ some finite number of times around the image of the major edge~$e$ of~$\loopgraph_0$ (each time, the image of $v$ making one full loop around the image of $e$) so that, in the end, there is a free homotopy between $\lambda'(\loopgraph_0)$ and $\lambda(\loopgraph_0)$ in which the image of $v$ does not leave its face of the patch system $\patchsystem$. To prove this claim first observe that it would be possible to do so, not by sliding along $e$, but by some \emph{free} homotopy of $\lambda'(\loopgraph_0)$. During this homotopy the image of $v$ makes a loop $\ell$ in $\Sigma$. The loop $\ell$ commutes, up to homotopy, with the loop $\lambda'(e)$ as $\lambda'(e) \simeq \lambda(e)$ by the previous paragraph and $\lambda(e) \simeq \ell \lambda'(e) \ell^{-1}$ by construction, where $\simeq$ denotes homotopy of loops relatively to their basepoint. Moreover, the loop $\lambda'(e)$ is non-contractible since $\lambda'$ is sparse, so it is also primitive by Lemma~\ref{L:primitive}. Thus $\ell$ is a power (up to homotopy) of $\lambda'(e)$.  Let us prove this.  The fundamental group of $\Sigma$ is a free group.  It is known (and we prove) that in a free group if two elements $x$ and~$y$ commute, then they are powers of some common element: indeed, the subgroup~$K$ generated by $x$ and~$y$ is an Abelian subgroup, which is free by the Nielsen--Schreier theorem; but the only Abelian free group is~$\cZ$; so $K$ is cyclic.  Now, as mentioned above, $\lambda'(e)$ is primitive.

Now in $\lambda'$ any bigon between a minor edge of $\loopgraph_0$ and an arc of $\Sigma$ does not contain any vertex of $\loopgraph$.  Indeed, otherwise, the major edge incident to that vertex would not be cyclically reduced, a contradiction.  So we can remove any innermost bigon by an ambient isotopy. When this is not possible anymore, by the preceding claim, the minor edges intersect the arcs of $\Sigma$ with the same sequence in $\lambda$ and~$\lambda'$. Here we make use of the fact (see above) that if two loops are homotopic relatively to their basepoint (or via a free homotopy in which the basepoint does not leave its face of $\Sigma$) and reduced, then they intersect the arcs of $\Sigma$ with the same sequence.
\end{proof}

\begin{proof}[Proof of Theorem~\ref{thm:surf} in the case $b\ge1$]
We apply Lemma~\ref{lem:canon-bd} to construct in time $O((g+b)m)$ time a loop system $Y$ of $\surface$, and a drawing $\varphi : H \to Y$, homotopic to the inclusion map $H \subset \surface$, such that $\varphi$ maps each edge of $H$ to a walk of length $O(g+b)$ in $Y$. 

We apply the algorithm of the proof of Theorem~\ref{thm:main-theorem} in Section~\ref{sec:end-of-proof} to the drawing $\varphi \circ \delta$, but where Lemma~\ref{lem:contract-boundary} replaces Proposition~\ref{prop:contraction-algorithm}, and where Lemma~\ref{lem:untangling-loops-boundary} and Theorem~\ref{thm:toth-et-al} replace Proposition~\ref{prop:untangling-a-loop-graph-general}. (Here we straighten the drawing of the sparse loop graph by first cyclically reducing the major edges and then reducing the minor edges, similar to Lemma~\ref{lem:straightening}.) That is, in time $O((g+b)^2 n \log((g+b)n))$, we determine whether $\varphi \circ \delta$ (and thus $\delta$) can be untangled, and if it is the case, in $O((g+b)n^2)$ additional time we compute a weak embedding $\delta'' : G \to Y$, homotopic to $\varphi \circ \delta$ (and thus to $\delta$), that maps each edge of $G$ to a walk of length $O((g+b)n)$ in $Y$.

Now Lemma~\ref{lem:canon-bd} pushes the 1-skeleton of $Y$ arbitrarily close to the 1-skeleton of $H$ by isotopy, and provides the resulting weak embedding $\psi : Y \to H$, such that $\psi$ maps each edge of $Y$ to a walk of length $O(m)$ in $H$. We compute $\delta' := \psi \circ \delta''$ in $O((g+b)mn^2)$ time, which is as desired.
\end{proof}

\section{Untangling piecewise linear drawings in the punctured plane}\label{sec:plane}

In this section we consider the model for the input where the surface is the punctured plane and the drawing is a piecewise linear map. More precisely we prove Theorem~\ref{thm:plane}.  (The \emph{size} of a piecewise linear drawing~$\delta$ of a graph~$\graph$  is the size of~$\graph$ plus the total number of segments comprising the edges of the images of~$\graph$ under~$\drawing$.)

\begin{figure}
    \centering
    \includegraphics[width=0.8\linewidth]{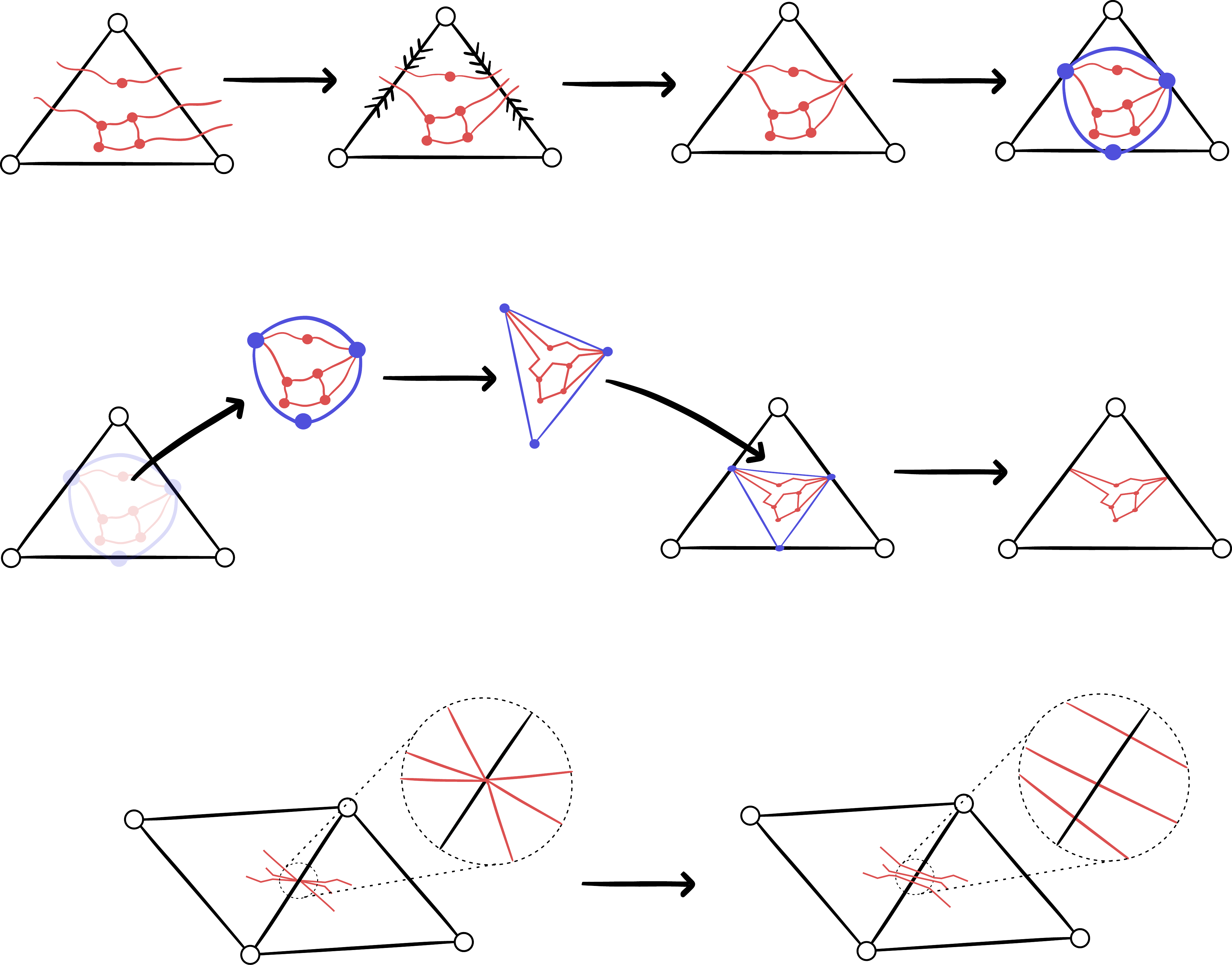}
    \caption{In the proof of Theorem~\ref{thm:plane}, the topological embedding $\drawing'$ (in red) is transformed into a piecewise-linear embedding. (Top) In a triangle $\Delta$ of $T$ (in black), the crossings of $\drawing'$ with the sides of $\Delta$ are packed at the midpoints of those edges, then a cycle with three vertices (in blue) is attached to the image $X$ of $\delta'$ in $\Delta$. (Middle) $X \cup C$ is replaced by a piecewise-linear embedding ambient isotopic to it, then $C$ is deleted. (Bottom) The resulting piecewise linear drawing $\psi'$ of $G$ is made an embedding again by unpacking its crossings the edges of $T$.}
    \label{fig:final embedding}
\end{figure}

\begin{proof}[Proof of Theorem~\ref{thm:plane}]
First we reduce to the combinatorial map model in a way very similar to Cabello, Liu, Mantler, and Snoeyink~\cite[Lemma~12]{clms-thpp-04}, in the same spirit as Colin de Verdière and de Mesmay~\cite[Section~5.2]{cm-tgis-14}. More precisely we do the following. We fix a closed box around each point of $P$ that does not intersect the image of $\drawing$ nor the other boxes, and we denote by $B_P$ the resulting collection of boxes. We also fix a bounding box $B$ that contains the image of $\drawing$ and all the boxes in $B_P$ in its interior. Then we construct in time almost linear in $p$ a cellular decomposition $T$ of $B \setminus B_P$, of size $O(p)$, whose edges are rectilinear. Without loss of generality $T$ and $\drawing$ are in general position. Finally, we consider the dual graph $H$ of $T$, and the topological drawing $f : G \to H$ that (1) maps each vertex $v$ of $G$ to the vertex of $H$ dual to the face of $T$ containing $\drawing(v)$, and (2) maps each edge $e$ of $G$ to the walk in $H$ encoding the crossings between the path $\drawing(e)$ and the edges of $T$. The drawing $f$ has size $O(\lambda n)$, where $\lambda$ is the maximum number of times a segment of the drawing intersects edges of $T$. Also, $f$ is computed in $O(np\log p)$ time by computing and sorting the intersection points of~$T$ with each line segment of $\drawing$, and is homotopic to $\drawing$ in $B \setminus P$. 

First assume that we are just interested in determining if there exists an embedding homotopic to $f$. In this case we can ensure $\lambda = O(\sqrt p)$ by constructing $T$ in such a way that each line in the plane crosses at most $O(\sqrt p)$ edges of $T$~\cite[Lemma~11]{clms-thpp-04}. We apply Theorem~\ref{thm:surf} to~$f$ and~$H$ to determine in $O(p^{5/2} n \log(pn))$ time if there exists an embedding homotopic to $f$, equivalently to $\drawing$, in $B \setminus P$. 

Now assume that there exists an embedding homotopic to $f$, and that we want to compute a piecewise linear embedding homotopic to $f$. In this case we prefer to construct $T$ so that it is a triangulation (for example, but not necessarily, a Delaunay triangulation), without caring about $\lambda$. We still have $\lambda = O(p)$. Let us now explain how to compute the desired embedding in additional $O(p^5 n^2)$ time. Theorem~\ref{thm:surf} provides in $O(p^5n^2)$ time a weak embedding $f': G \to H$, of size $O(p^5n^2)$. The algorithm of Akitaya, Fulek, and Tóth (Theorem~\ref{thm:toth-et-al}) provides in $O(p^5n^2 \log(pn))$ time an embedding $\drawing'$ approximating $f'$ in the patch system of $H$. Here the patch system of $H$ corresponds to the triangulation $T$, and the overlay between $T$ and the embedding $\drawing'$ is given \emph{topologically} by a an embedded graph of size $O(p^5 n^2)$ in which $T$ and $G$ are embedded. 

We transform the topological embedding $\drawing'$ into a piecewise-linear embedding in three steps. See Figure~\ref{fig:final embedding}. As a first step we do the following for each edge $e$ of $T$. We consider the crossings between $e$ and $\drawing'(G)$, and we modify $\drawing'$ by sliding those crossings along $e$ to pack all of them at the midpoint of $e$. We do so for every edge $e$ of $T$, and we consider the resulting drawing $\psi$ of $G$. The crossings between $\psi(G)$ and the edges of $T$ are now packed at the midpoints of those edges, but $\psi$ is still an embedding everywhere else. 

As a second step, we do the following in each triangle $\Delta$ of $T$. We consider the part $X$ of $\psi(G)$ that lies inside $\Delta$. We attach an outer-cycle $C$ to $X$, whose vertices are the midpoints of the edges of $\Delta$; this is the blue cycle in Figure~\ref{fig:final embedding}. If $X$ does not intersect every edge of $\Delta$ some of the vertices of $C$ may not be attached to $X$: this is fine, and is for example the case of the bottom vertex of the blue cycle in Figure~\ref{fig:final embedding}. We now construct an embedding isotopic to $X \cup C$ (in particular, $C$ remains the outer-cycle of the embedding). Ideally, we would like such an embedding in which every edge is a linear segment, a Fáry embedding. But this is not always possible since $X \cup C$ may have loops and parallel edges. This is easily solved: we insert two vertices in each edge of $X$, consider the resulting graph $X'$, observe that $X' \cup C$ has no loops nor parallel edges, and compute a F\'ary embedding isotopic to $X' \cup C$ instead, using classical algorithms~\cite{schnyder1990embedding,de1990draw,de1988small}, in linear time. Equivalently, this is a piecewise-linear embedding isotopic to $X \cup C$ in which every edge of $C$ is a linear segment and every edge of $X$ is a path of three linear segments. Up to applying an affine transformation to the embedding, we may assume without loss of generality that the three vertices of $C$ are embedded at the corresponding midpoints of the edges of $\Delta$. In $\psi$, we replace $X$ by its piecewise-linear embedding, and we forget $C$. We do so for every triangle $\Delta$ of $T$, and we consider the resulting piecewise-linear drawing $\psi'$ of $G$. 

As a third and final step, we make $\psi'$ an embedding by unpacking the crossings of $\psi'$ with each edge of $T$, while keeping $\psi'$ piecewise linear, as in Figure~\ref{fig:final embedding} (Bottom).
\end{proof}

\paragraph*{Acknowledgments.}  We would like to thank Arnaud de Mesmay for many useful discussions, and the anonymous referees whose numerous comments improved the presentation of this paper.

\bibliographystyle{plainurl}
\bibliography{bib}

\appendix

\section{Limit points in the universal cover}\label{A:limit}

In this section, we present the (classical) details omitted from Section~\ref{sec:limit-points}.  A good overview is provided by Farb and Margalit~\cite[Chapter~1]{fm-pmcg-12}, but we need a few more properties, so we provide proofs.

\subsection{Compactification of the hyperbolic plane and fixed points of translations}\label{sec:compactification}

Let $\cH$ be the hyperbolic plane, corresponding to the open unit disk in the Poincaré model.  One can compactify~$\cH$ by considering the set~$\partial\cH$ of ``points at infinity'', corresponding to the unit circle (in the Poincaré model) with its usual topology.  Equivalently~\cite[Chapter~1]{fm-pmcg-12}, the points in~$\partial\cH$ are the equivalence classes of unit speed geodesic rays, where two rays are equivalent if they stay at bounded distance from each other; the union $\bar\cH$ of $\cH$ and $\partial\cH$ is topologized via the basis containing the open sets of $\cH$ plus one open set $U_P$ for each open half-plane $P$ of $\cH$, where $U_P \cap \cH = P$ and $U_P \cap \partial \cH$ contains the equivalence class $\ell$ of unit speed geodesic rays if all the rays in $\ell$ eventually end up in $P$.

Isometries of~$\cH$ extend naturally to~$\bar\cH$, and the \emphdef{hyperbolic translations} are those with exactly two fixed points on~$\partial \cH$. In particular the identity is not considered a hyperbolic translation here. Any hyperbolic translation~$f$ admits a unique geodesic line~$A$, its \emphdef{axis}, such that $f(A) = A$ and such that $f$ is a real translation on $A$. A hyperbolic translation is uniquely determined by its axis and by the image of a point on this axis. Iterating any point of $\cH$ under $f$ makes it converge to one of the two fixed points of $f$ in $\partial \cH$, the \emphdef{fixed point at $+\infty$} of $f$, while iterating under~$f^{-1}$ makes it converge to the \emphdef{fixed point at~$-\infty$} of $f$. See~\cite[p.~13-14]{katok} for more details.

The following lemma is standard and results from simple computations in the Poincaré model of the hyperbolic plane, so we omit the proof.
\begin{lemma}\label{lem:geodesics}
In the hyperbolic plane $\cH$ let $L \subset \cH$ be a geodesic line and $a : \cR \to \cH$ be a unit speed geodesic ray. The distance between $a(t)$ and $L$ either tends to $+\infty$ as $t \rightarrow +\infty$ or it tends to zero. In the latter case, there is a unit speed parameterization $c : \cR \to L$ such that $c(t)$ and $a(t)$ remain at bounded distance over $t \geq 0$.
\end{lemma}

We will need the following lemma in the next section.
\begin{lemma}\label{lem:isometries-aux}
Fix any $x \in \cH$. Two hyperbolic translations $f, g : \cH \to \cH$ have the same fixed point at $+\infty$ if and only if they satisfy the following for some $D > 0$: there exist arbitrarily large values of $i, j \geq 0$ for which $f^i(x)$ and $g^j(x)$ are at distance less than $D$.
\end{lemma}
\begin{proof}
Let $\ell \in \partial \cH$ and $\ell' \in \partial \cH$ be the fixed point at $+\infty$ of respectively~$f$~and~$g$.

First assume the existence of some $D > 0$ such that there exist arbitrarily large values of $i, j \geq 0$ for which $f^i(x)$ and $g^j(x)$ are at distance less than $D$. To prove $\ell = \ell'$ we consider any open set $O$ of $\bar \cH$ that contains $\ell$ and we claim the existence of $i_0 \geq 0$ such that for every $i \geq i_0$ the ball of radius $D$ centered at $f^i(x)$ is contained in $O$. This claim, combined with our assumption, implies that there exist arbitrarily large values of $j \geq 0$ for which $g^j(x) \in O$. Since $g^j(x)$ tends to $\ell'$ as $x$ goes to $+\infty$, we have $\ell = \ell'$.

To prove this claim, and without loss of generality, we assume that $O$ belongs to the basis described above to define the topology of $\bar \cH$. Then, and since $\ell \in O$, there is an open half plane $P$ of $\cH$ such that $O = U_P$. Parameterize the axis of $f$ by some unit speed geodesic $a : \cR \to \cH$. By definition  (up to reversing $a$) there is $\alpha > 0$ such that $f \circ a(t) = a(t + \alpha)$ on every $t \in \cR$. Since $f$ is an isometry the distance between $f^i(x)$ and $a(i \alpha)$ remains constant over $i \in \cZ$. By construction $a$ eventually ends up in $P$. Let $L$ be the geodesic line that bounds $P$ in $\cH$. If the distance between $a(t)$ and $L$ goes to infinity as $t \in \cR$ goes to $+\infty$, then the claim is proved. Otherwise this distance goes to zero by Lemma~\ref{lem:geodesics} and there is some unit speed parameterization $c : \cR \to \mathbb{H}$ of $L$ such that $c(t)$ and $a(t)$ remain at bounded distance over $t \geq 0$. Thus $c(t)$ tends to $\ell$ as $t \in \cR$ goes to $+\infty$, contradicting the fact that $c(t) \notin P$ for every $t \in \cR$.

Conversely, assume $\ell = \ell'$. Consider unit speed parameterizations $a : \cR \to \cH$ and $b : \cR \to \cH$ of their respective axes such that $a(t)$ and $b(t)$ tend to $\ell$ as $t$ goes to $+\infty$. By definition of the limit points, there exists $D > 0$ such that for every $t \geq 0$ the distance between $a(t)$ and $b(t)$ is less than $D$. Let $\alpha > 0$ and $\beta > 0$ be the translation lengths of respectively $f$ and $g$ on their axes. There exists $D' > 0$ for which there exist $i, j \geq 0$ arbitrarily large such that $|\alpha i- \beta j| < D'$ and thus such that $a(\alpha i)$ and $b(\beta j)$ are at distance less than $D+D'$. Moreover the distance between $a(\alpha i) = f^i(a(0))$ and $f^i(x)$ does not depend on $i $ since $f$ is an isometry. The same holds for the distance between $b(\beta j)$ and $g^j(x)$.
\end{proof}

\subsection{Limit points of lifts of closed curves}\label{sec:homotopical-invariant}

Every orientable surface $\surface$ of genus at least two without boundary is homeomorphic to the quotient of the hyperbolic plane $\cH$ by the action of some (actually, many) group $\Gamma$ of isometries of $\cH$. The elements of~$\Gamma$ other than the identity are hyperbolic translations~\cite[p.~22]{fm-pmcg-12}. The action is free, in the sense that if $f \in \Gamma$ satisfies $f(x) = x$ on some $x \in \cH$, then $f$ is the identity. The action is also properly discontinuous in the sense that every $x \in \cH$ admits a neighborhood whose intersection with the $\Gamma$-orbit of $x$ is $\{x\}$. The surface $\surface$ then admits a unique hyperbolic metric for which the quotient map $\cH \to \surface$ is a local isometry. Also, the hyperbolic plane $\cH$ is a universal covering space of $\surface$, where the quotient map $\cH \to \surface$ is the covering map.

In this appendix, when we write a surface $\surface$ as a quotient $\cH / \Gamma$, we refer to the construction presented in the previous paragraph. We shall prove the four lemmas of Section~\ref{sec:limit-points}. First we need two preliminary lemmas.

\begin{lemma}\label{lem:lift-and-isometries}
Let $\surface = \cH / \Gamma$. Consider a lift $\tilde c : \cR \to \cH$ of a non-contractible closed curve on $\surface$. There is $f \in \Gamma \setminus \{1\}$ such that $\tilde c(t+1) = f(\tilde c(t))$ on every $t \in \cR$. Moreover $\lim_{+\infty} \tilde c$ and $\lim_{-\infty} \tilde c$ exist and are the fixed points of $f$ in $\partial \cH$, at respectively $+\infty$ and $-\infty$.
\end{lemma}
\begin{proof}
For every $t \in \cR$, there exists some $f_t \in \Gamma$ such that $f_t(\tilde c(t)) = \tilde c (t+1)$. Moreover every such $f_t$ is not the identity as $\tilde c(t) \neq \tilde c(t+1)$ since $\tilde c$ is a lift of a non-contractible closed curve. We claim that $f_t$ does not depend on $t$. This claim concludes the proof. We prove the claim by contradiction so assume the existence of some $t \in \cR$ fixed and of some $t' \in \cR$ arbitrarily close to $t$ such that $f_t \neq f_{t'}$. By choosing $t'$ close enough to $t$ we make the distance between $f_{t'}^{-1} \circ f_t(\tilde c(t))$ and $\tilde c(t)$ go to zero, contradicting the fact that $\Gamma$ acts properly discontinuously on $\cH$.
\end{proof}

\begin{lemma}\label{lem:limit-points-and-isometries}
  Let $\surface = \cH / \Gamma$. Assume that $f,g\in \Gamma \setminus\{1\}$ have the same fixed point at $+\infty$. There are $h \in \Gamma \setminus \{1\}$ and $n, m \geq 1$ such that $f = h^n$ and $g = h^m$.  In particular, they have the same fixed point at~$-\infty$.
\end{lemma}

\begin{proof}
Consider some arbitrary fixed $x \in \cH$. We claim the existence of $a, b \geq 1$ such that $f^a(x) = g^b(x)$. Indeed by Lemma~\ref{lem:isometries-aux} there is $D > 0$ that satisfies the following. There are $i, j \geq 0$ arbitrarily large such that $f^i(x)$ and $g^j(x)$ are at distance less than $D$. For every such $i, j$ the point $f^{-i} \circ g^j(x)$ belongs to the closed ball $B_x$ of radius $D$ centered at $x$. Since $\Gamma$ acts properly discontinuously on $\cH$ the $\Gamma$-orbit $\Gamma \cdot x$ of $x$ intersects $B_x$ in finitely many points. Indeed every such point $y \in B_x \cap \Gamma \cdot x$ admits a neighborhood whose intersection with $\Gamma\cdot x$ is $\{y\}$, and finitely many such neighborhoods suffice to cover the compact ball $B_x$. In particular there exist $i, j \geq 0$, $i' > i$ and $j' > j$ such that $f^{-i'} \circ g^{j'}(x) = f^{-i} \circ g^j(x)$. Letting $a := i'-i \geq 1$ and $b := j'-j \geq 1$ proves the claim.

Our claim implies that $f$ and $g$ have the same fixed points both at $+\infty$ and $-\infty$ and thus the same axis, say $A$, oriented from $-\infty$ to $+\infty$. Without loss of generality we assume that $x$ was chosen so that $x \in A$. Recall that $f$ and $g$ are real translations on $A$ and let $\alpha > 0$ and $\beta > 0$ be their respective periods. 

We consider the hyperbolic translation $h$ whose oriented axis is $A$ and whose period $\gamma$ is defined as follows. We proved $a \alpha = b \beta$ since $f^a(x) = g^b(x)$ and since $f^a(x)$ and $g^b(x)$ are the translations of $x$ along $A$ by a distance of respectively $a \alpha$ and $b \beta$. There are $n, m \geq 1$ relatively prime such that $na = mb$. By Bézout's theorem, there are $u,v \in \cZ$ such that $un + vm = 1$. We let $\gamma = u \alpha + v \beta$.

We claim that $h = f^u \circ g^v$, and thus $h \in \Gamma$. To prove this claim observe that the point $h(x)$ is the translation of the point $x$ by a distance $\gamma$ along the oriented axis $A$ of $h$. The point $g^v(x)$ is the translation of $x$ by a distance $v \beta$ along the same oriented axis $A$, and the point $f^u(g^v(x))$ is the translation of $g^v(x)$ by a distance of $u \alpha$ along $A$. Thus $h(x) = f^u \circ g^v(x)$. That proves the claim since hyperbolic translations are defined by their axis, here $A$, and by the image of any point on this axis, here $x$.

In the same way we have $h^n = f$ and $h^m = g$ since these hyperbolic translations have the same oriented axis by construction, and since their periods satisfy $n\gamma = \alpha$ and $m\gamma = \beta$, as seen by a straightforward computation. 
\end{proof}

We can now provide the omitted proofs of the lemmas in Section~\ref{sec:limit-points}, restated for convenience. From now on we fix a surface $\surface = \cH/\Gamma$. The universal cover $\tilde \surface$ of $\surface$ is thus identified with $\cH$.
\LemExistDistinctLimitPoints*
\begin{proof}[Proof of Lemma~\ref{lem:exist-distinct-limit-points}]
The result is given by Lemma~\ref{lem:lift-and-isometries}. 
\end{proof}

\LemLiftLimitPoints*
\begin{proof}[Proof of Lemma~\ref{lem:lift-limit-points}]
Consider the hyperbolic translations, say $f$ and $g$, given by Lemma~\ref{lem:lift-and-isometries} for $\tilde{c}$ and $\tilde{d}$ respectively. Since the homotopy $\tilde{c} \simeq \tilde{d}$ lifts the homotopy $c \simeq d$, the distance between the points $\tilde{c}(k) = f^k(\tilde c(0))$ and $\tilde{d}(k) = g^k(\tilde d(0))$ does not depend on $k \in \cZ$. By Lemma~\ref{lem:isometries-aux} $f$ and $g$ have the same fixed points.
\end{proof}

\LemSeparateLimitPoints*
\begin{proof}[Proof of Lemma~\ref{lem:separate-limit-points}]
By Lemma~\ref{lem:exist-distinct-limit-points}, the two limit points of~$\tilde c$ are distinct, and similarly for~$\tilde d$.  Assume, for a contradiction, that $\tilde{c}$ and $\tilde{d}$ have the same limit point at $+\infty$ (up to reversing $c$ or~$d$).  There are two cases.

First assume that $\tilde{c}$ and $\tilde{d}$ intersect exactly once in some point $x \in \cH$. Without loss of generality assume $\tilde{c}(0) = \tilde{d}(0) = x$. Consider the hyperbolic translations, say $f$ and $g$, given by Lemma~\ref{lem:lift-and-isometries} for $\tilde{c}$ and $\tilde{d}$ respectively. Then $f$ and $g$ have the same fixed point at $+\infty$. Thus by Lemma~\ref{lem:limit-points-and-isometries} there exist $i, j \geq 1$ such that $f^i = g^j$. Then $\tilde{c}(i) = f^i(x) = g^j(x) = \tilde{d}(j)$ and this point is distinct from $x$ since, for example, $\tilde c$ is non-contractible. This is a contradiction.

Now assume that $\tilde{c}$ and $\tilde{d}$ are disjoint lifts of the same curve $c$ in $\surface$. There is a geodesic closed curve $\alpha$ homotopic to $c$~\cite[Proposition~1.3]{fm-pmcg-12}. Lift the homotopy $c \simeq \alpha$ to homotopies $\tilde{c} \simeq \tilde{\alpha}$ and $\tilde{d} \simeq \tilde{\beta}$ for some lifts $\tilde{\alpha}$ and $\tilde{\beta}$ of $\alpha$. By the preceding lemmas, $\tilde{\alpha}$ and $\tilde{\beta}$ have the same limit points. Thus $\tilde{\alpha}$ and $\tilde{\beta}$ have the same image (a geodesic line), so they are equal up to homeomorphism $\cR \to \cR$. By the uniqueness part of the lifting property, $\tilde{c}$ and $\tilde{d}$ are then equal up to homeomorphism  $\cR \to \cR$. We proved that $\tilde{c}$ and $\tilde{d}$ intersect, a contradiction.
\end{proof}

\LemSameLimit*
\begin{proof}[Proof of Lemma~\ref{L:samelimit}]
Let $\tilde c$ and $\tilde d$ be lifts of respectively $c$ and $d$ and assume that $\tilde c$ and $\tilde d$ have the same limit points.  In a first step, if $\tilde c$ and~$\tilde d$ are disjoint, we apply a homotopy of~$c$ to make $\tilde c$ intersect~$\tilde d$; this does not change the limit points of~$\tilde c$ by Lemma~\ref{lem:lift-limit-points}.

So we can assume without loss of generality that $\tilde c$ and $\tilde d$ intersect, and thus also that $\tilde c(0) = \tilde d(0) = x \in \cH$. Consider the hyperbolic translations, say $f$ and $g$, given by Lemma~\ref{lem:lift-and-isometries} for $\tilde{c}$ and $\tilde{d}$ respectively. By Lemma~\ref{lem:limit-points-and-isometries} there exist $n, m \geq 1$ and $h \in \Gamma \setminus \{1\}$ such that $f = h^n$ and $g = h^m$. Thus $f(x) = h^n(x)$ and $g(x) = h^m(x)$. Let $\ell$ be a loop on $\surface$ that lifts to a path from $x$ to $h(x)$. Based at zero, $c$ and $d$ are homotopic as loops to respectively the $n$th power and the $m$th power of the loop $\ell$.
\end{proof}

\end{document}